\documentclass[a4paper,11pt]{article}

\usepackage{amssymb}
\usepackage{amsmath}
\usepackage{amsthm}
\usepackage{geometry}
\usepackage[onehalfspacing]{setspace}
\usepackage[UKenglish]{babel}
\usepackage{indentfirst}
\usepackage[utf8]{inputenc}
\usepackage{csquotes}
\usepackage[backend=biber,style=authoryear-comp,giveninits=true,natbib,maxcitenames=3,maxbibnames=50,uniquelist=false,uniquename=init,isbn=false,doi=false,useprefix=true,backref=true]{biblatex}
\DeclareNameAlias{sortname}{family-given}
\DeclareFieldFormat[article]{title}{#1}
\DeclareFieldFormat[article]{journaltitle}{#1}
\DeclareFieldFormat[book]{title}{#1}
\DeclareFieldFormat[article]{pages}{#1}
\renewbibmacro{in:}{%
 \ifentrytype{article}{}{\printtext{\bibstring{in}\intitlepunct}}}
\usepackage{xpatch}
\xpatchbibmacro{journal+issuetitle}{%
  \setunit*{\addspace}%
  \iffieldundef{series}}
  {%
  \setunit*{\addcomma\space}%
  \iffieldundef{series}}{}{}
  
\DefineBibliographyStrings{UKenglish}{
  backrefpage  = {$\rightarrow$},
  backrefpages = {$\rightarrow$},
}
\renewbibmacro*{pageref}{%
  \iflistundef{pageref}
    {}
    {\printtext[brackets]{%
       \ifnumgreater{\value{pageref}}{1}
         {\bibstring{backrefpages}\ppspace}
         {\bibstring{backrefpage}\ppspace}%
       \printlist[pageref][-\value{listtotal}]{pageref}}}}

\renewbibmacro*{volume+number+eid}{%
  \printfield{volume}
  \printfield[parens]{number}%
  \setunit{\addcomma\space}%
  \printfield{eid}}
  
\DeclareSortingNamekeyTemplate{
  \keypart{
    \namepart{family}
  }
  \keypart{
    \namepart{prefix}
  }
  \keypart{
    \namepart{given}
  }
  \keypart{
    \namepart{suffix}
  }
}
\renewbibmacro{begentry}{\midsentence}
\DeclareFieldFormat[misc]{title}{#1}
\DeclareFieldFormat[report]{title}{#1}

\usepackage{booktabs}
\usepackage[dvipsnames]{xcolor}
\usepackage{hyperref}
\hypersetup{
colorlinks=true,
linkcolor=Blue,
urlcolor=Blue,
citecolor=Blue,
filecolor=Blue,
bookmarksopen=false,
frenchlinks=false,
pdftitle={Equitable Auctions},
}
\usepackage{bm}
\usepackage{adjustbox}
\usepackage{subcaption}
\usepackage{accents}
\usepackage{bbold}
\usepackage{mathtools}
\usepackage[capitalise]{cleveref}
\usepackage{xspace}
\usepackage{placeins}

\theoremstyle{definition}
\newtheorem{theorem}{Theorem}
\newtheorem{corollary}{Corollary}
\newtheorem{definition}{Definition}
\newtheorem{example}{Example}
\newtheorem{assumption}{Assumption}

\newtheorem{lemma}{Lemma}
\newtheorem{proposition}{Proposition}

\newcommand\blfootnote[1]{
  \begingroup
  \renewcommand\thefootnote{}\footnote{#1}
  \addtocounter{footnote}{-1}
  \endgroup
}

\newcommand{\vbar}{\widebar{s}}
\newcommand{\bidders}{[n]}
\newcommand{\supply}{[k]}
\newcommand{\bid}[1][\delta]{\beta^{#1}}
\newcommand{\bbeta}{{\bm{\beta}}}
\newcommand{\ssb}{{\bm{s}}}
\newcommand{\bbi}{b_i}
\newcommand{\bbn}{b}
\newcommand{\bb}{\bm{b}}
\newcommand{\shati}{{\hat{s}_i}}
\newcommand{\xb}{\bm{x}}
\newcommand{\xp}{\bm{p}}

\newcommand{\support}{[0,\vbar)}
\newcommand{\osupport}{(0,\vbar)}
\newcommand*\diff{\mathop{}\!\mathrm{d}}
\newcommand{\gnk}{g_{k}^{n-1}}
\newcommand{\Gnk}{G_{k}^{n-1}}
\newcommand{\smi}{\ssb_{-i}}
\newcommand{\smihat}{{\hat{\ssb}}_{-i}}
\newcommand{\smj}{\ssb_{-j}}

\newcommand{\E}{\mathbb{E}}
\newcommand{\Prob}{\mathbb{P}}
\newcommand{\var}{\text{\normalfont Var}}
\newcommand{\wev}{\text{\normalfont WEV}\xspace}
\newcommand{\ev}{\text{\normalfont EV}}
\newcommand{\wvar}{\text{\normalfont WV}}
\newcommand{\PD}{pairwise differences\xspace}
\newcommand{\wevtext}{\text{\normalfont winners' empirical variance}\xspace}

\newcommand{\pd}[2]{\frac{\partial#1}{\partial#2}}
\newcommand{\pdbdeltas}{\pd{\beta^{\delta}}{s}}
\newcommand{\ind}[1]{{\mathbb 1}\left\{#1\right\}}
\newcommand{\kth}[1]{Y_{#1}}

\newenvironment{proof2}[1][Proof]
{\par\begingroup
   \pushQED{\qed}%
   \normalfont\noindent\textbf{#1.}\ }
  {\popQED\par\endgroup}

\newenvironment{proof2lower}[1][Proof]
  {\par\begingroup
   \pushQED{\qed}%
   \normalfont\noindent\textbf{#1.}\ }
  {\popQED\par\endgroup}

\newcommand{\figurenote}[1]{%
    \begin{quote}\footnotesize\emph{Note:} #1
    \vspace{-1.5em}
    \end{quote}
}

\makeatletter
\let\save@mathaccent\mathaccent
\newcommand*\if@single[3]{%
  \setbox0\hbox{${\mathaccent"0362{#1}}^H$}%
  \setbox2\hbox{${\mathaccent"0362{\kern0pt#1}}^H$}%
  \ifdim\ht0=\ht2 #3\else #2\fi
  }
\newcommand*\rel@kern[1]{\kern#1\dimexpr\macc@kerna}
\newcommand*\widebar[1]{\@ifnextchar^{{\wide@bar{#1}{0}}}{\wide@bar{#1}{1}}}
\newcommand*\wide@bar[2]{\if@single{#1}{\wide@bar@{#1}{#2}{1}}{\wide@bar@{#1}{#2}{2}}}
\newcommand*\wide@bar@[3]{%
  \begingroup
  \def\mathaccent##1##2{%
    \let\mathaccent\save@mathaccent
    \if#32 \let\macc@nucleus\first@char \fi
    \setbox\z@\hbox{$\macc@style{\macc@nucleus}_{}$}%
    \setbox\tw@\hbox{$\macc@style{\macc@nucleus}{}_{}$}%
    \dimen@\wd\tw@
    \advance\dimen@-\wd\z@
    \divide\dimen@ 3
    \@tempdima\wd\tw@
    \advance\@tempdima-\scriptspace
    \divide\@tempdima 10
    \advance\dimen@-\@tempdima
    \ifdim\dimen@>\z@ \dimen@0pt\fi
    \rel@kern{0.6}\kern-\dimen@
    \if#31
      \overline{\rel@kern{-0.6}\kern\dimen@\macc@nucleus\rel@kern{0.4}\kern\dimen@}%
      \advance\dimen@0.4\dimexpr\macc@kerna
      \let\final@kern#2%
      \ifdim\dimen@<\z@ \let\final@kern1\fi
      \if\final@kern1 \kern-\dimen@\fi
    \else
      \overline{\rel@kern{-0.6}\kern\dimen@#1}%
    \fi
  }%
  \macc@depth\@ne
  \let\math@bgroup\@empty \let\math@egroup\macc@set@skewchar
  \mathsurround\z@ \frozen@everymath{\mathgroup\macc@group\relax}%
  \macc@set@skewchar\relax
  \let\mathaccentV\macc@nested@a
  \if#31
    \macc@nested@a\relax111{#1}%
  \else
    \def\gobble@till@marker##1\endmarker{}%
    \futurelet\first@char\gobble@till@marker#1\endmarker
    \ifcat\noexpand\first@char A\else
      \def\first@char{}%
    \fi
    \macc@nested@a\relax111{\first@char}%
  \fi
  \endgroup
}
\makeatother

\usepackage{tikz}

\newcommand{\dotline}{%
  \tikz[baseline=-0.5ex]{%
    \fill (0,0) circle (0.5ex);
    \draw[line width=0.6pt] (-1.2ex,0) -- (1.2ex,0);
  }%
}

\crefname{supplappendix}
  {Supplementary Appendix}
  {Supplementary Appendices}
\Crefname{supplappendix}
  {Supplementary Appendix}
  {Supplementary Appendices}

\begin{document}

\title{\textsc{Equitable Auctions}}

\author{Simon Finster%
        \footnote{Johannes Kepler University Linz and Austrian Institute of Economic Research, simon.finster@jku.at} \and
        Patrick Loiseau\footnote{Inria/FairPlay, patrick.loiseau@inria.fr}
        \and
        Simon Mauras%
        \footnote{Inria/FairPlay, simon.mauras@inria.fr} \and 
        Mathieu Molina%
        \footnote{Tel Aviv University, mathieum@tauex.tau.ac.il} \and 
        Bary Pradelski%
        \footnote{\mbox{Maison Française d'Oxford, CNRS and Department of Economics, University of Oxford, bary.pradelski@cnrs.fr}}}
\date{September 19, 2026\\
\vspace{0.8em}
First version: March 12, 2024
}

\renewcommand*{\thefootnote}{\arabic{footnote}}

\maketitle

\begin{abstract}
We study how auction design shapes surplus division among bidders. We introduce a parsimonious ex post equity notion and examine standard multi-unit auctions. The equity-optimal pricing depends sharply on the private-common value composition. Among mixed-price auctions, at one extreme, uniform pricing is always optimal for pure common values and, at the other, pay-as-bid pricing is preferred for pure private values under standard regularity assumptions. For any fixed monotone allocation rule, we characterize a BNIC interim-IR implementation that equalizes surplus ex post across all bidders and is revenue-maximizing among those that equalize surplus ex ante; efficient allocation permits parity among winners.

\end{abstract}
\blfootnote{\emph{Acknowledgments}: We are grateful for feedback and comments from Pierre Boyer, Julien Combe, P{\'e}ter Es{\H{o}}, Pär Holmberg, Atulya Jain, Simon Jantschgi, Bernhard Kasberger, Paul Klemperer, Yves Le Yaouanq, Laurent Linnemer, Bing Liu, Simon Loertscher, Matías Nú\~nez, Sander Onderstal, Michael Ostrovsky, Ludvig Sinander, Alex Teytelboym, Kyle Woodward, and audiences at the Simons Laufer Mathematical Sciences Institute (Berkeley), CIRM (Marseille), CREST (Paris), NASMES 2024 (Nashville), CMID 2024 (Budapest), EARIE 2024 (Amsterdam), and Match-up 2024 (Oxford). \\
\indent This material is based upon work supported by the National Science Foundation under Grant No.~DMS-1928930 and by the Alfred P. Sloan Foundation under grant G-2021-16778, while Simon Finster and Bary Pradelski were in residence at the Simons Laufer Mathematical Sciences Institute (formerly MSRI) in Berkeley, California, during the Fall 2023 semester.
This work benefited from Hi! PARIS and State funding managed by the French National Research Agency (ANR) under the France 2030 program, reference ANR-23-IACL-000; it was partially supported by the French National Research Agency (ANR) through grants ANR-20-CE23-0007 and ANR-23-CE23-0002 and through the PEPR IA FOUNDRY project (ANR-23-PEIA-0003). 
Simon Mauras received funding from the European Research Council (ERC) under the European Union's Horizon 2020 research and innovation program (grant agreement No. 866132), as a postdoctoral fellow at Tel Aviv University.
Mathieu Molina received funding from the European Research Council (ERC) under the European Union's Horizon Europe program (grant agreement No. 101170373), as a postdoctoral fellow at Tel Aviv University.

An extended abstract of an earlier version of this article appeared in the \emph{Proceedings of the 26th ACM Conference on Economics and Computation} (EC'25) \citep{Finster-et-al-2025}.}

\newpage

\section{Introduction}\label{sec:Introduction}

Auctions allocate goods, but they also allocate rents. Ex ante revenue equivalence is silent on many market-design questions: two auction formats may be equivalent in allocation and expected revenue, yet imply very different ex post distributions of bidder surplus. Once firms borrow, invest, and exit in realized states, the distribution of auction rents becomes a first-order design object. However, while efficiency and revenue considerations have guided auction design, the ex post division of surplus between bidders has received little attention.
In this article, we show that auctions can be made more equitable by design. We propose a simple, intuitive measure of ex post equity based on pairwise surplus comparisons. We identify which auction formats used in practice distribute rents most evenly, and characterize payment rules that equalize ex post surplus.

Why does the distribution of realized surplus matter? A bidder's realized auction surplus can affect whether they invest, survive adverse shocks, or compete in future auctions. An auctioned asset can be resold or generate future cash flows through production, and it may serve as collateral; the auction payment reduces the winner's available capital. Thus, ex post surplus measures the net addition to a bidder's available capital after the auction. When capital markets are imperfect, a more equal division of a given amount of surplus among firms with the same initial resources leaves capital less concentrated. If returns on subsequent investment are diminishing, a more equal distribution can raise aggregate investment.\footnote{We assume that returns to investment are unrelated to bidders' values for the auctioned asset.} If firm survival probabilities are increasing and concave in available capital, a more equal distribution can also reduce the risk of exit and sustain competition in downstream product markets. These channels motivate our equity criterion, but we abstract from their effect on equilibrium bidding. The standard auctions we consider implement the same efficient allocation and generate the same expected seller revenue and, by payment equivalence, the same interim expected surplus for each bidder at every signal; the payment rule, however, determines the ex post division of surplus among winners.

A growing literature studies redistribution through market design, typically by changing allocations, relaxing market clearing, or constraining bids \citep{Dworczak-Kominers-Akbarpour-2021,Akbarpour-Dworczak-Kominers-2024,Jeong-Pycia-2026}. Moreover, distributional objectives already enter auction policy through set-asides, subsidies, contracting goals, or bidding credits.\footnote{These policy instruments support designated firms in procurement and spectrum auctions \citep{Pai-Vohra-2012,Athey-Coey-Levin-2013,SBA-2026,FCC-2016}. The European Commission emphasizes the treatment of small bidders under uniform pricing in EU ETS auctions \citep{European-Commission-2026}.} We study a different instrument: even holding the allocation rule fixed, the payment rule determines how infra-marginal rents are shared. In many multi-unit environments, efficiency and expected revenue do not pin down how those rents are distributed across winners. We ask whether the payment rule can redistribute realized surplus while holding allocation fixed.

To study this question, we introduce a notion of surplus equity based on pairwise ex post differences in realized surpluses. Our measure nests prominent inequality statistics such as the Gini index and the empirical variance. We apply it to standard multi-unit auctions with unit-demand bidders whose values range from purely private to purely common. Standard auctions provide a useful benchmark for our question because the efficient allocation and each bidder's interim expected surplus are the same across formats. Differences in realized surplus therefore isolate how the payment rule distributes rents ex post. To build intuition, consider the two canonical pricing rules used in practice -- the uniform-price and the pay-as-bid auction.\footnote{In treasury auctions both designs are common \citep{OECD-2021}, while most electricity markets feature the uniform-price rule. Further examples include auctions for emission certificates and online advertising.} In the uniform-price auction, all winners pay the first rejected bid; in the unique symmetric equilibrium, each bidder bids the expected value conditional on tying with the marginal competitor, and with pure private values this strategy is weakly dominant. In the pay-as-bid auction, each winning bidder pays its bid, which induces bid shading in the unique Bayesian-Nash equilibrium.

In a pure common-value environment, uniform pricing is more surplus-equitable than pay-as-bid pricing because all winners obtain identical ex post surplus under a common clearing price, whereas under pay-as-bid high-signal bidders suffer from a winner's curse. \cref{fig:common} illustrates this: in the uniform-price auction all winners' realized surpluses are equal, whereas under pay-as-bid pricing higher-signal bidders receive lower realized payoffs. This ranking holds for every realization of signals. Surpluses for small signals are similar across the two formats, but for larger signals pay-as-bid pricing penalizes realized surplus.

\begin{figure}[htbp!]
    \centering
    \hspace*{-1em}
    \begin{subfigure}{0.48\textwidth}
        \centering
        \captionsetup{skip=5pt,margin={20pt,0pt}}
        \includegraphics[width=1.05\linewidth]{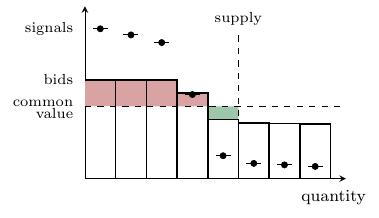}
        \caption{Pay-as-bid auction}
        \label{fig:PAB-common}
    \end{subfigure}
    \hfill
    \begin{subfigure}{0.48\textwidth}
        \centering
        \captionsetup{skip=5pt,margin={30pt,0pt}}
        \includegraphics[width=1.05\linewidth]{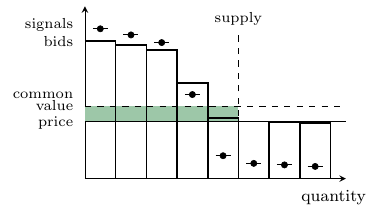}
        \caption{Uniform-price auction}
        \label{fig:UNI-common}
    \end{subfigure}
    \captionsetup{skip=6pt}
    \caption{Realized surpluses for \textbf{a pure common value}}
    \vspace{-0.75em}
    \figurenote{
    Illustrative example with $n=8$ bidders, $k=5$ items, and $\operatorname{Beta}(1.9,1)$-distributed signals. The $x$-axis shows the quantity and the $y$-axis shows the signals (\protect\dotline), bids (bars), price, and common value. Green shading shows positive realized surplus; red shading shows negative realized surplus. In the uniform-price auction all winners' surpluses are equal, while in the pay-as-bid auction higher-signal bidders receive a lower (in fact, negative) payoff than lower-signal bidders---a winner's curse effect. Revenue equivalence only holds ex ante, not ex post. The signals, bids, and surpluses are detailed in \cref{example:illustration-pab-vs-uniform-common} in \cref{app:further-results-examples-figures}.
    }
    \label{fig:common}
\end{figure}

With pure private values, the basic logic runs in the other direction. Uniform pricing magnifies rent differences because infra-marginal bidders with high signals pay the same price as marginal winners, whereas pay-as-bid pricing tends to compress realized surpluses by making stronger bidders pay more. \cref{fig:private} illustrates this intuition. We show, however, that the intuition does not hold in general: without additional restrictions on the signal distribution, uniform pricing can still outperform pay-as-bid pricing on our equity criterion. For signal distributions with log-concave densities -- for example, the uniform and exponential distributions -- pay-as-bid pricing distributes ex post surplus more equitably than uniform pricing.

\begin{figure}[htbp!]
    \hspace*{0em}
    \begin{subfigure}{0.48\textwidth}
        \centering
        \captionsetup{skip=5pt,margin={0pt,0pt}}
        \includegraphics[width=\linewidth]{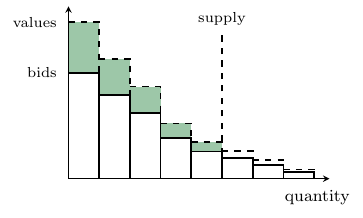}
        
        \caption{Pay-as-bid auction}
        \label{fig:PAB-private}
    \end{subfigure}
    \hfill
    \begin{subfigure}{0.48\textwidth}
        \centering
        \captionsetup{skip=5pt,margin={10pt,0pt}}
        \includegraphics[width=1\linewidth]{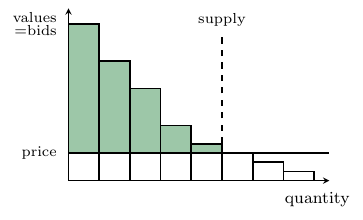}
        \caption{Uniform-price auction}
        \label{fig:UNI-private}
    \end{subfigure}
    \captionsetup{skip=6pt}
    \caption{Realized surpluses for \textbf{pure private values}}
    \vspace{-0.75em}
    \figurenote{Illustrative example with $n=8$ bidders, $k=5$ items, and uniformly distributed signals. The $x$-axis shows the quantity and the $y$-axis shows the values (dashed bars), bids (solid bars), and price. Green shading shows the bidders' positive realized surpluses. With private values the uniform-price auction is less equitable than the pay-as-bid auction. The signals and corresponding bids and surpluses are given in \cref{example:illustration-pab-vs-uniform-private} in \cref{app:further-results-examples-figures}.
    }
    \label{fig:private}
\end{figure}

These comparisons show that surplus equity depends jointly on the pricing rule and on the private-common value composition of the environment. Beyond uniform and pay-as-bid auctions, we therefore study mixed-price auctions in which prices are a combination of uniform and pay-as-bid pricing.

We then turn to more general mechanism design and ask whether ex post surplus parity can be achieved in an incentive-compatible way. Equity among winners is relevant when auction rents affect the investment and subsequent competition of the firms that obtain the goods. In winners-pay auctions, losers receive no item and make no payment, so their surplus remains zero. A designer may also care about surplus equity among all participants when they have equal claims to the benefits of a shared resource. For example, a fishing cooperative may assign quota blocks to the members who can use them most profitably while seeking equal realized surplus for every member. We therefore also consider payment rules that allow transfers to losing bidders. We identify mechanisms that equalize ex post surplus among winners under both efficient allocation and randomized allocation-fair assignments. Moreover, we characterize payment rules that equalize surplus across all bidders for arbitrary monotone allocation rules. We next provide a detailed discussion of our contribution.

\subsection{Our Contribution}

We consider standard and winners-pay multi-unit auctions for indivisible, identical goods, as well as general direct mechanisms. Each bidder demands one unit and receives a private, independent signal about their value from a commonly known distribution. Values range continuously from purely private to purely common.\footnote{Private-common values can, for example, represent resale opportunities; see \citet{Bikhchandani-Riley-1991,Klemperer-1998,Bulow-Klemperer-2002,Goeree-Offerman-2003}. In markets for pollution permits, a certificate is valuable for production (private value) but also has a resale value common to the market.}

Our first conceptual contribution is the definition of a parsimonious equity order based on pairwise differences in realized surplus: mechanism A dominates mechanism B if all absolute pairwise differences in ex post realized surpluses are weakly smaller in A than in B, and strictly smaller for at least one pair (\cref{def:pairwise-differences-winners}). This order is consistent with any anonymous equity metric that aggregates pairwise differences through an increasing function. Examples include the empirical variance, the Gini index, or top-bottom decile comparisons.

Our main contribution is organized around two related questions. Within standard auction formats, how does the pricing rule shape surplus equity? Beyond standard formats, which equity targets can be implemented while preserving incentive compatibility, and when can this also be done without sacrificing efficiency or expected revenue?

First, we study the classical formats -- the uniform-price and pay-as-bid auctions, which are commonly used in practice, and mixed auctions. Under the classical assumptions that deliver efficiency and expected-revenue equivalence within this class, these formats differ sharply in how they distribute ex post surplus. We show that uniform pricing is equity-optimal if and only if values are purely common (\cref{theorem:pure-common}). With pure private values, pay-as-bid need not be equity-optimal among mixed auctions (\cref{prop:counter-example}). With log-concave signal distributions, however, we obtain robust lower bounds on the discriminatory share required for equity dominance: for common-value proportion $c$, the auction with $\delta = 1-c$ dominates every mixed auction with less price discrimination, and pay-as-bid is equity-optimal among mixed auctions when $c = 0$ (\cref{theorem:bound-on-optimal-delta,cor:when-PAB-is-optimal}). Economically, the results imply that the private-common value composition determines which pricing rules can be defended on equity grounds.

Second, we turn from standard pricing rules to direct mechanism design. Having identified what canonical formats can and cannot achieve, we enlarge the design space and ask which equity targets are implementable through transfers. We show that ex post surplus parity among winners is compatible with efficient allocation and with randomized allocation-fair rules. In particular, there exists a Bayesian-Nash incentive-compatible (BNIC) direct mechanism that assigns the goods efficiently and equalizes realized surplus among winners without sacrificing expected revenue (\cref{theorem:surplus-equity}). For any monotone allocation rule, there is a BNIC payment rule that equalizes surplus ex post across all bidders and maximizes revenue among all BNIC, interim individually rational (IR) implementations of the same allocation rule that equalize surplus ex ante (\cref{theorem:surplus-equitable-monotone-alloc}). These results separate the two design questions of who gets an object and how the mechanism shares the rents conditional on that allocation.

We also provide several extensions and a detailed discussion of robustness, and we study surplus equity numerically through the winners' empirical variance, which aggregates ex post pairwise differences. For a range of signal distributions and common-value proportions, we compute the mixed-pricing rule that minimizes this measure and find a unique minimizer in the cases we study (\cref{sec:numerical-experiments}). \cref{sec:extensions} covers the distinction between equity and risk preferences, the difficulty of extending our rankings beyond log-concave signals, multi-unit demand, and how some of our pay-as-bid and uniform-price results carry over to reserve prices.

\subsection{Related Literature}

We relate and contribute to several strands of work:  redistribution through market design, fairness in auction allocation and payments, uniform, pay-as-bid, and mixed-price auctions, variance in surplus and revenue, and public finance and optimal taxation.

Our paper is closest in spirit to the recent literature on redistribution through market design. In that literature, the central issue is often an efficiency-equity trade-off. In a large buyer-seller market for a single object, \citet{Dworczak-Kominers-Akbarpour-2021} characterize the optimal efficiency-equity frontier when agents differ in marginal utility of money, and \citet{Akbarpour-Dworczak-Kominers-2024} characterize when non-market mechanisms dominate market-clearing prices for the allocation of a fixed supply; these mechanisms forgo efficiency to improve equity. In a single-object setting with i.i.d. private values, \citet{Jeong-Pycia-2026} show that first-price auctions with suitably chosen bid spaces and, where needed, equal subsidies to all bidders can optimize a class of objectives including revenue, welfare, and inequality in interim allocations, payments, or utilities. We compare realized surplus among simultaneous winners in multi-unit auctions with private-common values. Our standard mixed auctions have unrestricted bid spaces and the same allocation and interim payments. We also characterize direct mechanisms that equalize ex post surplus.

Fairness concerns in auctions have often been studied through favoritism, set-asides, or subsidies. In a model with explicit target-group favoritism, \citet{Pai-Vohra-2012} show that the optimal mechanism is a flat or type-dependent subsidy depending on the nature of the favoritism constraint. \citet*{Athey-Coey-Levin-2013} reach a related empirical conclusion for U.S. Forest Service timber auctions: set-asides for small bidders reduce efficiency and revenue, whereas subsidies can increase revenue and small-bidder profits with little cost. We complement this literature by identifying a further, often less distortionary, lever. Conditional on the allocation rule, payment design itself can improve the distribution of realized surplus across bidders.

The fair-allocation literature has introduced many notions of fairness, including envy-freeness, equal division, and no domination \citep[see][]{Thomson-2011}. Our notion is orthogonal to envy-freeness, which requires that no agent prefers another agent's allocation (object and price). In our setting, uniform pricing is the unique mixed-pricing rule that yields envy-freeness among winners, and with pure private values it yields envy-freeness among all participants.\footnote{With a positive proportion of common value, depending on the realization of signals, winners may experience the winners' curse and prefer not to have won an item.} However, envy-freeness does not take into account that bidders may draw different signals and hence different ex post values. Subscribing to the notion of envy-freeness, we would accept that realized utilities may be very unequal, depending on the realization of the private value component. We instead take realized utilities as the relevant baseline for fairness, closer to equal-treatment-of-equals \citep[cf., e.g.,][]{Thomson-2011} from an ex post perspective.

Our ex post surplus-equalizing payment rules also relate to mechanism-design work on ex post implementations of truthful mechanisms. \citet{DAspremont-Gerard-Varet-1979} show that ex post budget balance can be achieved in Bayesian incentive-compatible direct mechanisms under independent private values. \citet{Eso-Futo-1999} show that, in a setting with one object and independent private values, any incentive-compatible selling mechanism can be transformed into one with the same allocation and expected revenue and deterministic revenue for the seller. \citet{Le-2026} shows that the AGV mechanism minimizes total transfer variance among ex post budget-balanced mechanisms with fixed allocation and interim transfers. We equalize realized surplus across agents within each signal profile. Our results contribute to this strand of the literature by characterizing payment rules that equalize ex post surplus while respecting incentive constraints.

Mixed-price auctions, originating from \citet{Wang-Zender-2002} and \citet{Viswanathan-Wang-2002}, also appear in models with divisible or stochastic supply \citep{Marszalec-Teytelboym-Laksa-2020,Ruddell-Philpott-Downward-2017,Woodward-2021} and have been considered as pricing rules for capacity procurement in electricity markets \citep{Holmberg-Tangeras-2023}. \citet{Woodward-2021} studies the same price mixture for divisible goods with a common signal, and \citet{Pycia-Woodward-2026} analyze pay-as-bid equilibrium in that setting. We adapt the pricing rule to indivisible goods, unit demand, and independent signals, and rank formats by differences in realized surplus across winners.

The analysis of surplus variance within a bidder in a single-item auction goes back to \citet{Vickrey-1961}, who showed that the first-price rule yields lower surplus variance than the second-price rule for uniform signals. More recently, \citet{Bojko-et-al-2026} show that the first-price rule minimizes the variance of surplus -- and more generally any convex risk measure -- within a given bidder. \citet{Pekec-Smilgins-2026} show that, with risk-neutral unit-demand bidders and independent private values, the discriminatory auction has lower ex ante revenue variance than the uniform auction. We study ex post surplus differences between bidders within an auction. This question is inherently linked to the multi-unit context, since it concerns how realized rents are shared among multiple winners. At the same time, the empirical variance, which belongs to our class of equity measures, combines within-bidder variation and between-bidder correlation of surpluses (see \cref{lem:connecting-variances}).

Finally, our paper also relates to public finance and optimal taxation, where redistribution has long been debated. Full redistribution leading to equal surpluses is the solution to the classical utilitarian social welfare problem with concave homogeneous utility functions \citep{Edgeworth-1897,Piketty-Saez-2013}. A more general approach based on generalized social welfare weights was proposed by \citet{Saez-Stantcheva-2016}. Although one can aggregate our pairwise surplus differences with welfare weights, the economic object remains different: we study inequality in realized utilities that arises from heterogeneous signals taken as primitive, whereas optimal taxation models typically treat incomes as endogenous outcomes of agents' labor or investment choices. There is also a direct connection between mixed auctions and ex post taxation. In uniform auctions, the seller observes apparent surplus, the gap between winning bids and the clearing price \citep{Ruddell-Philpott-Downward-2017}. Taxing apparent surplus at a proportional rate is equivalent to charging mixed prices. Such taxation of infra-marginal rents has been discussed in electricity-market design, in the context of recovering infrastructure investment costs in power grids \citep{Ruddell-Philpott-Downward-2017,New-Zealand-Electricity-Authority-2014}.\\

The remainder of the article is organized as follows. In \cref{sec:setup}, we introduce the model and derive equilibrium bidding strategies in mixed auctions. We introduce and illustrate our notion of surplus equity in \cref{sec:surplus-equity}. In \cref{sec:mixed-auctions}, we develop results for uniform, pay-as-bid, and mixed auctions. In \cref{sec:surplus-equitable-mechanism}, we study direct mechanisms that achieve ex post surplus equity among winners and characterize, for arbitrary monotone allocation rules, a payment rule that equalizes surplus among all bidders. \cref{sec:conclusion} concludes the main text. The supplementary appendix contains extensions,
robustness discussion, and additional results,
proofs, examples, and figures.

\section{Market Model and Equilibrium}\label{sec:setup}

A finite number of bidders $\bidders\vcentcolon=\{1,\dots,n\}$ compete for a fixed supply of items $\supply\vcentcolon=\{1,\dots,k\}$, where $2 \leq k < n$. Each bidder demands one item. Bidder $i$ receives a private signal $s_i$, which is drawn independently from an absolutely continuous probability distribution $F$ with density $f$. We assume that $f$ has support on the positive reals with upper limit $\vbar$ (which may be infinite). 
Call $(0,\vbar)$ the open support of $F$, i.e.,~all signals $s$ so that $0<F(s)<1$, and assume that $f>0$ over $\osupport$. Finally, assume that the signals have a finite second moment $\E[s^2]< \infty$.

For $\ssb \vcentcolon=\{s_i\}_{i\in\bidders}$, a collection of i.i.d. signals, we denote by $\kth{m}(\ssb)$ the $m$-th highest value of the collection $\ssb$ (with $n$ entries). For example, $\kth{1}(\ssb)$ is the maximum and $\kth{n}(\ssb)$ is the minimum of the collection $\ssb$. Note that $\kth{m}(\ssb)$ is a random variable, and we denote its probability distribution $G_{m}^{n}(y)$ with corresponding density $g_{m}^{n}(y)$. $G_{m}^{n}$ is given by $G_{m}^{n}(y) = \sum_{j=0}^{m-1} \binom{n}{j} F(y)^{n-j} (1-F(y))^{j}$, where each summand is the probability that $j$ signals are above $y$. An expression for $g_{m}^{n}(y)$ is given in \cref{proof:lem:intG-log-concave}.

The value of bidder $i$ for an item is given by the valuation function $v(s_i, \smi)$, where $\smi := (s_j)_{j\neq i}$, and $v(s_i, \smi)$ is symmetric in other bidders' signals $\smi$.
\begin{assumption}[\textbf{Values}]\label{ass:bidder-values}
    Values $v(s_i,\smi)$ are given by $v(s_i, \smi) = (1-c) s_i + \frac{c}{n} \sum_{j\in \bidders}s_j$, where $c \in [0,1]$ is the \emph{proportion of the common value}.
\end{assumption}
Our model interpolates between a common value and private values, where the proportion of the common value $c$ encodes to what extent the value of any given bidder is influenced by the signals of the other bidders. In particular, $c=1$ defines a pure common value and $c=0$ pure private values.\footnote{In an alternative model, the common value might be distributed according to some prior distribution, and the bidders' private signals are drawn conditional on the realization of this common value. The alternative model has identical qualitative characteristics \citep{Goeree-Offerman-2003}: (i) the items are valued equally by all bidders in the common value component, and (ii) the winner's curse is present, i.e.,~winning an item is ``bad news'', in that the winner's expectation of the item's value was likely too optimistic.} We note that the value function satisfies the \emph{single-crossing} condition, since for all $i,j \in \bidders$, $i \neq j$, and for all $\ssb$, ${\partial v(s_i,\smi)}/{\partial s_i} \geq {\partial v(s_j,\ssb_{-j})}/{\partial s_i}$. A change in signal $s_i$ affects bidder $i$'s valuation at least as much as any bidder $j$'s valuation. As a result, bidders' ex post values are ordered the same way as their signals.

Auction mechanisms are represented by allocations and transfers $\{\pi_i(s_i,\smi),$ $ p_i(s_i,\smi)\}_{i\in\bidders}$, where $\pi_i(s_i,\smi)$ is defined as the probability that an item is allocated to bidder $i$ when the reported signals are $s_i$ and $\smi$, and $p_i(s_i,\smi)$ is the corresponding price charged to the bidder, which is symmetric in its second argument, i.e.,~$p_i(s_i,s_{-i}) = p_i(s_i,s_{-i}')$ for all permutations $s_{-i}'$ of $s_{-i}$.
An auction is \emph{standard} if the $k$ highest bids win the items, and \emph{winners-pay} if only the winners pay and their payment is no more than their bid. A standard auction, in any symmetric and increasing equilibrium and if values satisfy the single-crossing condition, is \emph{efficient} \citep{Krishna-2009}, i.e.,~the bidders with the $k$ highest values are assigned the items.

We consider two classes of mechanisms. First, we consider $k$-unit \emph{mixed auctions} (formal definition below), which are standard and winners-pay. Each bidder submits a bid $\bbi$, resulting in the vector of submitted bids $\bb$. Restricting our attention to symmetric and monotonically increasing bidding strategies $\bbi = \bbn(s_i)$, we write allocations and prices as functions of signals only. Second, we consider truthful, direct mechanisms in which bidders submit their signal; here we relax the standard and winners-pay restrictions, allowing arbitrary monotone allocation rules and, for surplus equity among all bidders, allowing transfers to losers.

In a standard auction, the allocation of bidder $i$ is given by $\pi_i(\shati,\smihat) = \ind{\shati>\kth{k}(\smihat)}$
when signals $\shati$ and $\smihat$ are reported.
A bidder's utility (or surplus) when reporting signal $\shati$, and the remaining $n-1$ bidders reporting signals $\smi$, is given by
\begin{equation*}
    u_i(s_i,\shati,\smi,\smihat) = \ind{\shati>\kth{k}(\smihat)} \cdot v(s_i, \smi) - p_i(\shati,\smihat).
\end{equation*}
Given a signal $s_i$, recall that we denote by $\kth{k}(\smi)$ the $k$th highest among the signals $\smi$. $\kth{k}(\smi)$ has probability distribution $\Gnk$ and density $\gnk$.

Furthermore, we denote equilibrium bidding strategies by $(\beta_i)_{i\in[n]}=\bbeta$.
\begin{definition}[Mixed auctions]\label{def:mixed-auctions}
    In the $k$-unit \emph{$\delta$-mixed auction}, parameterized by a given $\delta\in[0,1]$, each bidder $i$ pays 
    \[p_i(\bb) = \left[\delta\cdot \bbi + (1-\delta)\cdot  Y_{k+1}(\bb)\right]\cdot \ind{b_i > Y_{k+1}(\bb)}.\]
\end{definition}
For $\delta=0$, this resolves to \emph{first-rejected-bid (FRB) uniform pricing} or for short \emph{uniform pricing}, where each winning bidder $i$ pays the $(k+1)$-th highest bid $Y_{k+1}(\bb)$.\footnote{In our market, all standard uniform-price auctions generate the same equity outcomes according to our metric (introduced below) and our results thus hold for all of them.} For $\delta=1$, this resolves to \emph{pay-as-bid pricing}, where each winning bidder $i$ pays their bid $\bbi$. Finally, if $\delta\in(0,1)$, we say that the auction and the pricing are \emph{strictly mixed}.

\subsection{Interim Values, Payments, and Utilities}

For all $(x,y)\in \support^2$, we define the expected value given $s_i = x$ and $\kth{k}(\smi) = y$ as follows:
\begin{equation*}
    \widetilde V(x,y):=\E_\ssb[v(s_i,\smi) \mid s_i = x, \kth{k}(\smi) = y]
\end{equation*}
The expected value is taken over $n-2$ signals not including the bidder's own signal and the $k$th highest among their $n-1$ opponents. Observe that because $v(s_i, \smi)$ is continuous and non-decreasing, $\widetilde V(x,y)$ is continuous in $x$ and $y$, strictly increasing in $x$, and non-decreasing in $y$. We define $V(y) := \widetilde V(y,y)$, the expectation of the value of an item conditional on the bidder being exactly at the winning cutoff against the relevant competing signal, the $k$th highest among its competitors.
Furthermore, we introduce \emph{interim payments} $P_i(s_i) = \E_{\smi}[p_i(s_i,\smi)]$ and \emph{interim utilities} $U_i(s_i,\shati) = \E_{\smi}[u_i(s_i,\shati,\smi,\smi)]$ with $U_i(x)$ being the shorthand of $U_i(x,x)$.

Interim incentive compatibility (IC) requires $U_i(s_i,s_i) \geq U_i(s_i,\shati)$ for all $s_i,\shati$, and interim individual rationality (IR) demands $U_i(s_i,s_i) \geq 0$ for all $s_i$. It follows from a standard application of the envelope theorem \citep{Milgrom-Segal-2002} that the auctions we consider result in the same expected payment for each bidder. The interim utility is given by $U_i(s_i,\shati) = \E_{y=\kth{k}(\smi)}[\ind{\shati\geq y}\widetilde V(s_i,y)] - P_i(\shati)$. Letting $G:=\Gnk$ and $g:=\gnk$, it follows from standard arguments that
\[
P_i(s_i) = \int_0^{s_i} \widetilde V(s_i,y) g(y) \diff y  - U_i(0) - \int_0^{s_i} U_i'(x) \diff x.
\] 
From the winners-pay requirement and continuity of the signals follows $U_i(0) = P_i(0) = 0$. Revenue equivalence extends to any standard auctions with independent signals in which winners pay, including mixed auctions, as the allocation rule is identical.

\subsection{Equilibrium Bidding}\label{sec:equilibrium-bidding}

We derive the unique symmetric and increasing Bayesian-Nash equilibrium in $\delta$-mixed auctions. The proofs adapt standard equilibrium arguments to our setting. We therefore provide them in \cref{suppl-apx:additional-proofs}, together with the derivative calculations used below. This equilibrium is the center of our analysis of surplus equity in mixed auctions in \cref{sec:mixed-auctions}. Throughout, $\bid$ denotes the equilibrium bid in the $\delta$-mixed auction; we write $\bid[U]$ for the uniform-price case ($\delta=0$) and use explicit superscripts when the underlying $\delta$ matters.
\begin{proposition}[e.g.,~\citealt{Krishna-2009}]\label{prop:equilibrium-bid-FRB}{\normalfont[Proof in \cref{proof:prop:equilibrium-bid-FRB}]}
    The unique symmetric equilibrium bidding strategy in the uniform-price auction, i.e.,~the case $\delta=0$, is given by $\bid[U](s) := \widetilde V(s,s) = \E[v(s_i,\smi) \mid s_i = s, \kth{k}(\smi) = s]$.
\end{proposition}
At $s=0$, the bidding strategy is defined by continuous extension through $\bid(0)=V(0)$.
Note that equilibrium strategies are weakly dominant strategies with pure private values. The bid function below, $\bid$, converges to $\beta^U$ as $\delta \rightarrow 0$.

\begin{proposition}\label{prop:equilibrium-bid-delta}{\normalfont[Proof in \cref{proof:prop:equilibrium-bid-delta}]}
    The unique symmetric equilibrium bidding strategy in the $\delta$-mixed auction, for $\delta\in(0,1]$, is given by
    \begin{align}\label{equ:alternative-equilibrium-delta}
        \bid(s)= V(s) -\frac{\int_{0}^{s}  V'(y)\Gnk(y)^{\frac{1}{\delta}}\diff y}{\Gnk(s)^{\frac{1}{\delta}}}.
    \end{align}
\end{proposition} 

\section{Surplus Equity}\label{sec:surplus-equity}

We call our proposed equity notion \emph{dominance in \PD}, for short \emph{pairwise differences}. Our results hold for the family of equity measures that are defined by any increasing function of pairwise differences in ex post surpluses. We call the collection $\{u_i\}_{i\in\bidders}$ of ex post utilities (realizations of random variables) an outcome. Given two mechanisms, $A$ and $B$, we call the collections
$\{u^A_i\}_{i\in\bidders}$ and $\{u^B_i\}_{i\in\bidders}$ joint realizations, when they arise from the same signal draws.

\begin{definition}[Dominance in pairwise differences]\label{def:pairwise-differences-winners}
Given two mechanisms, $A$ and $B$, with a shared set of winners $W$, \emph{mechanism $A$ dominates mechanism $B$ in \PD~}if for almost all joint realizations 
    $\{u^A_i\}_{i\in\bidders}$ and $\{u^B_i\}_{i\in\bidders}$
    \[|u^A_i - u^A_j| \leq |u^B_i - u^B_j| \quad \text{ for all } i,j\in W,
\]
and the inequality is strict for at least one pair $i,j\in W$ on a set of positive measure.

We then say that \emph{mechanism $A$ is equity-preferred to mechanism $B$}. A mechanism is \emph{equity-optimal} in a class if it is equity-preferred to every other mechanism in that class.
\end{definition}

Pairwise differences induce a partial dominance ranking over outcomes and therefore a dominant mechanism may not exist in a given class of mechanisms. 
The all-bidders version of dominance in \PD~corresponds to taking $W=\bidders$, comparing pairwise utility differences across winners and losers. We use it as a substantive equity criterion in its own right.

Several prominent equity measures belong to the family of \PD, that is, measures that are constructed by \emph{any increasing function} of \PD \citep[cf., e.g.,][]{Patty-Penn-2019},\footnote{\citet{Feldman-Kirman-1974} measure envy by aggregating positive pairwise differences between an agent's hypothetical utility from another agent's bundle and that agent's realized utility.} e.g.,~the empirical variance, the Gini index \citep{Gini-1921},\footnote{In our setting with uncertainty about signal realizations, one could define the expected Gini index among winners $G = \frac{1}{2k^2 \E[u_1 \mid 1,\dots, k \text{ win}]} \E[ \sum_{i=1}^{k}\sum_{j=1}^{k} |u_i - u_j|  \mid 1,\dots, k \text{ win} ]$. A small distinction is the normalization by the expected surplus.} the comparison of top and bottom deciles of realized utilities, or measures that penalize larger differences more. We focus on the \emph{winners' empirical variance (\wev)} of surplus between the winners for some of our results. The metric ranks mechanisms by expected surplus dispersion. Throughout, \PD\ remains the primitive equity criterion. We use \wev\ when a scalar, continuous comparison is helpful, but a \wev-minimizing auction need not be dominant in \PD.
\begin{definition}[Winners' empirical variance---\wev]\label{def:WEV}
Define the winners' empirical variance as
    \begin{align}
    \wev~= \E\left[\frac{1}{2k(k-1)} \sum_{i=1}^{k}\sum_{j=1}^k (u_i - u_j)^2 \middle | 1,\ldots,k \text{ win}\right].
\end{align}
\end{definition}

The winners' empirical variance links surplus variance (within bidder) and correlation of surpluses (between bidders). The definition of the empirical variance for all bidders is obtained by replacing $k$ with $n$ and omitting the conditioning on winning signals. We discuss this further in \cref{sec:risk-aversion}. \cref{example:bid-functions} in \cref{app:further-results-examples-figures} illustrates the \wevtext for different common values and pricing rules.

Finally, we discuss two standard equity axioms, anonymity and the Pigou-Dalton principle \citep[see][]{Atkinson-1970,Sen-Foster-1973}. Anonymity requires an equity measure to remain unchanged when surpluses are permuted among the individuals being compared. The winners' empirical variance satisfies this axiom. Dominance in \PD\ is unchanged when the same permutation of bidder labels is applied to both mechanisms. The Pigou-Dalton principle requires that any transfer from a wealthier agent to a poorer one must reduce inequality, provided the original welfare ranking between the two agents is maintained \citep{Pigou-1912,Dalton-1920,Moulin-2004}.
The Pigou-Dalton principle is not satisfied by \PD. A Pigou--Dalton transfer narrows the gap between the two agents involved, but it may widen one of their gaps with a third agent; the post-transfer outcome therefore need not \PD-dominate the pre-transfer outcome. However, we establish the Pigou-Dalton principle for the winners' empirical variance.%
\begin{proposition}\label{prop:WEV-satisfies-Pigou-Dalton}\hyperref[proof:prop:WEV-satisfies-Pigou-Dalton]{\normalfont[Proof in \cref{proof:prop:WEV-satisfies-Pigou-Dalton}]}
    If the winners' empirical variance is finite, it satisfies the Pigou-Dalton principle for transfers between winners.
\end{proposition}

With our equity measures at hand, we can now turn to evaluating auction formats.

\section{Uniform, Pay-as-bid, and Mixed Auctions}\label{sec:mixed-auctions}

Surplus equity in multi-unit auctions is shaped by two forces: the pricing rule and the relative importance of private and common values. Together, these elements determine the extent to which a bidder's own signal affects their surplus and how strongly competing signals do. This section develops the central argument of the paper: equity-optimal pricing depends systematically on the private-common value composition.

\subsection{Uniform Pricing and the Necessity of Mixing}\label{sec:existence-WEV-minimal-mixed}

Consider first a pure common-value environment ($c=1$). As all bidders value the item identically ex post, realized utilities among winners are equal if and only if everyone pays the same price. This results in \PD in utilities of zero.

\begin{theorem}\label{theorem:pure-common}\hyperref[proof:theorem:pure-common]{\normalfont[Proof in \cref{proof:theorem:pure-common}]}
    The uniform-price auction is equity-optimal in the class of mixed auctions if and only if the common value proportion equals one.
\end{theorem}
With pure common values, any pricing rule that charges winners the same price equalizes surplus---the uniform-price rule does exactly that. Once the private value component enters the value function with a non-zero weight, the picture is less clear: it may be pay-as-bid that is equity-optimal, or it may be some degree of price mixing; however, it cannot be uniform pricing.

Is mixed pricing strictly necessary to improve on surplus equity? The answer is straightforwardly obtained if the equity metric is continuous in the common-value component. As pairwise differences imply only a partial order, this continuity is not given. Therefore, we consider the continuous \wevtext. The next proposition shows that, near the pure common-value benchmark, a strictly mixed pricing rule can lower expected surplus dispersion relative to both uniform pricing and pay-as-bid pricing. For a designer, this means that the most equitable pricing format need not lie at either extreme.

\begin{proposition}\label{prop:mixed-minimizing-WEV}\hyperref[proof:prop:mixed-minimizing-WEV]{\normalfont[Proof in \cref{proof:prop:mixed-minimizing-WEV}]}
    For any signal distribution, there exists $c^*<1$, such that for common value proportions in the interval $(c^*,1)$, there exist mixed auctions with lower \wevtext than pay-as-bid and uniform auctions.
\end{proposition}
When values are purely private, intuitively one might think that the picture reverses: the uniform price benefits bidders with high signals disproportionately, whereas pay-as-bid pricing can compress surplus differences. A high-signal winner pays more than a low-signal winner. If this extra payment is smaller than the difference in their private signals, the high-signal winner still earns more surplus, but the surplus gap is smaller than under uniform pricing. The latter intuition, however, is nuanced and requires regularity in signal distributions. In what follows, we describe a scenario where it fails.

\subsection{Challenging the Intuition: Private Values and Uniform Pricing}\label{sec:counterexample}

We now demonstrate that even under pure private values, uniform pricing may yield a lower \wevtext\ than pay-as-bid pricing.

\begin{example}\label{example:counter}
Consider an auction with $n$ bidders and $k = n-1$ items. Each bidder $i$ has a pure private value ($c = 0$) given by its signal $s_i$. The signal is drawn as the sum of a Bernoulli random variable with parameter $1-\varepsilon$ and an independent perturbation following $\text{Beta}(1, 1/\eta)$, where $\eta>0$. This construction yields a continuous signal distribution supported on $[0,2]$, with quantile function
\begin{align*}
    \forall x\in [0,1],\qquad
    F^{-1}(x) = \ind{x \geq \varepsilon} + \gamma_\eta(x)
    \qquad\text{where}\qquad \gamma_\eta(x) = \begin{cases}
    1 - \left(1-\frac{x}{\varepsilon}\right)^\eta & \text{if }x < \varepsilon\\
    1 - \left(1-\frac{x-\varepsilon}{1-\varepsilon}\right)^\eta & \text{if }x \geq \varepsilon.\\
    \end{cases}
\end{align*}
\end{example}
The corresponding CDF and quantile function are shown in \cref{fig:counterexample-signals-quantiles} in \cref{app:further-results-examples-figures}. Intuitively, this distribution can be viewed as a smooth approximation of a Bernoulli signal with small probability mass near zero and most of the mass concentrated near one. Further derivations are provided in \cref{app:counterexample}. We state the main result next.
\begin{proposition}\label{prop:counter-example}\hyperref[proof:prop:counter-example]{\normalfont[Proof in \cref{proof:prop:counter-example}]}
    Fix $n\geq 5$ and let the values be distributed according to the quantile function $F^{-1}$ defined above with $\varepsilon=0.1/n$. Then, there exists $\eta^*>0$, such that for all $\eta\in(0,\eta^*]$ the \wevtext\ is lower under uniform pricing than under pay-as-bid pricing.
\end{proposition}

To understand the reversal of the usual intuition, consider \PD~in utility, the building block for \wev. Ex post absolute differences in utility must, with sufficiently high probability, be larger under pay-as-bid than under uniform pricing. It is straightforward to see that, for any two arbitrarily close signals, a necessary condition for this is that the slope of the equilibrium bid function under pay-as-bid pricing is greater than two. In this case, high-signal bidders shade their bids less (relative to their value) than low-signal bidders.
Consequently, the ex post surplus difference between two bidders under pay-as-bid pricing exceeds the difference in their signals---which is exactly the surplus difference under uniform pricing. 

\cref{example:counter} highlights that, without further restrictions on the signal distribution, equity-optimal pricing cannot be characterized merely by the pricing rule. In the next section, we impose additional regularity assumptions---specifically, log-concavity of the signal distribution---to derive a partial characterization of equity-optimal pricing.

\subsection{Equity-Optimal Pricing under Signal Regularity}\label{sec:bounds-on-WEV-minimal-pricing}

We show that for the class of log-concave signal distributions, simple bounds tell us which auction designs are candidates for being equity-optimal in the class of mixed auctions.
The family of log-concave distributions contains many common distributions, for example uniform, normal, exponential, logistic or Laplace distributions \citep{Bagnoli-Bergstrom-2005}.%

\begin{definition}
    A function $h:\mathbb{R}\to\mathbb{R}_{\geq0}$ is \emph{log-concave} if $\{x:h(x)>0\}$ is convex and $\log(h)$ is concave on that set. A signal distribution is log-concave if its density $f$ is log-concave.
\end{definition}

\begin{theorem}\label{theorem:bound-on-optimal-delta}\hyperref[proof:theorem:bound-on-optimal-delta]{\normalfont[Proof in \cref{proof:theorem:bound-on-optimal-delta}]}
    Suppose that the signal density $f$ is log-concave. Then, for a given private value proportion $1-c$, the mixed auction with price discrimination $\delta=1-c$ dominates in \PD~every mixed auction with less price discrimination. Moreover, every mixed auction with price discrimination $\delta \in (0,\min\{1,2(1-c)\}]$ dominates the uniform-price auction in \PD.
\end{theorem}
\cref{theorem:bound-on-optimal-delta} gives a lower bound on the amount of price discrimination needed to avoid being dominated in \PD, and a range for which mixed pricing is guaranteed to dominate uniform pricing. The \wev-minimizing pricing rule remains undetermined in general. This yields a clear design implication: given an (estimated) common-value proportion $c$, pricing rules below the diagonal $\delta=1-c$ can be ruled out. For a designer, the appeal of this bound is that, within the log-concave class, it only requires an assessment of the relative importance of common and private values. If that assessment is imprecise, the theorem should be read as a screening rule: it rules out pricing rules with too little price discrimination. We illustrate the first statement in \cref{fig:theorem_2}: all pricing rules in the red shaded area are dominated by the diagonal $1-c$. The theorem also identifies when mixed pricing is guaranteed to improve on uniform pricing. For a given common-value proportion $c$, any pricing rule in the green shaded area of \cref{fig:theorem_3} dominates uniform pricing.
\begin{figure}[htp]
    \centering
    \captionsetup{skip=6pt}
    \begin{minipage}[t]{0.45\textwidth}
        \vspace{0pt}
        \includegraphics[scale=0.35,trim={0 1.0cm 0 0},clip]{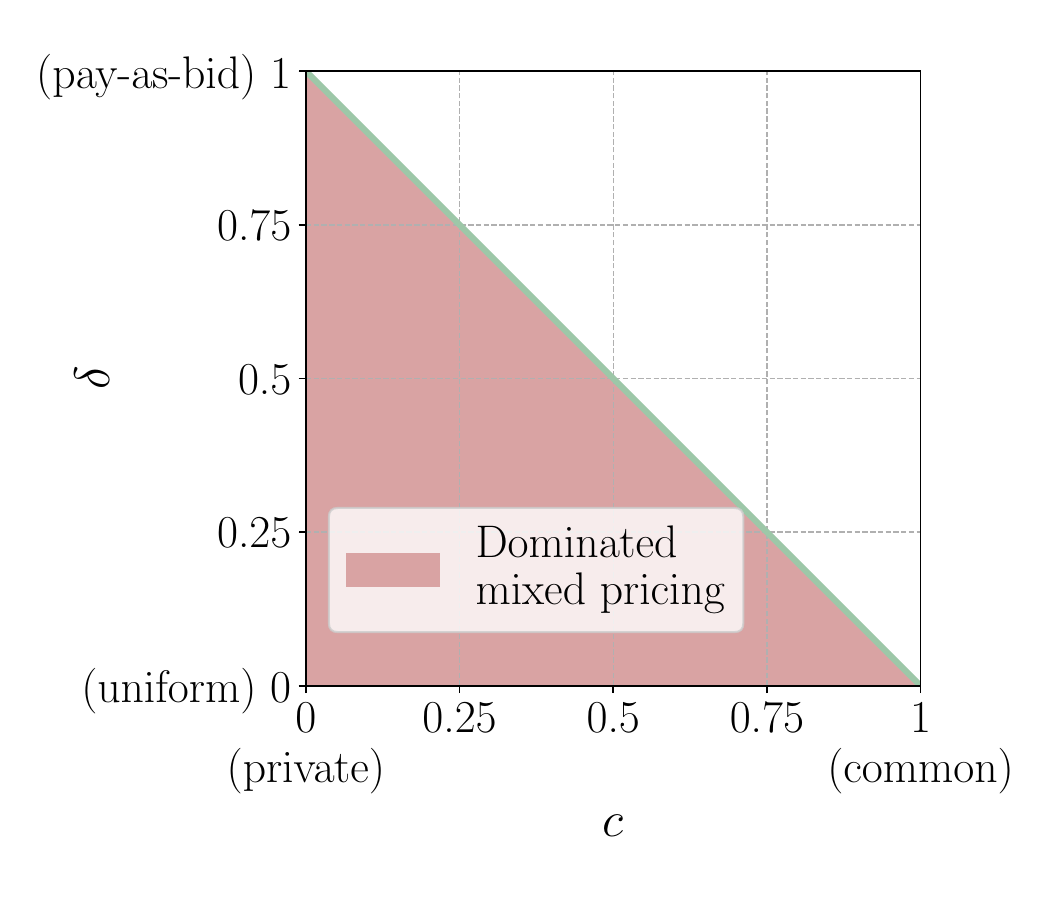}
    \caption{Dominated mixed pricing rules}
    \label{fig:theorem_2}
    \end{minipage}
    \hspace{1cm}
    \begin{minipage}[t]{0.45\textwidth}
        \vspace{0pt}
        \includegraphics[scale=0.35,trim={0 1.0cm 0 0},clip]{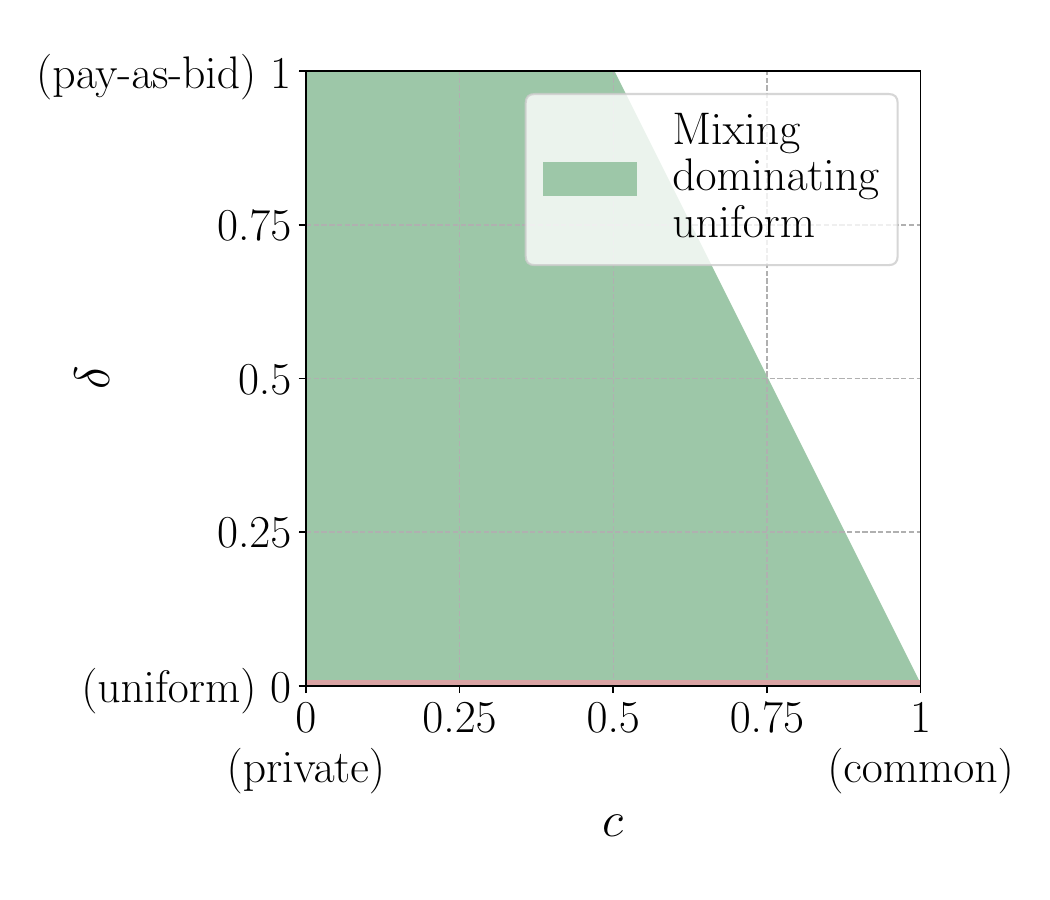}
    \caption{Mixed pricing rules that dominate uniform pricing}
    \label{fig:theorem_3}
    \end{minipage}
\end{figure}

The intuition behind \cref{theorem:bound-on-optimal-delta} relies on the notion of \emph{monotone ex post utility (MEU)}.
\begin{definition}[Monotone ex post utility among winners]\label{def:monotone-ex-post-utility}
    The ex post utility $u(\ssb)$ satisfies \emph{monotone ex post utility (MEU)} iff, for any realized signal profile $\ssb$ and any two winning bidders $i,j$, $s_i \leq s_j \Rightarrow u_i(\ssb) \leq u_j(\ssb)$.
\end{definition}
In words, MEU implies that bidders with higher signals obtain higher surplus ex post. As long as higher-signal bidders obtain a higher surplus, more equity can be achieved by taxing those bidders more than lower-signal bidders. This progressive taxation is implemented through more price discrimination because the change in the $\delta$-weighted equilibrium bid in $\delta$ is increasing in a bidder's signal. A similar intuition explains the second part of \cref{theorem:bound-on-optimal-delta}, the dominance of mixed pricing over uniform pricing, where the benefit of higher price discrimination compared to the absence of price discrimination can be realized up to a certain threshold.

The following corollary applies the first part of \cref{theorem:bound-on-optimal-delta} at $c=0$ and the second at $\delta=1$.
\begin{corollary}\label{cor:when-PAB-is-optimal} 
    Suppose that the signal density $f$ is log-concave. Then, for pure private values, pay-as-bid pricing is dominant in \PD in the class of mixed auctions, and for a common-value proportion $c \leq \frac{1}{2}$, pay-as-bid pricing dominates uniform pricing in \PD.
\end{corollary}
Our numerical experiments in the next section also demonstrate that, for $c<\frac{1}{2}$, pay-as-bid pricing in fact minimizes the \wevtext among all mixed auctions for several common distributions. 

\subsection{Simulation Results}
\label{sec:numerical-experiments}

We consider a uniform, a truncated exponential and a truncated normal distribution (both log-concave), as well as a Beta distribution with shape parameters $(0.5, 0.5)$, which is not log-concave. We illustrate the \wev-minimal pricing, a lower bound on the minimizer, and combinations of common value shares and mixed pricing for which monotone ex post utility holds (\cref{def:monotone-ex-post-utility}) for a market with $n=10$ bidders and $k=4$ items. We report the results for the uniform and Beta distributions in \cref{fig:WEV-optimal-various-distrib}.  The complete comparison is given in \cref{fig:WEV-optimal-various-distrib-full} in \cref{app:further-results-examples-figures}.\footnote{We compute the equilibrium bid functions and winners' empirical variance by numerical integration of our analytical formulae, using the simplifications in \cref{lem:value-function-bounded} and \cref{lem:alternative-variance-2} in \cref{app:further-results-examples-figures}. We express multidimensional expectations as nested one-dimensional integrals. Further computational details are given in \cref{app:further-results-examples-figures}. Our code is available on \href{https://github.com/simonfinster/equitable_auctions}{GitHub}.}

\begin{figure}[htp]
    \centering
    \begin{subfigure}[t]{0.36\textwidth}
        \subcaptionbox*{Uniform$[0,1]$}{%
        \includegraphics[width=\linewidth,trim={0 0.3cm 1.75cm 0},clip]{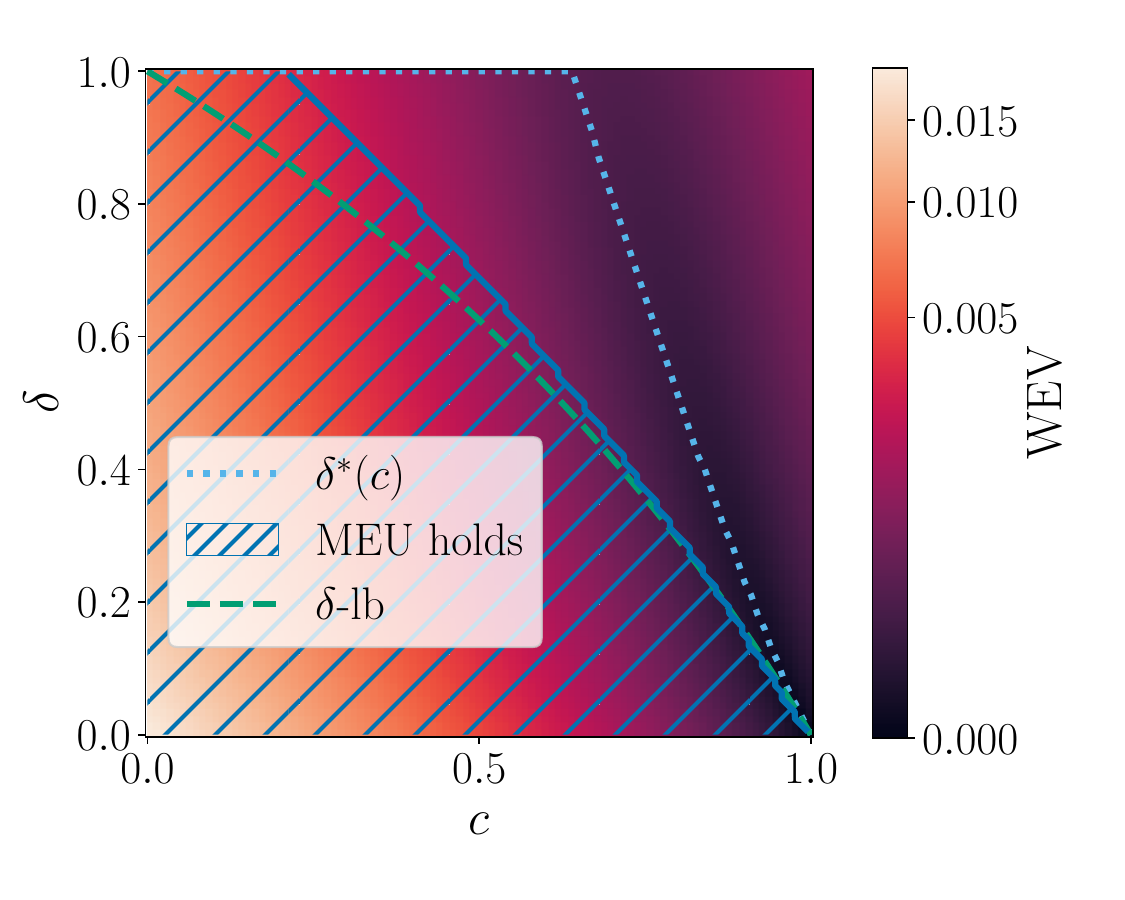}}
    \end{subfigure}
    \hspace{0.8cm}
    \begin{subfigure}[t]{0.36\textwidth}
        \subcaptionbox*{Beta$(0.5,0.5)$}{%
        \includegraphics[width=\linewidth,trim={1.2cm 0.3cm 0 0},clip]{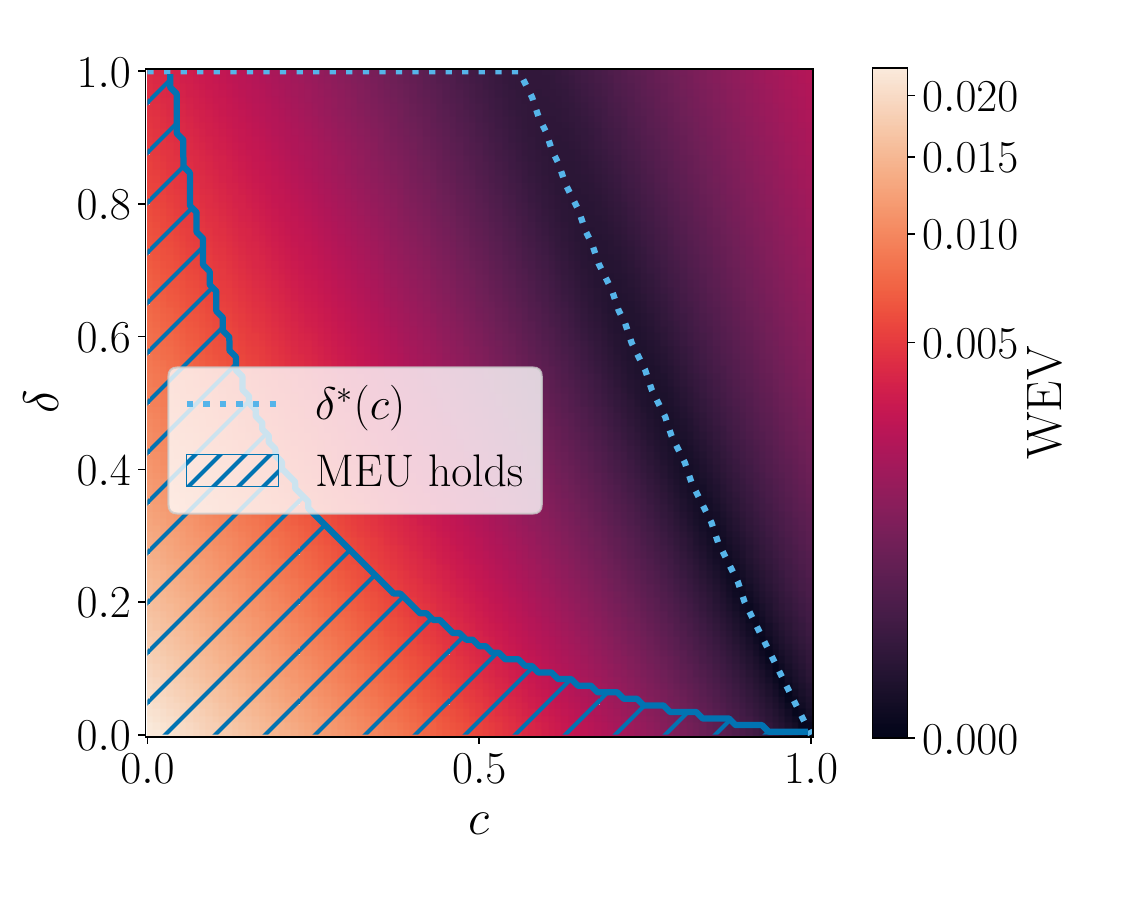}}
    \end{subfigure}
    \caption{\wev-minimizing design $\delta^*(c)$, lower bounds $\delta$-lb on the minimizer, and monotone ex post utility for representative signal distributions}
    \label{fig:WEV-optimal-various-distrib}
\end{figure}

First, consider the \wev-minimal pricing rule $\delta^*(c)$. The curve is qualitatively similar across the considered distributions. In line with \cref{theorem:bound-on-optimal-delta}---the uniform, exponential, and normal distribution are log-concave---the numerical results show that with a high private value component (low $c$), pay-as-bid pricing ($\delta=1$) minimizes \wev; with higher common value components (high $c$), strictly mixed pricing for some $\delta\in(0,1)$ minimizes \wev~(cf., \cref{prop:mixed-minimizing-WEV}); and with a pure common value ($c=1$), uniform pricing ($\delta=0$) minimizes \wev~(cf., \cref{theorem:pure-common}). For small common values, monotone ex post utility holds for all $\delta$ and so pay-as-bid pricing is dominant in \PD~(cf., \cref{prop:PD-decreasing}) and thus equity-optimal. Even for larger common values the \wev-minimal pricing is pay-as-bid, but eventually strictly mixing ($\delta\in(0,1)$) minimizes \wev. For a pure common value, uniform pricing is always \wev-minimal.\footnote{The \wevtext at the minimal $\delta^*$ decreases in $c$. With a higher common value share, bidders' values given different signal realizations as well as the corresponding bids move closer together, thus explaining smaller differences in utilities (ex post and in expectation).} In each case studied, we find a unique \wev-minimizer $\delta^*(c)$ at the resolution of our numerical grid.

Second, we note that for uniform and the truncated exponential distributions, we extend the region where \PD~are monotonically decreasing slightly beyond the diagonal $1-c$; details are given in \cref{cor:better-bounds} in \cref{app:further-results-examples-figures}. The full comparison in \cref{fig:WEV-optimal-various-distrib-full} shows these improved lower bounds and the general lower bound $1-c$ for the normal distribution. Each pricing rule at such a bound dominates any mixed-pricing rule with a lower level of price discrimination. For the Beta distribution (not log-concave), we do not have a theoretical lower bound on the \wev-minimal design $\delta^*$. However, the region where monotone ex post utility holds can be determined numerically, and its frontier also provides a lower bound for the \wev-minimal design.

Further remarks are due for the Beta distribution. The area where monotone ex post utility holds is much smaller; however, it is a sufficient condition for the monotonicity of the \wevtext, while it is necessary and sufficient for the monotonicity of \PD. From the heat map, we see that the \wevtext is monotone in $\delta$ for any given $c$ up to $\delta^*$. Otherwise, analogous interpretations to the log-concave cases hold. Even for the Beta distribution, pay-as-bid pricing is equity-preferred to uniform pricing for a range of private-common values including pure private values. For this result to invert (cf.,~\cref{sec:counterexample}), one would need an even more extreme probability density function, inducing an equilibrium in which ex post utilities are decreasing in signals for sufficiently many signal realizations.

Finally, note that, for the majority of common values, full price discrimination is needed to achieve \wev-minimizing pricing among mixed pricing rules. As an exploratory extension, we also consider $\delta>1$, under which each winner pays the first rejected bid plus $\delta$ times the gap between their bid and the first rejected bid. The supplementary appendix reports numerical examples in which this rule further reduces \wev. Although mixed-price auctions are found in practice, they generally leave positive surplus differences when values have a private component. More general mechanisms can achieve perfect surplus equity.

\section{Surplus-Equitable Mechanisms}\label{sec:surplus-equitable-mechanism}

This section turns from standard auction formats to direct payment design. \Cref{sec:mixed-auctions} showed how far one can go by changing only the pricing rule within standard auctions. We now ask what becomes feasible once transfers themselves are designed directly. We first consider surplus equity among winners, showing that it is compatible with both efficient and allocation-fair allocation rules, as well as any interpolation between the two. We then turn to surplus equity among all bidders, winners and losers alike, and show that this can be achieved for any monotone allocation rule.

Let a mechanism be characterized by a random allocation rule $\ssb \mapsto \xb(\ssb) \in \operatorname{conv}\{\mathbf{1}_S : S \subseteq [n],\, |S| \le k\}$, which assigns, for each signal profile $\ssb$, a distribution over feasible allocations, and a payment rule $(\ssb,X) \mapsto \xp(\ssb,X) \in \mathbb{R}^n$, where $X$ is the allocation. We write $x_i(\ssb)\in[0,1]$ for the allocation probability conditional on the reported signal profile and $X_i(\ssb)\in\{0,1\}$ for the realized allocation draw, with $\E[X_i(\ssb)\mid\ssb]=x_i(\ssb)$ and $\sum_i X_i(\ssb)\leq k$ almost surely. The ex post utility of bidder $i$ is
$u_i(\ssb)=X_i(\ssb)\cdot v_i(\ssb)-p_i(\ssb,X_i(\ssb))$,
so all ex post statements below refer to a realized draw $X_i$, while interim and ex ante statements average over both signals and the allocation randomization.

\subsection{Surplus-Equitable Design Among Winners}\label{sec:surplus-equity-winners}

Consider mechanisms that equalize surplus ex post among the winning bidders. Losers receive zero surplus (they pay nothing and receive no item). We present two such mechanisms: one that allocates items efficiently, and one that allocates items via a uniform lottery. Their convex combination traces out a family of mechanisms that interpolate between efficiency and allocative equity, while maintaining surplus equity among winners throughout.

\subsubsection{Efficient and Surplus-Equitable Mechanism}\label{sec:efficient-mechanism}

In our efficient mechanism, realized surplus is distributed equitably among the winning bidders. This means pairwise differences in ex post utilities are zero. The efficient allocation rule is given by $x_i(\ssb)=\ind{s_i \ge \kth{k}(\ssb)}, i = 1,\dots,n $, where $\kth{k}(\ssb)$ denotes the $k$-th highest element of the signal vector $\ssb$. Since signals are drawn from a continuous distribution, ties occur with probability zero. For completeness, we break ties at $s_i = \kth{k}(\ssb)$ uniformly at random. We denote $y:= \kth{k}(\smi)$ the first rejected signal with CDF $G:=\Gnk$, and pdf $g:=\gnk$.
\begin{theorem}\label{theorem:surplus-equity}\hyperref[proof:theorem:surplus-equity]{\normalfont[Proof in \cref{proof:theorem:surplus-equity}]}
There exists a BNIC direct mechanism that allocates the $k$ items efficiently
(to the $k$ highest signals) and equalizes ex post surplus among winners.
Losers pay nothing. Payments are given by
\begin{align}
\widetilde p_i(s_i,s_{-i})
=
\Big(
(1-c)\big(s_i-y-\tfrac{G(y)}{g(y)}\big)
+
V(y)
\Big)\ind{s_i>y},
\label{equ:equitable-payment}
\end{align}
where $y=\kth{k}(s_{-i})$ is the first rejected signal,
$G$ and $g$ are respectively the CDF and density of $y$,
and recalling that
\[
{V(y):=\widetilde V(y,y),\qquad
\widetilde V(x,y):=\E\!\left[v(s_i,s_{-i})\mid s_i=x,\ \kth{k}(s_{-i})=y\right].}
\]
\end{theorem}
The proof constructs ex post payments from the twin requirements of incentive compatibility and surplus equity. First, the term $(1-c)s_i$ removes the idiosyncratic part of each bidder's realized value that is due to their own signal. To align interim incentives, this term is adjusted by a uniform transfer $(1-c)(y + G(y)/g(y))$, which cancels the idiosyncratic payment in expectation. Thus, the second-price payment $V(y)$, the expected item value conditional on tying with the $k$th-highest competitor, induces truthful reporting for any given signal $s_i$.

Our surplus-equitable mechanism is ex post individually rational for bidders in the pure private value case. With mixed and common values, only interim individual rationality holds. The mechanism may require positive transfers to winners, as illustrated in \cref{example:equitable-mechanism}. The seller can afford the transfers in expectation; in fact, our mechanism is revenue-equivalent to a standard auction.

Distributing surplus equitably among winners ex post is the strongest implementation of equity in incentive-compatible, efficient, winners-pay auctions. Under efficient allocation, interim surplus is strictly increasing in a bidder's signal in any BNIC mechanism. Expected surplus conditional on winning can be equalized across signals if losing bidders also make payments. An example is given in \cref{example:payments-from-losers} in \cref{app:further-results-examples-figures}. The losing bidders then also enter equity considerations. Further intuition for ex post payments is given in the example below.

\begin{example}\label{example:equitable-mechanism}
    We consider the simplest market with three bidders competing for two items and signals which are distributed uniformly on $[0,1]$.
    We illustrate the ex post payment function (\cref{equ:equitable-payment}) for signals $s=0.2, 0.5, 0.8, 0.95$ as a function of the first rejected signal $y$ in \cref{fig:equ-pay-uniform-n3-k2}. Note that each payment corresponding to a signal $s$ is only plotted for $y\leq s$, as for $y > s$ the signal $s$ does not win and the payment is zero. The payments are from the bidder to the seller, i.e.,~negative payments constitute a transfer from the seller.
    \begin{figure}[htp]
    \centering
    \captionsetup{skip=6pt}
    \begin{adjustbox}{max width=0.93\textwidth}
        \hspace{-1cm}
        \begin{minipage}[t]{0.23\textwidth}
        \captionsetup[subfigure]{font=footnotesize,margin={1.2cm,0.5cm}}
        \subcaptionbox*{$c = 0$}{%
        \includegraphics[scale=0.347]{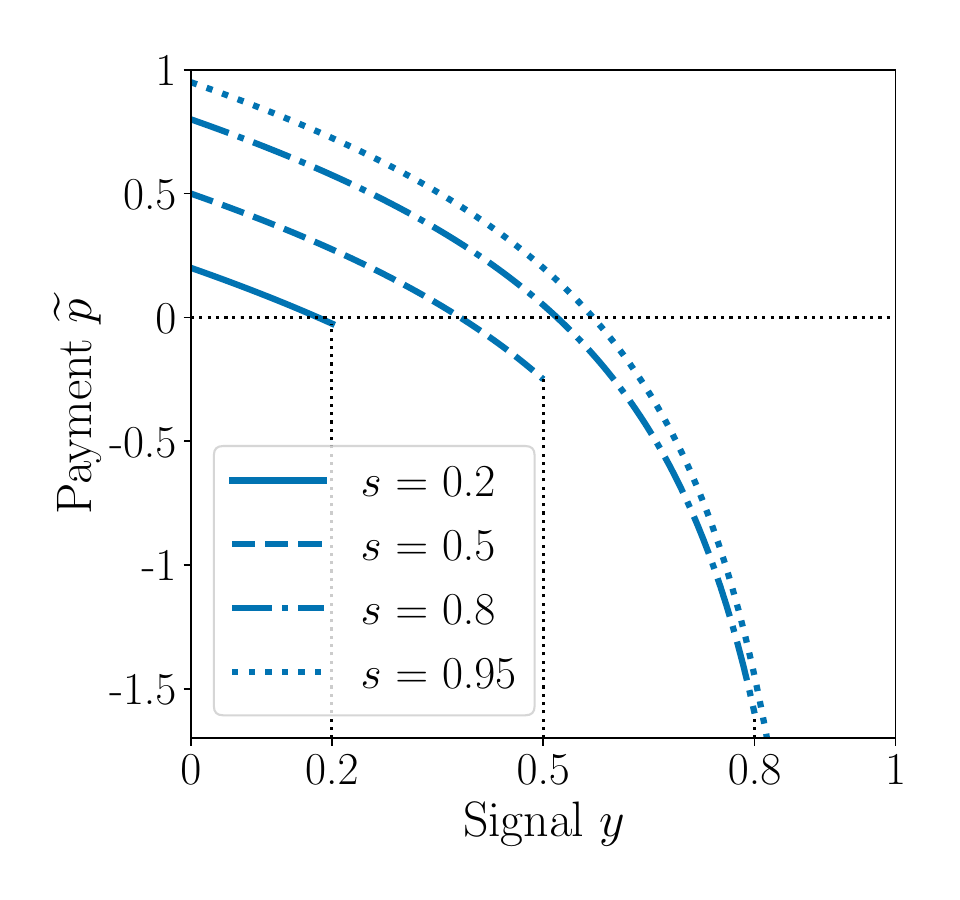}
        }
        \label{fig:equ-pay-uniform-c-0.00}
        \end{minipage}
        \hspace{1.35cm}
        \begin{minipage}[t]{0.25\textwidth}
        \captionsetup[subfigure]{font=footnotesize,margin={0.5cm,0.5cm}}
        \subcaptionbox*{$c = 0.5$}{%
        \includegraphics[scale=0.34]{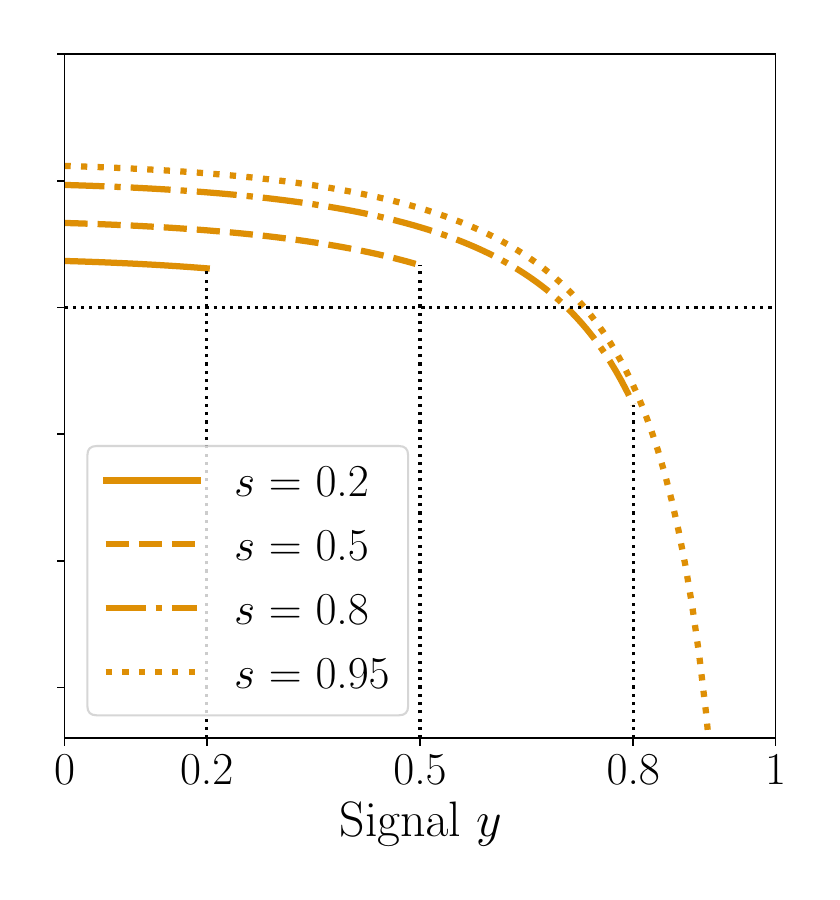}
        }
        \label{fig:equ-pay-uniform-c-0.50}
        \end{minipage}
        \hspace{0.23cm}
        \begin{minipage}[t]{0.25\textwidth}
        \captionsetup[subfigure] {font=footnotesize,margin={0.5cm,0.5cm}}
        \subcaptionbox*{$c = 0.8$}{%
        \includegraphics[scale=0.34]{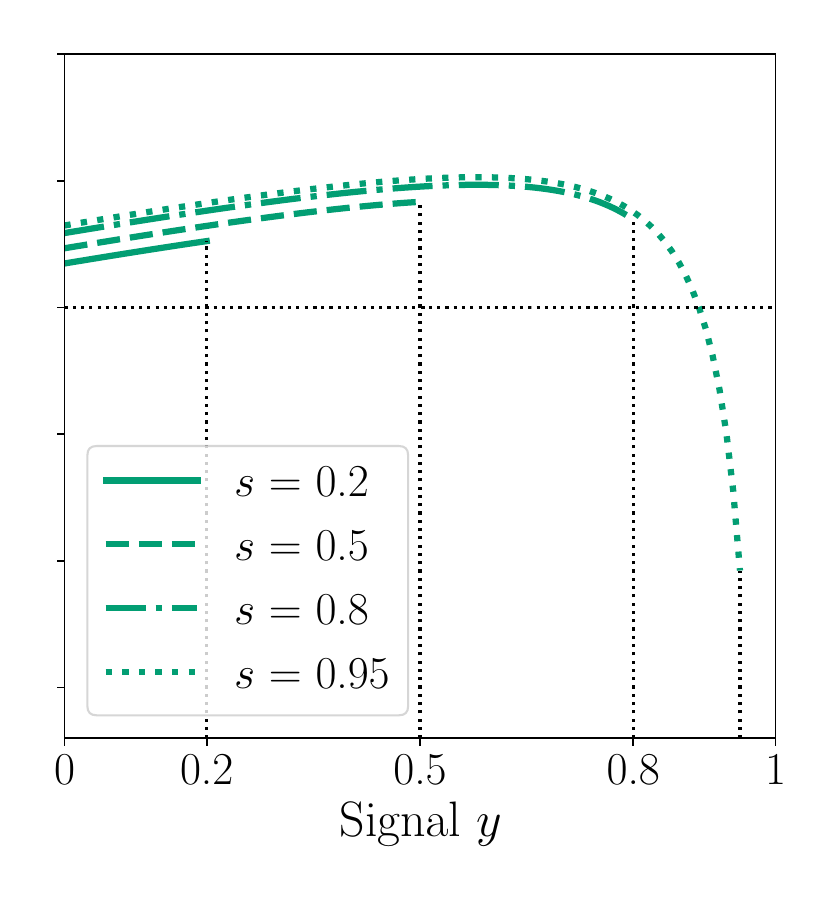}
        }
        \label{fig:equ-pay-uniform-c-0.80}
        \end{minipage}
        \hspace{0.23cm}
        \begin{minipage}[t]{0.25\textwidth}
        \captionsetup[subfigure] {font=footnotesize,margin={0.5cm,0.5cm}}
        \subcaptionbox*{$c = 1$}{%
        \includegraphics[scale=0.34]{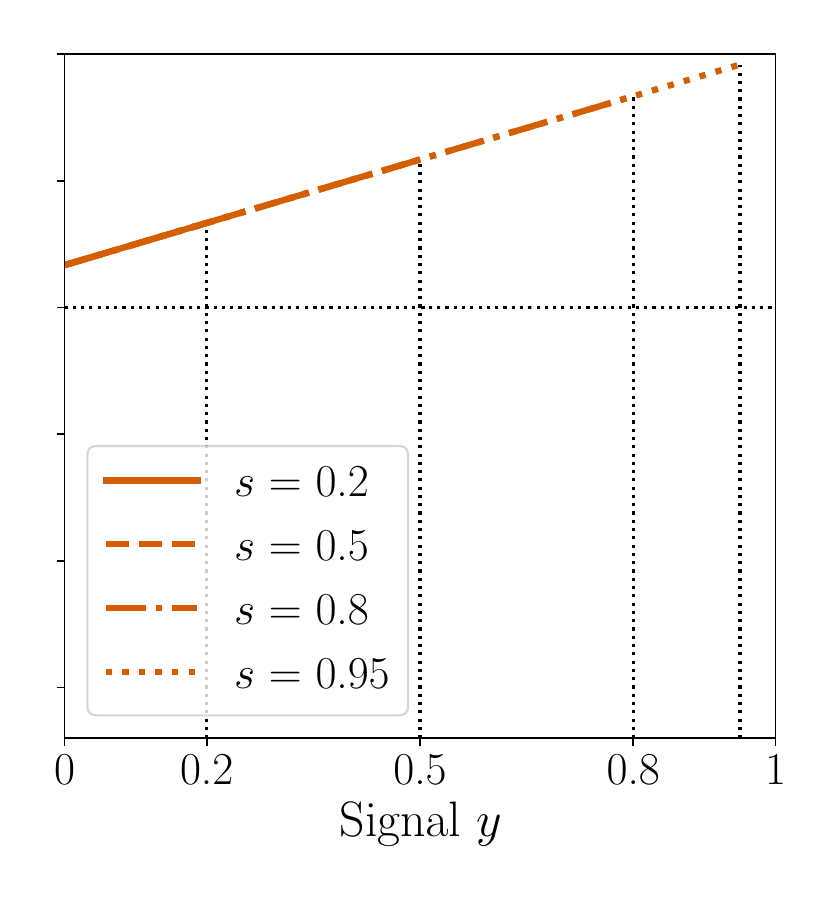}
        }
        \label{fig:equ-pay-uniform-c-1.00}
        \end{minipage}
    \end{adjustbox}
    \caption{Surplus-equitable payments, $\widetilde p$, for uniform signal distributions as a function of the signal $s$ and the first rejected signal, $y$, for common value parameters $c\in\{0,0.5,0.8,1\}$.}
    \label{fig:equ-pay-uniform-n3-k2}
    \end{figure}

    First, consider pure private values. If the first rejected signal is high relative to a winning bidder's signal, this bidder is paid a positive transfer. For such signal realizations, little redistribution is needed to equalize winning bidders' surpluses. If a high-signal bidder wins together with a low-signal bidder, a ``tax'' from the realized private value is levied in order to equalize ex post surpluses. In expectation, transfers paid when the first rejected signal is high are large enough to balance the incentives for under-reporting one's signal, which result from the redistributive tax. Thus, ex ante even a high-signal bidder has an incentive to report truthfully.

    Next, consider a pure common value. The payment reduces to the equilibrium bid of the first rejected signal in the FRB uniform-price auction. In this case, no taxation is needed, as the ex post surpluses are equal under any uniform payment rule. Truthful reporting is incentivized by the second-price rule; our mechanism in this case is the direct implementation of the FRB uniform-price auction. Recall that $V(y)$ is the expected value of a bidder with signal $y$ conditional on tying with the $k$th highest among their $n-1$ competitors.

    For intermediate cases of the common value, the two interpretations unite. Consider the case with $c=0.8$ in \cref{fig:equ-pay-uniform-n3-k2}: for low first rejected signals, the payment is simply increasing, reflecting the second-price incentive structure; but higher signals do pay more to cover the modest need for redistribution. However, for very high signals $s$, the redistribution must be made incentive-compatible in expectation, and thus the payment function eventually reaches its maximum and reverts into a transfer on a small range of first rejected signals $y$.
\end{example}
This construction shows that the limits of standard mixed auctions are not limits of incentive-compatible design more generally: once payments can be tailored directly, ex post surplus parity among winners is compatible with efficient allocation. A natural next question is whether one can strengthen the equity objective further by also requiring allocative fairness.

\subsubsection{Allocation-Fair and Surplus-Equitable Mechanism}\label{sec:allocative-equity}

A stronger notion of equity concerns the allocation of items itself---ensuring each bidder has the same chance of receiving an item, regardless of their signal. However, with indivisible items, allocative equity can only be achieved in expectation, creating again ex post equity concerns. Moreover, some social welfare must be sacrificed for allocative equity.

The uniform lottery is a parsimonious, fair allocation rule that distributes the $k$ items uniformly at random among the $n$ bidders, i.e.,~assignments are given by $x_i(\ssb)=\frac{k}{n}, i = 1,\dots,n$. The question is whether such equity is compatible with surplus equity. The following proposition states that there exists a winners-pay mechanism that delivers both ex ante allocative equity and ex post surplus equity.
\begin{proposition}\label{prop:alloc_fair}\hyperref[proof:prop:alloc_fair]{\normalfont[Proof in \cref{proof:prop:alloc_fair}]}
    There exists a mechanism that allocates items via uniform lottery and equalizes the winners' surpluses ex post. The mechanism is weakly dominant-strategy IC, interim individually rational, and ex ante budget balanced.
\end{proposition}
The mechanism's payments are given in the proof. This result establishes that there is no inherent trade-off between allocative equity among all bidders and surplus equity among winners. The budget-balanced payment rule gives the seller zero expected revenue. This is in stark contrast to the surplus-equitable, efficient mechanism, which is revenue-equivalent to a uniform-price auction. Ex post individual rationality can also be guaranteed, at the cost of an expected budget deficit when $c<1$. With a pure common value ($c=1$), the lottery has zero payments and is ex post individually rational (see the proof of \cref{prop:alloc_fair}).

The efficient and allocation-fair mechanisms can be combined to trade off efficiency and allocation fairness while retaining ex post surplus equity among winners.
Given mechanisms $(\xb,\xp)$ and $(\xb',\xp')$, and a parameter $\lambda \in [0,1]$, their mixture implements $(\xb,\xp)$ with probability $\lambda$ and $(\xb',\xp')$ with probability $1-\lambda$. The selected mechanism determines both the allocation and payments. Bidders know $\lambda$ and report their signals before learning which mechanism is selected. The following corollary is immediate.

\begin{corollary}\label{cor:mixture}
    The convex combination of the efficient and the allocation-fair surplus-equitable mechanisms is BNIC, interim IR, and generates a weakly positive expected budget surplus.
\end{corollary}

\subsection{Surplus-Equitable Design Among All Bidders}\label{sec:surplus-equity-all-bidders}

We now turn to the design goal of equalizing surplus among all bidders---winners and losers alike. This may be the designer's goal when all participants have equal claims to the benefits of a resource. We therefore also allow transfers to losing bidders. Our central result here is that this can be achieved for any monotone allocation rule.

Consider a feasible monotone allocation rule $x$, where for every bidder $i$ and fixed $s_{-i}$, $x_i(z,s_{-i})$ is non-decreasing in $z$. Let $q_i(z):=\E_{s_{-i}}\!\big[x_i(z,s_{-i})\big]$ denote the interim allocation probability for bidder $i$ who reports $z$, while the other bidders report truthfully. Further, define
\[
A_i(z)=(1-c+\tfrac{c}{n})\int_0^z q_i(t) \, dt-(1-c+\tfrac{c}{n})\E_{s_i}\!\left[\int_0^{s_i} q_i(t) \, dt\right],
\]
the mean-allocation advantage for bidder $i$.

\begin{theorem}\label{theorem:surplus-equitable-monotone-alloc}\hyperref[proof:theorem:surplus-equitable-monotone-alloc]{\normalfont[Proof in \cref{proof:theorem:surplus-equitable-monotone-alloc}]}
Consider a monotone allocation rule $x$, and suppose that signals are independent. The payment rule $p_i(\ssb,X_i(\ssb))={X_i(\ssb)} \cdot v_i(\ssb) - {\bar u(\ssb)}$ induces a mechanism that is BNIC, interim individually rational, and maximizes revenue among all mechanisms that implement the allocation rule $x$, are BNIC and interim IR, and equalize surplus ex ante; all bidders receive the ex post common surplus
\[
\bar u(\ssb) =  \sum_{i=1}^n A_i(s_i)+ (1-c+\tfrac{c}{n}) \max_{j \in [n]} \E_{s_j}\!\left[\int_0^{s_j} q_j(t)\,dt\right].
\]
\vspace{-0.75\baselineskip}
\end{theorem}
Since losers have $X_i(\ssb)=0$, the payment rule may require the auctioneer to make payments to, or collect payments from, losing bidders. The mechanism guarantees interim IR, but not necessarily ex post IR or ex post budget balance.
\cref{theorem:surplus-equitable-monotone-alloc} only requires independence of the signals; bidders may be non-symmetric. In particular, if bidders are symmetric and the allocation rule treats them symmetrically, then $q_i(z)=q_j(z)$ for every report $z$ and every pair $i,j$, and the revenue-maximizing BNIC, interim-IR implementation equalizes surplus ex ante. The payment rule therefore maximizes revenue among all BNIC, interim-IR mechanisms with the same allocation rule. Note that properties of allocative equity and efficiency depend on the allocation rule, while an appropriate choice of the pricing rule allows the mechanism designer to equalize surplus independently of the allocation rule.

We apply the general result to a uniformly random allocation. Then, our payment rule is also ex ante weakly-budget-balanced.

\begin{corollary}\label{cor:alloc_fair_all}\hyperref[proof:cor:alloc_fair_all]{\normalfont[Proof in \cref{proof:cor:alloc_fair_all}]}
For bidders with independent and identically distributed signals, there exists a mechanism that allocates items via uniform lottery and equalizes all bidders' surpluses ex post. The mechanism is Bayesian-Nash IC, interim individually rational, and ex ante weakly-budget-balanced.
\end{corollary}
The mechanism's payments are given in the proof. In contrast to \cref{prop:alloc_fair}, with a pure common value ($c=1$) the winners' payments are no longer zero, reflecting the stronger requirement of equalizing surplus across both winners and losers.

More broadly, these results separate two design choices: the allocation rule and how the resulting rents are shared. The allocation rule determines whether the mechanism is efficient or allocatively fair. Conditional on that rule, the payment rule can be designed to equalize surplus either among winners or among all bidders. In that sense, surplus equity can be achieved without introducing any further allocative distortion beyond the one already implied by the chosen allocation rule.

\section{Conclusion}\label{sec:conclusion}

We demonstrate that the division of bidder surplus depends systematically on the private-common value composition of bidders' values. Between uniform-price, pay-as-bid, and mixed auctions, more private-value environments call for more price discrimination. For many signal distributions, a simple bound rules out auction designs that price-discriminate too little from an equity perspective.

Standard pricing rules constrain equity-minded auction design. General mechanism design can, as we show, achieve perfect ex post surplus parity. For any fixed monotone allocation rule, our payment rule equalizes surplus ex post across all bidders and it maximizes revenue among BNIC, interim-IR implementations that equalize surplus ex ante. Under efficient allocation, winners' surpluses can be equalized ex post.

From a practical perspective, equity can be pursued, in addition to existing instruments, through an auction's payment rule directly. In applications such as treasury auctions, electricity markets, spectrum sales, or emissions trading, the distribution of surplus may affect firms' investment, survival, and subsequent competition through their available capital; our analysis treats these effects as motivation and abstracts from their effect on equilibrium bidding. Our results suggest a simple design use of the theory: For signal distributions with log-concave densities, the auctioneer can use the common-value proportion $c$ to rule out mixed auctions with price discrimination below $1-c$.

Several directions remain open. A natural next step is multi-unit demand, where equilibrium allocations may be inefficient and efficiency-equity trade-offs emerge (\cref{sec:multi-unit-demand}). Extending our work to other formats used in practice, such as dynamic or combinatorial auctions, would further clarify the scope of equity as a first-order design objective.

\clearpage

\appendix
\crefalias{section}{appendix}
\crefalias{subsection}{appendix}

\section{\texorpdfstring{Proofs for \cref{sec:setup} and \cref{sec:surplus-equity}}{Proofs for the Market Model and Surplus Equity}}

We derive several derivatives of the equilibrium bid function that we refer to later. To establish the derivatives, we first need some limit properties.

\begin{lemma}\label{lem:derivatives}\hyperref[proof:lem:derivatives]{\normalfont[Proof in \cref{proof:lem:derivatives}]} The following formulas hold for $s\in\osupport$, with derivatives in $s$ interpreted almost everywhere:
{\small
\begingroup
\allowdisplaybreaks
\begin{align*}
\bid(s) &= \begin{cases}
V(s) &\text{if }\delta=0\\
V(s) - \int_0^s V'(y) \left(\frac{G(y)}{G(s)}\right)^{\frac{1}{\delta}}\diff y &\text{if }\delta > 0
\end{cases}\\
\frac{\partial(\bid(s))}{\partial s} &= \begin{cases}
V'(s) &\text{if }\delta=0, s>0\\
\frac{g(s)}{G(s)}\int_0^s \frac{V'(y)}{\delta} \left(\frac{G(y)}{G(s)}\right)^{\frac{1}{\delta}}\diff y &\text{if }\delta > 0
\end{cases}\\
\delta\frac{\partial \bid(s)}{\partial \delta} &= \begin{cases}
0 &\text{if $\delta=0$ or $s=0$} \\
\int_0^s V'(y) \ln\left(\left(\frac{G(y)}{G(s)}\right)^{1/\delta}\right)\left(\frac{G(y)}{G(s)}\right)^{\frac{1}{\delta}}\diff y  &\text{for  }\delta,s >0
\end{cases}\\      
\frac{\partial^2( \delta\bid(s))}{\partial s\partial \delta} &= \frac{-g(s)}{\delta G(s)}\int_0^s V'(y) \log\left( \left( \frac{G(y)}{G(s)}\right)^{1/\delta} \right) \left( \frac{G(y)}{G(s)}\right)^{1/\delta} \diff y \text{  for  }\delta,s >0
\end{align*}
\endgroup
}
Moreover, $\bid(s)$ and $\delta\frac{\partial \bid(s)}{\partial \delta}$ are continuous in $\delta$. Both functions are defined at $s=0$ by continuous extension.
\end{lemma}

\begin{proof2}[Proof of \cref{prop:WEV-satisfies-Pigou-Dalton}]\phantomsection\label{proof:prop:WEV-satisfies-Pigou-Dalton}
    Fix a realized profile with winning set $W$, and choose $i,j\in W$ with $u_i>u_j$. For $0<t\leq (u_i-u_j)/2$, set $u_i'=u_i-t$ and $u_j'=u_j+t$, leaving all other utilities unchanged. The transfer preserves the ranking of $i,j$ and the sum of the winners' utilities.

    Using $\sum_{\ell,m\in W}(u_\ell-u_m)^2 = 2k\sum_{\ell\in W}u_\ell^2 - 2\left(\sum_{\ell\in W}u_\ell\right)^2$, we obtain
    \begin{align*}
        &\frac{1}{2k(k-1)}\sum_{\ell,m\in W}\bigl[(u'_\ell-u'_m)^2-(u_\ell-u_m)^2\bigr] = \frac{2t(t-u_i+u_j)}{k-1}<0.
    \end{align*}
    Taking the conditional expectation in \cref{def:WEV} proves the claim, with a strict decrease whenever positive transfers occur on a set of positive conditional probability.
\end{proof2}

\section{\texorpdfstring{Proofs for \cref{sec:mixed-auctions}}{Proofs for Mixed Auctions}}\label{app:sec:equity-comparisons}

\begin{proof2}[Proof of \cref{theorem:pure-common}]\phantomsection\label{proof:theorem:pure-common}
    To prove the ``if'' direction, note that for $c=1$, the realized value is identical for all bidders $i\in\bidders$ as $v(\ssb) = \frac{1}{n}\sum_{j\in\bidders}s_j$. Thus, with a uniform price that is identical between bidders, they all have identical surplus. For any $\delta > 0 $, the payment differs between the winners at least for some signal realizations. 
    
    To prove the ``only if'' let $\varphi^{\delta}(s)=(1-c)\cdot s-\delta \bid(s)$. We then have $(u(s_i)-u(s_j))^2=(\varphi^{\delta}(s_i)-\varphi^{\delta}(s_j))^2$ for two winning bids $s_i,s_j$ (see the \hyperref[proof:prop:delta-dominating-uniform]{proof} of \cref{prop:delta-dominating-uniform} for details). Thus,
    \begin{equation*}
        \frac{\partial}{\partial \delta} (\varphi^{\delta}(s_i)-\varphi^{\delta}(s_j))^2=-2(\varphi^{\delta}(s_i)-\varphi^{\delta}(s_j))\left(\bid(s_i)-\bid(s_j) + \delta \frac{\partial \bid(s_i)}{\partial \delta} - \delta \frac{\partial \bid(s_j)}{\partial \delta}\right).
    \end{equation*}

    Using \cref{lem:derivatives}, we take the limit of $\bid$, $\delta \frac{\partial \bid(s)}{\partial \delta}$, and $\varphi^{\delta}(s)$, as $\delta$ goes to $0$. We have that $(\varphi^{\delta}(s_i)-\varphi^{\delta}(s_j))\rightarrow (1-c)(s_i-s_j)$ and $(\bid(s_i)-\bid(s_j) + \delta \frac{\partial \bid(s_i)}{\partial \delta} - \delta \frac{\partial \bid(s_j)}{\partial \delta})\rightarrow (V(s_i)-V(s_j))$. As $V$ is increasing, the product $(V(s_j)-V(s_i))(s_i-s_j)$ is strictly negative almost surely, which concludes the proof. 
\end{proof2}

\bigbreak

\begin{proof2}[Proof of \cref{prop:mixed-minimizing-WEV}]\phantomsection\label{proof:prop:mixed-minimizing-WEV}
Fix any $c<1$ and let $\varphi^{\delta}(s):=(1-c)s-\delta \bid(s)$, so that for two winning signals $s_i>s_j$ we have
$(u_i^\delta-u_j^\delta)^2 = (\varphi^\delta(s_i)-\varphi^\delta(s_j))^2$.
From the proof of \cref{theorem:pure-common},
\[
\left.\frac{\partial}{\partial \delta}(\varphi^\delta(s_i)-\varphi^\delta(s_j))^2\right|_{\delta=0}
=-2(1-c)(s_i-s_j)(V(s_i)-V(s_j))<0,
\]
because $V$ is increasing and $1-c>0$.  Hence the right derivative at $\delta=0$ of each squared pairwise surplus difference is strictly negative, and therefore the same holds by dominated convergence for \wev, which is their conditional expectation by \cref{def:WEV}. It follows that for every $c<1$, there exists $\delta(c)\in(0,1)$ such that the $\delta(c)$-mixed auction has strictly lower \wev~than uniform pricing.

At $c=1$, \cref{theorem:pure-common} implies that uniform pricing has strictly lower \wev~than pay-as-bid pricing. Since \wev~is continuous in $c$, there exists $c^*<1$ such that for every $c\in(c^*,1)$ uniform pricing has lower \wev~than pay-as-bid pricing. Hence, for every $c\in(c^*,1)$ there exists an interior mixed auction with lower \wev~than both uniform and pay-as-bid pricing.
\end{proof2}

\subsection{\texorpdfstring{Further Derivations and Details on \cref{sec:counterexample}}{Further Derivations and Details on the Counterexample}}\label{app:counterexample}

First, we develop intuition for the example. Consider any winning signals $s_i > s_j$, $s_i,s_j \in \support$ and private values  ($c=0$). Let $u_i^0(s_i,\smi)$ and $u_i^1(s_i,\smi)$ denote bidder $i$'s utility in the uniform-price and pay-as-bid auction, respectively. Let $\bid[0]$ and $\bid[1]$ denote the corresponding symmetric equilibrium bid functions and $Y_{k+1}(\bbeta)$ the first rejected bid. For $\delta\in [0,1]$ and $c=0$, we have $u_i^\delta (s_i,\smi) = s_i - \delta\bid(s_i) - (1-\delta)Y_{k+1}(\bbeta)$. Thus, we have $\Delta u^0 := \vert u_i^0 - u_j^0 \vert = \vert s_i - s_j \vert$ and $\Delta u^1:=\vert u_i^1 - u_j^1 \vert = \vert s_i - \bid[1](s_i) - (s_j - \bid[1](s_j)) \vert = \vert s_i - s_j - (\bid[1](s_i) - \bid[1](s_j)) \vert$. It holds that 
\begin{align}
    & \Delta u^0 < \Delta u^1 \\
    \Leftrightarrow~ & s_i - s_j < \vert s_i - s_j - (\bid[1](s_i) - \bid[1](s_j)) \vert \label{equ:U0-and-U1}\\
    \Rightarrow~ & 2 (s_i - s_j) < \bid[1](s_i) - \bid[1](s_j)\label{equ:U0-and-U1-2}.
\end{align}
As bid functions are increasing, if $s_i - s_j - (\bid[1](s_i) - \bid[1](s_j))$ was positive, \cref{equ:U0-and-U1} could never hold. Thus, \cref{equ:U0-and-U1-2} follows as a necessary condition for uniform pricing to have lower \PD~than pay-as-bid pricing. For the same statement to hold for \wev, it must be that the bid function has a slope of at least 2 for a sufficient mass of signals $s_i$ and $s_j$.

We have to show that it is indeed possible to construct an \emph{equilibrium} bid function with a slope greater than 2 for a sufficient mass of signals. For this, we require an extreme signal distribution where, broadly speaking, signals are equal to zero with probability $\varepsilon$ and equal to one with probability $1-\varepsilon$. However, to compute a Bayesian-Nash equilibrium, we need a continuous signal distribution (with respect to the Lebesgue measure, without mass points) with connected support to solve the first-order condition. Thus, we add a small perturbation.

For \cref{example:counter}, we express the order statistics in terms of quantiles $x=F(s)$. For convenience, we define the following distribution functions and densities.
\begin{align*}
    \widetilde G(x) := \Gnk(F^{-1}(x)) &=
    1-(1-x)^{n-1} &
    \widetilde g(x) &:=
    (n-1)(1-x)^{n-2}\\
    \widetilde H(x) := G^{n-2}_{k-1}(F^{-1}(x)) &=
    1-(1-x)^{n-2} &
    \widetilde h(x) &:=
    (n-2)(1-x)^{n-3}
\end{align*}
We choose a continuous distribution of signals, with support $[0,2]$, where each signal is given by the sum of a Bernoulli$({1-\varepsilon})$ random variable and a random perturbation drawn from Beta$(1, 1/\eta)$, with $\varepsilon = 0.1/n$ and $\eta$ a small constant. First, we compute the distribution function $F$ and quantile function $F^{-1}$ of the signal distribution. Using the law of total probabilities yields $\forall s \in [0,2]$
{\small
\begin{align*}
F(s) &=\Prob[\text{Bernoulli}({1-\varepsilon})+\text{Beta}(1,1/\eta) \leq s]
\\&= \Prob[\text{Bernoulli}({1-\varepsilon}) = 0]\cdot \Prob[\text{Beta}(1, 1/\eta) \leq s] + \Prob[\text{Bernoulli}({1-\varepsilon}) = 1]\cdot \Prob[\text{Beta}(1, 1/\eta) \leq s-1].
\end{align*}
}
Simplifying this expression depending on the value of $s$, we get
\vspace{-0.3cm}
$$
\forall s\in [0,2],\qquad F(s) = \begin{cases}
    \varepsilon \cdot (1-(1-s)^{1/\eta})
    & \text{if }s \leq 1,\\
    \varepsilon + (1-\varepsilon)\cdot (1-(2-s)^{1/\eta})
    & \text{if }s \geq 1.
\end{cases}\\
$$
Then, computing piece-by-piece the inverse of $F$, we obtain
\vspace{-0.3cm}
$$
\forall x\in [0,1],\qquad F^{-1}(x) = \begin{cases}
    1-\left(1-\frac{x}{\varepsilon}\right)^{\eta} &\text{if }x\leq \varepsilon,\\
    2-\left(1-\frac{x-\varepsilon}{1-\varepsilon}\right)^{\eta} &\text{if }x\geq \varepsilon.\\
\end{cases}
$$
See \Cref{fig:counterexample-signals-quantiles} for the CDF and quantile function of the signals. A bidder with quantile $x\in [0,1]$ bids (truthfully) their signal $F^{-1}(x)$ in the uniform-price auction ($\delta = 0$), which we write as $b_\eta^0(x) := F^{-1}(x) = \ind{x \geq \varepsilon} + \gamma_\eta(x)$, where
\vspace{-0.3cm}
\[
\forall x\in [0,1],\qquad \gamma_\eta(x) := \begin{cases}
    1-\left(1-\frac{x}{\varepsilon}\right)^{\eta} &\text{if }x< \varepsilon,\\
    1-\left(1-\frac{x-\varepsilon}{1-\varepsilon}\right)^{\eta} &\text{if }x\geq \varepsilon.\\
\end{cases}
\]
For mixed auctions with $\delta > 0$, the equilibrium bid function is given by \Cref{prop:equilibrium-bid-delta}. Letting $b_\eta^\delta(x) := \beta^\delta(F^{-1}(x))$ denote the equilibrium bid of a bidder with quantile $x\in [0,1]$, we have
$$
b_\eta^\delta(x) = \frac{\int_0^{F^{-1}(x)} V(s) \gnk(s) \Gnk(s)^{\frac{1}{\delta}-1}\diff s}{\delta \Gnk(F^{-1}(x))^{\frac{1}{\delta}}} = \frac{\int_0^{x} F^{-1}(y) \widetilde g(y) \widetilde G(y)^{\frac{1}{\delta}-1}\diff y}{\delta \widetilde G(x)^{\frac{1}{\delta}}},
$$
where we used the change of variable $y = F(s)$. Finally, using the additive form of $F^{-1}$ we write the equilibrium bid function (see \cref{fig:example-bidfunction}) as $b_\eta^\delta(x) = b_0^\delta(x) + \xi_\eta^\delta(x)$, where
\begin{align*}
    \forall x \in [0,1],\qquad b_0^\delta(x) &:=
        \begin{cases}
            0 &\text{if }x < \varepsilon\\
            1-\left(\frac{\widetilde G(\varepsilon)}{\widetilde G(x)}\right)^{\frac{1}{\delta}} &\text{if }x \geq \varepsilon
        \end{cases}
    \\
    \qquad \xi_\eta^\delta(x) &:= \frac{\int_0^x \gamma_\eta(y) \widetilde g(y)\widetilde G(y)^{\frac{1}{\delta}-1}\diff y}{\delta \widetilde G(x)^{\frac{1}{\delta}}}
\end{align*}

Next, we define the function $\phi_\eta^\delta(x) := F^{-1}(x) - \delta b_\eta^\delta(x)$, the utility of a winning bidder plus the common uniform-price component. Denoting $\wev_\eta^\delta$ the \emph{winners' empirical variance} in a $\delta$-mixed auction with noise level $\eta$, we write
\begin{align*}
    \forall \delta \in [0,1],\;\forall \eta > 0,\qquad 
    \wev_\eta^\delta &= \E_{\mathbf x} \left[ \frac{(\phi_\eta^\delta(x_1)-\phi_\eta^\delta(x_2))^2}{2} \,|\, x_1,x_2 > \kth{k+1}(\mathbf x)\right].
\end{align*}
where $\mathbf x$ is a random vector of quantiles, with $n$ independent coordinates distributed uniformly on $[0,1]$. For every $x\in [0,1)$, observe that $\gamma_\eta(x)$ and $\xi_\eta^\delta(x)$ converge towards $0$ when taking $\eta$ arbitrarily small, and thus $\phi_\eta^\delta(x)$ converges towards $\phi_0^\delta(x) := \ind{x\geq \varepsilon}-\delta b_0^\delta(x)$. Therefore, because all relevant quantities are uniformly bounded thus dominated convergence applies, $\wev_\eta^\delta$ converges towards $\wev_0^\delta$, defined by
\begin{align*}
    \forall \delta \in [0,1],~ 
    \wev_0^\delta := & \E_{\mathbf x}\left[\frac{((\ind{x_1 \geq \varepsilon} - \delta b_0^\delta(x_1)) - (\ind{x_2 \geq \varepsilon} - \delta b_0^\delta(x_2)))^2}{2}\,|\, x_1,x_2 > \kth{k+1}(\mathbf x)\right] \\
    = & \E_{\mathbf x} \left[ \frac{(\phi_0^\delta(x_1)-\phi_0^\delta(x_2))^2}{2} \,|\, x_1,x_2 > \kth{k+1}(\mathbf x)\right]  =\lim_{\eta\rightarrow 0} \wev_\eta^\delta.
\end{align*}
We are now equipped to prove the proposition.
\medskip

\begin{proof2}[Proof of \cref{prop:counter-example}]\phantomsection\label{proof:prop:counter-example}
We write \\
$\wev^\delta_0 = \E_{\mathbf x}\left[\phi_0^\delta(x_1)^2\,|\, x_1 > \kth{k+1}(\mathbf x)\right]
     - \E_{\mathbf x}\left[\phi_0^\delta(x_1)\phi_0^\delta(x_2) \,|\, x_1,x_2 > \kth{k+1}(\mathbf x)\right]$, with
{\small
\begin{align*}
    \E_{\mathbf x}\left[\phi_0^\delta(x_1)^2\,|\, x_1 > \kth{k+1}(\mathbf x)\right] &= \frac{n}{n-1}\int_0^1 \phi_0^\delta(x)^2\widetilde G(x) \diff x\\
    \E_{\mathbf x}\left[\phi_0^\delta(x_1)\phi_0^\delta(x_2) \,|\, x_1,x_2 > \kth{k+1}(\mathbf x)\right] &= \frac{n}{n-2} \int_0^1 \left(\int_t^1 \phi_0^\delta(x) \diff x\right)^2\widetilde h(t)\diff t
\end{align*}
}
We next compute these quantities for uniform and discriminatory pricing.
For uniform pricing ($\delta = 0$) we have that $\phi_0^0(x) = \ind{x\geq \varepsilon}$. We derive
{\small
\begin{align*}
    \E_{\mathbf x}\left[\phi_0^0(x_1)^2\,|\, x_1 > \kth{k+1}(\mathbf x)\right]
    &= \frac{n}{n-1}\int_\varepsilon^1 \widetilde G(x)\diff x = \frac{n(1-\varepsilon)-(1-\varepsilon)^n}{n-1}\\
    \E_{\mathbf x}\left[\phi_0^0(x_1)\phi_0^0(x_2) \,|\, x_1,x_2 > \kth{k+1}(\mathbf x)\right] &= \frac{n}{n-2}\int_0^\varepsilon (1-\varepsilon)^2 \widetilde h(t)\diff t+\frac{n}{n-2}\int_\varepsilon^1 (1-t)^2 \widetilde h(t)\diff t\\
    &= \frac{n(1-\varepsilon)^2-2(1-\varepsilon)^n}{n-2}
\end{align*}
}
{\small
and finally
\begin{align*}
    \wev^0_0 &= \frac{n(1-\varepsilon)-(1-\varepsilon)^n}{n-1} - \frac{n(1-\varepsilon)^2-2(1-\varepsilon)^n}{n-2}\\
    &=\frac{n\left[(1-\varepsilon)^n+(1-\varepsilon)(\varepsilon(n-1)-1)\right]}{(n-1)(n-2)} \leq \frac{(\varepsilon n)^2/2}{n}=\frac{0.005}{n}
\end{align*}
}
For discriminatory pricing ($\delta = 1$) we have that $\phi_0^1(x) = \ind{x\geq \varepsilon}-b_0^1(x) = \ind{x\geq \varepsilon} \frac{\widetilde G(\varepsilon)}{\widetilde G(x)}$. \medskip 

We prove a few technical inequalities that will be used immediately after. First,
{\small
\begin{align*}
    \int_\varepsilon^1 \frac{1}{\widetilde G(x)}\diff x
    &= \int_\varepsilon^1 \frac{1}{1-(1-x)^{n-1}}\diff x
    =\int_\varepsilon^1 \sum_{i=0}^\infty (1-x)^{(n-1)i} \diff x =\sum_{i=0}^\infty \frac{(1-\varepsilon)^{(n-1)i+1}}{(n-1)i+1}\\
    & \geq (1-\varepsilon)+\sum_{i=1}^\infty\frac{(1-\varepsilon)^{ni}}{ni} \geq (1-\varepsilon)+\frac{1}{n}\sum_{i=1}^\infty\frac{0.9^i}{i} =
    1-\frac{0.1}{n}-\frac{\ln(0.1)}{n} \geq 1+\frac{2.2}{n}\\
    \int_0^1 \frac{x}{\widetilde G(x)}\diff x
    &= \int_0^1 \frac{x}{1-(1-x)^{n-1}}\diff x
    =\int_0^1 \sum_{i=0}^\infty x(1-x)^{(n-1)i} \diff x
    \\&=\frac{1}{2}+\sum_{i=1}^\infty \frac{1}{((n-1)i+1)((n-1)i+2)}
    \leq \frac{1}{2}+\frac{1.65}{n^2}
    \qquad \text{(when $n \geq 5$)}
\end{align*}
}

Then, we prove the upper bound $\int_\varepsilon^1 1/\widetilde G(x)\diff x \leq 1.7$. Let $q:=(1-\varepsilon)^{n-1}$. Since $1-y\leq e^{-y}$ and $n\geq5$,
$q\leq e^{-0.1(n-1)/n}\leq e^{-0.08}$. Using the series expansion above,
\begin{equation*}
    \int_\varepsilon^1 \frac{1}{\widetilde G(x)}\diff x
    \leq 1+\frac{1}{n-1}\sum_{i=1}^\infty\frac{q^i}{i}
    =1-\frac{\ln(1-q)}{n-1}\leq 1-\frac{\ln(1-e^{-0.08})}{4}
    <1.7,
\end{equation*}
where the second inequality uses $n\geq5$.

Finally, using $\widetilde G(t)\leq(n-1)t$ by Bernoulli's inequality and
$\widetilde G'(t)=(n-1)(1-t)^{n-2}$, we have
\begin{equation*}
    \int_\varepsilon^x
    \frac{\widetilde H(t)}{\widetilde G(t)}\diff t
    =
    \int_\varepsilon^x
    \left(
    1-\frac{t(1-t)^{n-2}}{\widetilde G(t)}
    \right)\diff t\leq
    \int_\varepsilon^x
    \left(
    1-\frac{(1-t)^{n-2}}{n-1}
    \right)\diff t=
    (x-\varepsilon)
    -\frac{\widetilde G(x)-\widetilde G(\varepsilon)}
    {(n-1)^2}.
\end{equation*}
We now bound and expand the following two expressions:
{\small
\begin{align*}
    \E_{\mathbf x}\left[\phi_0^1(x_1)^2\,|\, x_1 > \kth{k+1}(\mathbf x)\right]
    &= \frac{n}{n-1}\int_\varepsilon^1 \frac{\widetilde G(\varepsilon)^2}{\widetilde G(x)}\diff x
    \geq \frac{n\widetilde G(\varepsilon)^2}{n-1} \left(1+\frac{2.2}{n}\right)
\end{align*}
}
and
{\small
\begingroup
\allowdisplaybreaks
\begin{align*}
    &\E_{\mathbf x}\left[\phi_0^1(x_1)\phi_0^1(x_2) \,|\, x_1,x_2 > \kth{k+1}(\mathbf x)\right] \\
    =& \frac{n}{n-2}\int_0^\varepsilon \left(\int_\varepsilon^1 \frac{\widetilde G(\varepsilon)}{\widetilde G(x)}\diff x\right)^2 \widetilde h(t)\diff t +\frac{n}{n-2}\int_\varepsilon^1 \left(\int_t^1 \frac{\widetilde G(\varepsilon)}{\widetilde G(x)}\diff x\right)^2 \widetilde h(t)\diff t\\
    =&\underbrace{\frac{n\widetilde H(\varepsilon)}{n-2}\left(\int_\varepsilon^1 \frac{\widetilde G(\varepsilon)}{\widetilde G(x)}\diff x\right)^2+\frac{n}{n-2}
    \left[\left(\int_t^1 \frac{\widetilde G(\varepsilon)}{\widetilde G(x)}\diff x\right)^2 \widetilde H(t)\right]_\varepsilon^1}_{=0}+\frac{2n}{n-2}\int_\varepsilon^1 \left(\int_t^1 \frac{\widetilde G(\varepsilon)}{\widetilde G(x)}\diff x\right)\frac{\widetilde G(\varepsilon)}{\widetilde G(t)}\widetilde H(t)\diff t\\
    =&\frac{2n}{n-2}\int_\varepsilon^1 \left(\int_\varepsilon^x \frac{\widetilde G(\varepsilon)}{\widetilde G(t)}\widetilde H(t)\diff t\right)\frac{\widetilde G(\varepsilon)}{\widetilde G(x)}\diff x.
\end{align*}
\endgroup
}

Using $\varepsilon=0.1/n$, $\widetilde G(\varepsilon)\leq0.1$,
$1+2.2/n\leq\int_\varepsilon^1 1/\widetilde G(x)\diff x\leq1.7$,
and $\int_0^1 x/\widetilde G(x)\diff x\leq1/2+1.65/n^2$, we obtain
{\small
\begin{align*}
    \E_{\mathbf x}\left[\phi_0^1(x_1)\phi_0^1(x_2) \,|\, x_1,x_2 > \kth{k+1}(\mathbf x)\right]
    &\leq \frac{2n\widetilde G(\varepsilon)^2}{n-2}\left(\int_0^1\frac{x}{\widetilde G(x)}\diff x-\varepsilon\int_\varepsilon^1\frac{1}{\widetilde G(x)}\diff x-\frac{(1-\varepsilon)-\widetilde G(\varepsilon)\int_\varepsilon^1 \frac{1}{\widetilde G}}{(n-1)^2}\right)\\
    &\leq \frac{2n\widetilde G(\varepsilon)^2}{n-2}\left(\frac{1}{2}-\frac{0.1}{n}+\frac{1.43}{n^2}-\frac{0.83-0.1/n}{(n-1)^2}\right)
\end{align*}
}
Finally, writing
$B_n:=\frac{n}{n-1}(1+2.2/n)
-\frac{2n}{n-2}
\left(
\frac12-\frac{0.1}{n}
+\frac{1.43}{n^2}
-\frac{0.83-0.1/n}{(n-1)^2}
\right)$,
the previous two bounds give
{
\begin{align*}
    \wev_0^1\geq
    \frac{n\widetilde G(\varepsilon)^2}{n-1}
    \left(1+\frac{2.2}{n}\right)
    -\frac{2n\widetilde G(\varepsilon)^2}{n-2}
    \left(
    \frac12-\frac{0.1}{n}
    +\frac{1.43}{n^2}
    -\frac{0.83-0.1/n}{(n-1)^2}
    \right)=\widetilde G(\varepsilon)^2B_n.
\end{align*}
}
A direct simplification gives $nB_n-\frac{89}{100}
    =
    \frac{51n^3-364n^2+567n-108}
    {100(n-2)(n-1)^2}$, the numerator equals $2$ at $n=5$ and is increasing for $n\geq5$, so $B_n\geq0.89/n$. Since
$\widetilde G(\varepsilon)\geq1-e^{-0.08}>0.076$, we conclude
$\wev_0^1\geq\widetilde G(\varepsilon)^2B_n
\geq(0.076)^2\cdot0.89/n>0.005/n$.
Combining the two bounds, $\wev_0^1>0.005/n\geq\wev_0^0$, so $\wev_0^1-\wev_0^0>0$, for $n\geq 5$. Since $\eta\mapsto\wev_\eta^\delta$ is continuous at $\eta=0$ for both $\delta=0$ and $\delta=1$ (as established in the limit argument preceding the proof), there exists $\eta^*>0$ such that for all $\eta\in(0,\eta^*]$, $\wev_\eta^0 < \wev_\eta^1$, i.e., uniform pricing yields strictly lower \wev~than pay-as-bid pricing.
\end{proof2}

\subsection{\texorpdfstring{Proof of \cref{theorem:bound-on-optimal-delta}}{Proof of the Bound on Optimal Delta}}\label{proof:theorem:bound-on-optimal-delta}
To prove \cref{theorem:bound-on-optimal-delta}, we combine two propositions on monotonicity of \PD~and dominance of \PD, respectively, with a third proposition that bounds the slope of bid functions. First, we identify the property of \emph{monotone ex post utility} (MEU, \cref{def:monotone-ex-post-utility}) as a fundamental and sufficient condition for our dominance results. MEU relates to the slope of equilibrium bids as follows.
\begin{lemma}\label{lem:monotone-ex-post-equivalence}\hyperref[proof:lem:monotone-ex-post-equivalence]{\normalfont[Proof in \cref{proof:lem:monotone-ex-post-equivalence}]}
    For $\delta\in(0,1]$, the symmetric equilibrium satisfies MEU among winners if and only if $\pd{\bid}{s}\leq \frac{1-c}{\delta} $ for almost every $s\in\osupport$.
\end{lemma}
Among winners, the ex post difference in utilities depends only on the private value proportion $(1-c)$ and the discriminatory part of the payment $\delta\bid$. Thus, as long as the discriminatory payment does not grow faster in the signal than the private-value share, ex post utilities are monotone.

The Bayesian-Nash equilibrium exhibits several monotonicity properties, and these hold beyond uniform signals. By assumption, equilibrium bids are increasing in the bidder's own signal. Equilibrium bids are also decreasing in the extent of price discrimination: if a higher proportion of one's own bid affects the price, the incentive to bid-shade increases. Finally, the change in the payment corresponding to a bidder's own bid, due to a signal increase, is increasing in the weight of price discrimination, and vice versa. The latter monotonicity is crucial for \cref{prop:PD-decreasing}.
\begin{lemma}\label{lem:equilibrium-bid-monotone}\hyperref[proof:lem:equilibrium-bid-monotone]{\normalfont[Proof in \cref{proof:lem:equilibrium-bid-monotone}]}
    The equilibrium bid functions satisfy the following monotonicity properties:
    \begin{enumerate}
        \item $\bid(s)$ is strictly increasing in $s$, for all fixed $\delta\in[0,1]$ (consistent with the assumption), and
        is strictly decreasing in $\delta$, for all fixed $s\in\osupport$.
        \item  $\frac{\partial(\delta\bid(s))}{\partial \delta}$ is strictly increasing in $s$, for all fixed $\delta\in[0,1]$, and $\frac{\partial(\delta\bid(s))}{\partial s}$ is strictly increasing in $\delta$, for almost every $s\in\osupport$.
    \end{enumerate}
\end{lemma}
We now characterize the fundamental role of monotone ex post utility: it is equivalent to the monotonicity property of \PD.
\begin{proposition}\label{prop:PD-decreasing}\hyperref[proof:prop:PD-decreasing]{\normalfont[Proof in \cref{proof:prop:PD-decreasing}]}
    For a given common value $c$ and for some $\bar\delta \in [0,1]$, \PD~are monotonically decreasing over $[0,\bar\delta]$ if and only if the equilibrium (which depends on $c$ and $\bar\delta$) satisfies MEU.
\end{proposition}
\begin{proposition}\label{prop:delta-dominating-uniform}\hyperref[proof:prop:delta-dominating-uniform]{\normalfont[Proof in \cref{proof:prop:delta-dominating-uniform}]}
    For a given common value $c$, consider any $\delta$-mixed auction, $\delta\in(0,1]$, and suppose the equilibrium bidding function $\bid$ satisfies $ \pdbdeltas \leq \frac{2(1-c)}{\delta}$ for almost every $s\in\osupport$. Then, absolute pairwise surplus differences among winners are weakly smaller under $\delta$-mixed pricing than under uniform pricing.
\end{proposition}
The equivalence between decreasing \PD~and MEU being satisfied in equilibrium is crucial in our proof of \cref{theorem:bound-on-optimal-delta}. When MEU holds, the slope of the equilibrium bid function is sufficiently flat and more price discrimination impacts higher signal bidders more than lower signal bidders. In contrast, including more uniform pricing in the price mix will, proportionally to the change in $\delta$, offer higher signal bidders a higher discount than lower signal bidders and thus does not improve surplus equity. A similar intuition holds for \cref{prop:delta-dominating-uniform}. The final crucial proposition bounds the slope of the equilibrium bid functions by 1 for the family of log-concave signal distributions.
\begin{proposition}\label{prop:derivative-beta-bounded}\hyperref[proof:prop:derivative-beta-bounded]{\normalfont[Proof in \cref{proof:prop:derivative-beta-bounded}]}
    If the signal density $f$ is log-concave, then $ \frac{\partial \bid(s)}{\partial s} \leq 1$ for all signals $s\in\osupport$.
\end{proposition}
With pure private values, this bound implies that ex post utility is non-decreasing in signals for log-concave signal distributions. Indeed, for $u(s)=s-\delta \bid(s) - (1-\delta)Y_{k+1}(\bbeta)$, we have that $\frac{\partial u}{\partial s} = 1- \delta \frac{\partial \bid(s)}{\partial s} \geq 0$.

Now we are able to conclude the proof. Because of \cref{prop:derivative-beta-bounded}, we have that under log-concave signal distributions MEU holds if $\delta \leq (1-c)$, as $\pdbdeltas \leq 1 \leq \frac{(1-c)}{\delta}$ (see \cref{lem:monotone-ex-post-equivalence}). Thus, applying \cref{prop:PD-decreasing}, \PD~are monotonically decreasing for $\delta \in [0, 1-c]$, and dominance of $\delta =1-c$ follows.

Similarly, because of \cref{prop:derivative-beta-bounded}, it holds that with log-concave signals, $\pdbdeltas \leq 1 \leq \frac{2(1-c)}{\delta}$ if $\delta \leq 2(1-c)$. By \cref{prop:delta-dominating-uniform}, the pairwise surplus inequalities hold weakly. For $\delta<2(1-c)$, they are strict for distinct winning signals. At $\delta=2(1-c)>0$, \cref{equ:alternative-equilibrium-delta} and $V'\leq1$ give $\bid(s)-\bid(0)<s$ for every $s>0$. Hence some interior pair $s_i>s_j$ satisfies $\bid(s_i)-\bid(s_j)<s_i-s_j$, so its surplus inequality is strict. Continuity and the positive signal density then give strict inequality on a set of winning profiles with positive probability. This proves dominance throughout the stated range.
\bigbreak

\begin{proof2}[Proof of \cref{lem:monotone-ex-post-equivalence}]\phantomsection\label{proof:lem:monotone-ex-post-equivalence}
    Fix a realized signal profile $\ssb$ and two winning bidders $i,j$ with $s_i \leq s_j = s_i + \varepsilon$, for some $\varepsilon > 0$. Since both bidders win, they face the same uniform-price term $Y_{k+1}(\bbeta)$. Then
    \begin{equation*}
    \begin{array}{rcl}
        u_i(\ssb)  \leq  u_j(\ssb) 
    &    \Leftrightarrow&\quad (1-c)s_i - \delta\bid(s_i)  \leq  (1-c)s_j - \delta\bid(s_j) \\
     &   \Leftrightarrow&\quad (1-c)(s_j - s_i)  \geq  \delta (\bid(s_j) - \bid(s_i)).
    \end{array}
    \end{equation*}
    By \cref{equ:equilibrium-bid-delta}, $\bid$ is absolutely continuous on every compact subinterval of $\osupport$: the numerator is an indefinite integral of a locally integrable function, and the denominator is absolutely continuous and bounded away from zero there. Because $k\geq2$, any two interior signals can occur as winning signals. If MEU holds, the preceding inequality therefore holds for every such pair. Difference quotients give the slope bound wherever $\bid$ is differentiable, hence almost everywhere.

    Conversely, suppose the slope bound holds almost everywhere. For winning signals $0<s_i\leq s_j<\vbar$,
    \begin{equation*}
        u_j(\ssb)-u_i(\ssb)
        =\int_{s_i}^{s_j}\left(1-c-\delta\pd{\bid(t)}{t}\right)\diff t\geq0.
    \end{equation*}
    Continuity of $\bid$ at zero extends the conclusion to $s_i=0$.
\end{proof2}

\bigbreak

\begin{proof2}[Proof of \cref{lem:equilibrium-bid-monotone}]\phantomsection\label{proof:lem:equilibrium-bid-monotone}
Monotonicity follows from the derivatives computed in \Cref{lem:derivatives}.
\end{proof2}

\bigbreak

\begin{proof2}[Proof of \cref{prop:PD-decreasing}]\phantomsection\label{proof:prop:PD-decreasing}
    At $\delta=0$, MEU holds directly. The case $\bar\delta=0$ is therefore immediate; suppose $\bar\delta>0$. First, assume that the equilibrium at $\bar\delta$ satisfies MEU. By \cref{lem:monotone-ex-post-equivalence}, this is equivalent to
        $\bar\delta \pd{\bid[\bar\delta](s)}{s}\leq 1-c$ for almost every $s\in\osupport$.
    Since $\frac{\partial(\delta\bid(s))}{\partial s}$ is increasing in $\delta$ by \cref{lem:equilibrium-bid-monotone} (2), it follows that $\delta \pd{\bid[\delta](s)}{s}\leq \bar\delta \pd{\bid[\bar\delta](s)}{s}\leq 1-c$ for each $\delta\in(0,\bar\delta]$ and almost every $s\in\osupport$. Hence every equilibrium on $[0,\bar\delta]$ satisfies MEU. Now fix two winning bidders with signals $s_i\geq s_j$ and define $\varphi^{\delta}(s):=(1-c)s-\delta \bid[\delta](s)$. Then, for any $0\leq \delta_1 \leq \delta_2 \leq \bar\delta$, by MEU, we have that $u^{\delta_1}(s_i)\geq  u^{\delta_1}(s_j)$ and $u^{\delta_2}(s_i) \geq u^{\delta_2}(s_j)$. Therefore
    {\small
    \begin{align*}
        |u^{\delta_1}(s_i) - u^{\delta_1}(s_j) |
        =u^{\delta_1}(s_i) - u^{\delta_1}(s_j)
        = \varphi^{\delta_1}(s_i)-\varphi^{\delta_1}(s_j) 
        = (1-c)(s_i-s_j)-\delta_1 \left( \bid[\delta_1](s_i) - \bid[\delta_1](s_j)\right) \\
        \geq (1-c)(s_i-s_j)-\delta_2 \left(\bid[\delta_2](s_i) - \bid[\delta_2](s_j) \right) 
        = \varphi^{\delta_2}(s_i)-\varphi^{\delta_2}(s_j) 
        = u^{\delta_2}(s_i) - u^{\delta_2}(s_j) 
        = |u^{\delta_2}(s_i) - u^{\delta_2}(s_j) |,
    \end{align*}
    }
    where the inequality follows from \cref{lem:equilibrium-bid-monotone} (2). Therefore \PD~are monotonically decreasing on $[0,\bar\delta]$.
    Now, for the converse, assume that \PD~are monotonically decreasing on $[0,\bar\delta]$. Fix two winning bidders with signals $s_i\geq s_j$ and define
    \[
        D(\delta):=\varphi^{\delta}(s_i)-\varphi^{\delta}(s_j)
        =(1-c)(s_i-s_j)-\delta\left(\bid[\delta](s_i)-\bid[\delta](s_j)\right).
    \]
    By \cref{lem:derivatives}, $\bid[\delta](s)$ is continuous in $\delta$, hence so is $D(\delta)$. Moreover, \cref{lem:equilibrium-bid-monotone} (2) implies that $D(\delta)$ is weakly decreasing in $\delta$. Since \PD~are monotonically decreasing, $|D(\delta)|$ is weakly decreasing on $[0,\bar\delta]$. 
    If $c=1$, then MEU and monotonicity of pairwise differences on $[0,\bar\delta]$ both hold only in the trivial case $\bar\delta=0$; indeed, for every $\delta>0$, all bidders have the same ex-post value but higher-signal winners pay a strictly larger discriminatory component. Now suppose that $c<1$. At $\delta=0$, almost surely, $D(0)=(1-c)(s_i-s_j)> 0$. Suppose for contradiction that $D(\bar\delta)<0$. By continuity, there exists $\delta^*\in[0,\bar\delta)$ such that $D(\delta^*)=0$. Because $D$ is weakly decreasing, we have $D(\delta)\leq 0$ for all $\delta\in[\delta^*,\bar\delta]$, and therefore $|D(\delta)|=-D(\delta)$ is weakly increasing on $[\delta^*,\bar\delta]$. Since $D(\bar\delta)<0$, this increase is strict somewhere on that interval, contradicting the monotone decrease of \PD. Hence $D(\bar\delta)\geq 0$, i.e. $u_i^{\bar\delta}(s_i)-u_j^{\bar\delta}(s_j)\geq 0$. As this holds a.s. for every pair of winning signals with $s_i\geq s_j$, the equilibrium at $\bar\delta$ satisfies MEU. 
\end{proof2}

\bigbreak

\begin{proof2}[Proof of \cref{prop:delta-dominating-uniform}]\phantomsection\label{proof:prop:delta-dominating-uniform}
    Let $u_i^\delta(s_i,\smi)$ denote bidder $i$'s utility in the $\delta$-mixed auction, and let $u_i^{U}(s_i,\smi)$ denote bidder $i$'s utility in the uniform-price auction. Now let $i,j \in \bidders$ be two winning bidders. As above, $\bid$ (resp. $\bid[U]$) denotes the symmetric equilibrium bid function in the $\delta$-mixed (resp. uniform-price) auction. Let $Y_{k+1}(\bbeta)$ denote the first rejected bid. Then, canceling out $(1-\delta)Y_{k+1}(\bbeta)$, we have
    \begin{align*}
        \vert u_i^\delta - u_j^\delta \vert & = \vert (v_i(s_i,\smi) - \delta\bid(s_i) ) - (v_j(s_j,\smj) - \delta\bid(s_j)) \vert \\
        & = \vert ((1-c)s_i + \frac{c}{n}\sum_{\ell\in\bidders}s_\ell - \delta\bid(s_i) ) - ((1-c)s_j + \frac{c}{n}\sum_{\ell\in\bidders}s_\ell - \delta\bid(s_j)) \vert \\
        & = \vert ((1-c)s_i - \delta\bid(s_i)) - ((1-c)s_j - \delta\bid(s_j)) \vert  = \vert \varphi^\delta(s_i)-\varphi^\delta(s_j) \vert,
    \end{align*}
    where $\varphi^{\delta}(s)=(1-c)s - \delta\bid(s)$. It also holds that
    \begin{align*}
        \vert u_i^{U} - u_j^{U} \vert = \vert (v_i(s_i,\smi) - Y_{k+1}(\bbeta)) - (v_j(s_j,\smj) - Y_{k+1}(\bbeta)) \vert = \vert (1-c) (s_i - s_j) \vert.
    \end{align*}
    Since $\bid$ is locally absolutely continuous on $\osupport$, so is $\varphi^\delta$. The derivative bound $\vert\pd{\varphi^\delta}{s}\vert\leq1-c$ almost everywhere therefore implies the required inequality by integration. Continuity at zero covers the endpoint. We have $\pd{\varphi^{\delta}}{s}=1-c - \delta \pdbdeltas$.
    As $\bid$ is increasing in $s$, $\vert\pd{\varphi^{\delta}}{s}\vert\leq 1-c$ holds whenever $\delta \pdbdeltas \leq 2(1-c)$. Therefore
    \begin{equation*}
        \vert u_i^\delta - u_j^\delta \vert = \vert \varphi^{\delta}(s_i) - \varphi^{\delta}(s_j) \vert \leq \vert (1-c)(s_i -s_j) \vert = \vert u_i^{U} - u_j^{U} \vert.
        \qedhere
    \end{equation*}
\end{proof2}

\bigbreak

\begin{proof2lower}[Proof of \cref{prop:derivative-beta-bounded}]\phantomsection\label{proof:prop:derivative-beta-bounded}
Bounding the bid function slope for log-concave distributions requires three main observations, which we detail in the lemmas below and then use to prove \cref{prop:derivative-beta-bounded}.
The first lemma simplifies $V(s)$ which allows bounding $V'(s)$ by 1.
\begin{lemma}\label{lem:value-function-bounded} \hyperref[proof:lem:value-function-bounded]{\normalfont[Proof in \cref{proof:lem:value-function-bounded}]}
    Assuming common values ($c=1$) and $s\in\osupport$, $V(s)$ can be expressed as
    \begin{equation*}
        V(s)= \frac{2}{n}s+\frac{n-k-1}{n} \frac{\int_0^{s} tf(t)\diff t}{F(s)}+\frac{k-1}{n}\frac{\int_s^{\vbar} tf(t)\diff t}{1-F(s)}
    \end{equation*}
    Moreover, if the signal density $f$ is log-concave, then $V$ is differentiable on $\osupport$ and $V'(s)\leq 1$ for all signals $s \in \osupport$.
\end{lemma}
The proof proceeds by noticing that order statistics conditioned on other order statistics behave just like order statistics of a truncation of the original distribution. Thus, a more tractable expression of the expected valuation $V$ can be derived for the pure common value case ($c=1$). Together with results on log-concavity by \citet{Bagnoli-Bergstrom-2005}, we use this expression to show that $V'\leq 1$ for all signals $s\in\osupport$.

The next lemma establishes a sufficient condition for the equilibrium bid function slope to be bounded by 1 in the pure private value case. Differentiating twice $\int_0^s G^{1/\delta}$, we establish that its log-concavity is equivalent to $\frac{\partial \bid(s)}{\partial s} \leq 1$.

\begin{lemma}\label{lem:private-derivative-bounded}\hyperref[proof:lem:private-derivative-bounded]{\normalfont[Proof in \cref{proof:lem:private-derivative-bounded}]}
    Assuming private values ($c=0$), for any $\delta \in (0,1]$, $\frac{\partial \bid(s)}{\partial s} \leq 1$ for almost every $s\in\osupport$ iff $\int^{s}_0 G^{\frac{1}{\delta}}(y)\diff y $ is log-concave on $\osupport$.
\end{lemma}
Finally, we establish that a log-concave signal density is sufficient for the integral of their order statistics to be log-concave, using closure properties of product and integration of log-concave distributions, and results by \citet{Bagnoli-Bergstrom-2005}.
\begin{lemma}\label{lem:intG-log-concave}\hyperref[proof:lem:intG-log-concave]{\normalfont[Proof in \cref{proof:lem:intG-log-concave}]}
    If the signal density $f$ is log-concave, then so is $\int^{s}_0 G^{\frac{1}{\delta}}(y)\diff y$.
\end{lemma}
With the three lemmas above, we can prove \cref{prop:derivative-beta-bounded}. Recall the expression of the derivative of the bid function for any $s \in (0,\vbar)$:
\begin{equation}
    \pd{\bid(s)}{s}= \frac{g(s)}{G(s)}\frac{\int_0^s V'(y) G^{\frac{1}{\delta}}(y)\diff y}{\delta G^{\frac{1}{\delta}}(s)}
\end{equation}
Note that for any $c \in [0,1]$, $V(s)=(1-c)s+c\,V^{cv}(s)$, where $V^{cv}$ is the value function in the pure common-value case ($c=1$). By \cref{lem:value-function-bounded}, $(V^{cv})'(s)\leq 1$, so $V'(s)=(1-c)+c\,(V^{cv})'(s)\leq 1$.
Moreover, because of \cref{lem:intG-log-concave}, we know that $\int_0^s G^{\frac{1}{\delta}}$ is log-concave, and we can therefore apply \cref{lem:private-derivative-bounded}. 
Hence using the above results,
\begin{equation*}
    \frac{\partial \bid(s)}{\partial s} \leq   \frac{g(s)}{G(s)}\frac{\int_0^s 1\cdot G^{\frac{1}{\delta}}(y)\diff y}{\delta G^{\frac{1}{\delta}}(s)} \leq 1.
    \qedhere
\end{equation*}
\end{proof2lower}

\bigbreak

\begin{proof2}[Proof of \cref{lem:value-function-bounded}]\phantomsection\label{proof:lem:value-function-bounded}
We first rewrite $\widetilde V(x,y)$ for $c=1$ in terms of the order-statistics of $s_{-i}$.
{\small
\begingroup
\allowdisplaybreaks
\begin{align*}
    \widetilde V(x,y)&=\E[v(s_i,s_{-i}) \mid s_i=x,Y_k=y]=\E[\frac{1}{n} \sum_{j\in\bidders} s_j \mid s_i=x,Y_k=y]=\frac{x}{n}+\E[\frac{1}{n} \sum_{\substack{j\in\bidders,\\j \neq i}} s_j \mid s_i=x,Y_k=y] \\&=\frac{x}{n}+\E[\frac{1}{n} \sum_{j\in[n-1]} Y_j \mid s_i=x,Y_k=y]  =\frac{x}{n}+\frac{y}{n}+\E\!\left[\tfrac{1}{n} \sum_{\substack{j\in[n-1],\\j\neq k}} Y_j \,\middle|\, s_i=x,Y_k=y\right] \\
    &=\frac{x}{n}+\frac{y}{n}+\frac{1}{n}\sum_{j=1}^{k-1} \E[  Y_j \mid s_i=x,Y_k=y]\quad{}+\frac{1}{n}\sum_{j=k+1}^{n-1} \E[  Y_j \mid s_i=x,Y_k=y]
\end{align*}
\endgroup
}

Note that the previous decomposition is similar to the equilibrium bid in an English auction once $k$ bidders have dropped out, as derived in \citet{Goeree-Offerman-2003}. However, we offer a careful derivation in the multi-unit setting of our model. We now use Theorem $2.4.1$ and Theorem $2.4.2$ from \citet{Arnold-Balakrishnan-Nagaraja-2008} on the conditional distribution of order statistics. They state that, for $j<k$, the distribution of $Y_j$ given $Y_k=y$ is the same as the distribution of the $j$-th order statistic of $k-1$ independent samples of the original distribution left-truncated at $y$, and we denote $Z^l_j$ a random variable drawn according to this distribution. Hence, for $j<k$, $\E[Y_j \mid Y_k=y]=\E[Z^l_j]$. Similarly for $j>k$ we have that the distribution of $Y_j$ given $Y_k=y$ is the same as the distribution of the $j-k$-th order statistic of $n-k-1$ independent samples of the original distribution right-truncated at $y$, and we denote by $Z^r_j$ a random variable drawn according to this distribution. Hence, for $j>k$, 
$\E[Y_j \mid Y_k=y]=\E[Z^r_j]$. 
Notice that summing all order statistics drawn from some samples recovers the sum of original samples. Thus we obtain
{\small
\begin{align*}
    \sum_{j=1}^{k-1} \E[Y_j \mid s_i=x,Y_k=y] &= \sum_{j=1}^{k-1} \E[Z^l_j] = \E\!\left[\sum_{j=1}^{k-1} Z^l_j\right] = \E\!\left[ \sum_{j=1}^{k-1} s_j \,\middle|\, \forall j \in [k-1],\, s_j \geq y \right] =\sum_{j=1}^{k-1} \E[s_j \mid s_j \geq y].
\end{align*}
}
The same can be done for the $Z^r_j$. Finally, the $s_j$ are i.i.d. and thus have identical conditional expectations. We obtain
\begin{align}
    \widetilde V(x,y) &= \frac{x}{n}+\frac{y}{n}+\frac{n-k-1}{n}\E[s_j \mid s_j \leq y]+\frac{k-1}{n}\E[s_j \mid s_j \geq y]   \\
    &= \frac{x}{n}+\frac{y}{n}+\frac{n-k-1}{n} \frac{\int_0^{y} tf(t)\diff t}{F(y)}+\frac{k-1}{n}\frac{\int_y^{\vbar} tf(t)\diff t}{1-F(y)}, \label{equ:alternative-v(xy)}
\end{align}
which readily yields a formula for $V(s)=\widetilde V(s,s)$. Since a log-concave density is continuous on the interior of its support, $V$ is well defined and differentiable on the open support of $F$.

We now prove that $V'(s)\leq 1$. First, consider the derivatives of the two ratios with an integral in the numerator in \cref{equ:alternative-v(xy)}. Integration by parts yields
\begin{equation*}
    \frac{\int_0^s tf(t)\diff t}{F(s)}= \frac{\left[ tF(t)\right]_0^s - \int_0^s F(t)\diff t}{F(s)}=s-\frac{\int_0^s F(t)\diff t}{F(s)},
\end{equation*}
and using that for positive random variables $\int_0^{\vbar} tf(t)\diff t=\int_0^{\vbar}(1-F(t))\diff t=\E[s_i]<\infty$, which guarantees convergence of the integral, we have that
\begin{align*}
    \frac{\int_s^{\vbar} tf(t)\diff t}{1-F(s)}&=\frac{\E[s_i]-\int_0^s tf(t)\diff t}{1-F(s)}=\frac{\int_0^{\vbar}(1-F(t))\diff t-sF(s) + \int_0^s F(t)\diff t}{1-F(s)}\\
    &=\frac{\int_0^{\vbar}(1-F(t))\diff t+s(1-F(s)) -s + \int_0^s F(t)\diff t}{1-F(s)}=s+\frac{\int_s^{\vbar}(1-F(t))\diff t}{1-F(s)}
\end{align*}
Now, taking derivatives, we have
\begin{align*}
    \pd{}{s} \frac{\int_0^s tf(t)\diff t}{F(s)}=1-\frac{F(s)^2 - f(s)\int_0^s F(t)\diff t}{F(s)^2}=\frac{f(s)\int_0^s F(t)\diff t}{F(s)^2}.
\end{align*}
By a similar argument as in the proof of \cref{lem:private-derivative-bounded}, using log-concavity of $f$, the above derivative is bounded by $1$. Taking the derivative of the second ratio, we have
\begin{equation}\label{equ:second-ratio}
    \pd{}{s} \left( s+\frac{\int_s^{\vbar}(1-F(t))\diff t}{1-F(s)} \right)= 1+ \frac{-(1-F(s))^2 + f(s)\int_s^{\vbar}(1-F(t))\diff t }{(1-F(s))^2}=\frac{f(s)\int_s^{\vbar}(1-F(t))\diff t }{(1-F(s))^2}.
\end{equation}
The second derivative of $\log(\int_s^{\vbar}(1-F(t))\diff t)$:
\begin{equation}\label{equ:second-derivative-second-ratio}
    \frac{\partial^2}{(\partial s)^2} \log(\int_s^{\vbar}(1-F(t))\diff t) = \pd{}{s} \frac{-(1-F(s))}{\int_s^{\vbar}(1-F(t))\diff t} = \frac{f(s)\int_s^{\vbar}(1-F(t))\diff t-(1-F(s))^2}{\left(\int_s^{\vbar}(1-F(t))\diff t\right)^2}.
\end{equation}
\cref{equ:second-derivative-second-ratio} is non-positive iff $f(s)\int_s^{\vbar}(1-F(t))\diff t/(1-F(s))^2 \leq 1$. This means that the log-concavity of $\int_s^{\vbar}(1-F(t))\diff t$ is equivalent to \cref{equ:second-ratio} being smaller than $1$. As the log-concavity of $\int_s^{\vbar}(1-F(t))\diff t$ follows from the log-concavity of $f$ and $(1-F)$ \citep[Theorem 3]{Bagnoli-Bergstrom-2005}, $f(s)\int_s^{\vbar}(1-F(t))\diff t/(1-F(s))^2 \leq 1$ is implied. 
Finally, using the above derivatives it is clear that $V'(s)>0$, and
\begin{equation*}
    V'(s) \leq \frac{2}{n} + \frac{n-k-1}{n}\cdot 1 + \frac{k-1}{n} \cdot 1 = 1. \qedhere
\end{equation*}
\end{proof2}

\bigbreak

\begin{proof2}[Proof of \cref{lem:private-derivative-bounded}]\phantomsection\label{proof:lem:private-derivative-bounded}
    The logarithm of $\int_0^s G^{1/\delta}(y)\diff y$ has a locally absolutely continuous derivative on $\osupport$, since $G$ is absolutely continuous and positive there. Its concavity is therefore equivalent to a non-positive second derivative almost everywhere. We compute
    \begin{align*}
        \frac{\partial^2 \log\left( \int^{s}_0 G^{\frac{1}{\delta}}(y)\diff y \right)}{(\partial s)^2 } &= \frac{\partial}{\partial s} \left( \frac{ G^{\frac{1}{\delta}}(s)}{\int^{s}_0 G^{\frac{1}{\delta}}(y)\diff y } \right)=\frac{ \frac{1}{\delta} g(s) G^{\frac{1}{\delta}-1}(s) \int^{s}_0 G^{\frac{1}{\delta}}(y)\diff y - G^{\frac{2}{\delta}}(s)}{\left(\int^{s}_0 G^{\frac{1}{\delta}}(y)\diff y \right)^2} \\
        &= \frac{G^{\frac{1}{\delta}-1}(s)}{\left(\int^{s}_0 G^{\frac{1}{\delta}}(y)\diff y \right)^2}\left( \frac{1}{\delta} g(s) \int^{s}_0 G^{\frac{1}{\delta}}(y)\diff y - G^{\frac{1}{\delta} +1}(s) \right). 
    \end{align*}
    Notice that the left-hand fraction is always positive. Hence log-concavity of $\int^{s}_0 G^{\frac{1}{\delta}}(y)\diff y $ is equivalent to $\frac{1}{\delta} g(s) \int^{s}_0 G^{\frac{1}{\delta}}(y)\diff y - G^{\frac{1}{\delta} +1}(s)$ being non-positive almost everywhere on $\osupport$. The latter is equivalent to
    \begin{equation*}
        1 \geq \frac{g(s) \int^{s}_0 G^{\frac{1}{\delta}}(y)\diff y}{\delta G^{\frac{1}{\delta} +1}(s) } = \frac{\partial \bid(s)}{\partial s}. \qedhere
    \end{equation*}
\end{proof2}

\bigbreak

\begin{proof2}[Proof of \cref{lem:intG-log-concave}]\phantomsection\label{proof:lem:intG-log-concave}
    To prove \cref{lem:intG-log-concave}, we will use properties of log-concave distributions from \citet{Bagnoli-Bergstrom-2005}. Namely their Theorems $1$ and $3$ state together that log-concavity of a density $f$ implies log-concavity of the corresponding CDF $F$ and of the complementary CDF $1-F$, and that log-concavity of $F$ or $1-F$ imply log-concavity of respectively $\int_0^s F$ or $\int_s^{\vbar}(1-F)$, where $\vbar$ is the upper limit of the support of $f$ (either a constant or $+\infty$). Additionally, we also have that the product of two log-concave functions is log-concave also.
    Using the above properties, we have that $F$ and $1-F$ are log-concave.
    Moreover, alternative expressions for the order statistics are given, e.g.,~in \citet{Fisz-1965}.
        $G_{m}^{n}(s) =
        \frac{n!}{(n-m)!(m-1)!}\int_0^{F(s)} t^{n-m} (1-t)^{m-1} \diff t$ and  
        $g_{m}^{n}(s) =
        \frac{n!}{(n-m)!(m-1)!} F(s)^{n-m} (1-F(s))^{m-1}f(s)$.
    Thus, the order statistics density $g$  is a product of $F$, $1-F$, and $f$, and $g$ as well as the corresponding CDF $G$ are also log-concave. 
    Furthermore, $G^{\frac{1}{\delta}}$ is log-concave because $\log(G^{\frac{1}{\delta}})=\frac{1}{\delta} \cdot \log(G)$. Finally, we remark that $G^{\frac{1}{\delta}}$ is right-continuous non-decreasing by composition with $x\mapsto x^{\frac{1}{\delta}}$, which is continuous non-decreasing, and $G^{\frac{1}{\delta}}(0)=0$, as well as $G^{\frac{1}{\delta}}(\vbar)=1$ (if $\vbar=\infty$, the equality is understood as a limit). Therefore $G^{\frac{1}{\delta}}$ is a CDF, and applying one last time \citet{Bagnoli-Bergstrom-2005}, we obtain that $\int_0^s G^{\frac{1}{\delta}}$ is log-concave. 
\end{proof2}

\section{\texorpdfstring{Proofs for \cref{sec:surplus-equitable-mechanism}}{Proofs for Surplus-Equitable Mechanisms}}\label{app:sec:proof-for-surplus-equitable-mechanism}

\begin{proof2}[Proof of \cref{theorem:surplus-equity}]\phantomsection\label{proof:theorem:surplus-equity}
Under the efficient allocation, bidder $i$ receives an item iff $s_i>y$ ($y:=\kth{k}(s_{-i})$). Let $U_i(s_i)$ denote bidder $i$'s interim utility under truthful reporting. Since the allocation rule is monotone, Bayesian-Nash incentive compatibility is equivalent to the envelope condition
$U_i'(s_i)=\E\!\left[\partial_{s_i}v_i(\ssb)\,\ind{s_i>y}\mid s_i\right]$.
Since $\partial_{s_i}v_i(\ssb)=1-c+\frac{c}{n}$ and $\Prob(s_i>y\mid s_i)=G(s_i)$,
$U_i'(s_i)=\Bigl(1-c+\frac{c}{n}\Bigr)G(s_i)$ follows.
Integrating and normalizing $U_i(0)=0$ yields
\[
U_i(s_i)=\Bigl(1-c+\frac{c}{n}\Bigr)\int_0^{s_i} G(t)\,dt.
\]

Turning to surplus equity, notice that if bidders $i$ and $j$ both win at some profile $\ssb$, then
$v_i(\ssb)-v_j(\ssb)=(1-c)(s_i-s_j)$.
Hence a winners-pay mechanism equalizes the winners' realized utilities if and only if there exists a transfer term $\tau$ common to all winners such that
$p_i(\ssb,X_i(\ssb))=\bigl((1-c)s_i+\tau(\ssb)\bigr)\ind{s_i>y}$.
In that case, every winner's realized utility is
\begin{equation}\label{equ:winners-equal-util}
u_i(\ssb)=\frac{c}{n}\sum_{j=1}^n s_j-\tau(\ssb).
\end{equation}
To prove existence, it is therefore enough to construct such a common term. A convenient first construction is obtained by prescribing a common realized utility $\overline U(y)$ for all winners, where $y$ is the first rejected signal. The corresponding winners-pay payment rule is
\[
p_i(\ssb,X_i(\ssb))=\Bigl((1-c)s_i+\frac{c}{n}\sum_{j=1}^n s_j-\overline U(y)\Bigr)\ind{s_i>y}.
\]
Under truthfulness,
$U_i(s_i)=\int_0^{s_i}\overline U(y)g(y)\,dy$.
Differentiating and comparing with the envelope condition yields
$\overline U(y)=\left(1-c+\frac{c}{n}\right)\frac{G(y)}{g(y)}$.
Thus an admissible ex post-equitable payment rule is
\[
p_i(\ssb,X_i(\ssb))=\Bigl((1-c)s_i+\frac{c}{n}\sum_{j=1}^n s_j-\Bigl(1-c+\frac{c}{n}\Bigr)\frac{G(y)}{g(y)}\Bigr)\ind{s_i>y}.
\]
This equalizes realized utilities among winners, but depends on the full profile through $\sum_{j=1}^n s_j$.

\paragraph{Implementation using only the first rejected signal.}
We now show that the same interim utility and ex post surplus equality can be implemented with payments that depend only on the bidder's own signal and the first rejected signal, that is, with a rule of the form
$p_i(\ssb,X_i(\ssb))=\bigl((1-c)s_i+\tau(y)\bigr)\ind{s_i>y}$. Under truthful reporting,
\[
U_i(s_i)=\int_0^{s_i}\Bigl(\widetilde V(s_i,y)-(1-c)s_i-\tau(y)\Bigr)g(y)\,dy.
\]
By definition of $V$,
\[
V(y)
=
(1-c)y+\frac{c}{n}\Bigl(y+\E\!\Bigl[\sum_{j\neq i}s_j\,\Bigm|\,\kth{k}(s_{-i})=y\Bigr]\Bigr),
\]
and therefore, for all $x,y$,
\[
\widetilde V(x,y)=V(y)+(1-c+\frac{c}{n})(x-y).
\]
Indeed, conditioning on $\kth{k}(s_{-i})=y$ only restricts the other bidders' signals, so the conditional expectation of $\sum_{j\neq i}s_j$ does not depend on $x$. Hence
\[
\widetilde V(x,y)
=(1-c)x+\frac{c}{n}\Bigl(x+\E\!\Bigl[\sum_{j\neq i}s_j\,\Bigm|\,\kth{k}(s_{-i})=y\Bigr]\Bigr)
=V(y)+\Bigl(1-c+\frac{c}{n}\Bigr)(x-y).
\]
Differentiating $U_i(s_i)$ via the Leibniz rule and using the previous identity gives 
\begin{align*}
U_i'(s_i)&=\frac{d}{ds_i}\int_0^{s_i} (V(y) - (1-c) y + \frac{c}{n}(s_i-y)-\tau(y)) g(y)\mathrm dy\\
&= (V(s_i) - (1-c) s_i -\tau(s_i)) g(s_i)+\int_0^{s_i} \frac{d}{ds_i} (V(y) - (1-c) y + \frac{c}{n}(s_i-y)-\tau(y)) g(y)\mathrm dy\\
&= (V(s_i) - (1-c) s_i + \frac{c}{n}(s_i-s_i)-\tau(s_i)) g(s_i)+\int_0^{s_i} \frac{c}{n} g(y)\mathrm dy\\
&=\bigl(V(s_i)-(1-c)s_i-\tau(s_i)\bigr)g(s_i)+\frac{c}{n}G(s_i).
\end{align*}
Comparing this with the envelope condition yields $\tau(y)=V(y)-(1-c)\left(y+\frac{G(y)}{g(y)}\right)$.

Substituting into the candidate payment rule gives \cref{equ:equitable-payment}. Because the allocation rule is monotone and the truthful interim utility satisfies the envelope expression above, BNIC follows from the standard characterization. The normalization $U_i(0)=0$ implies interim individual rationality. Ex post surplus equity among winners follows immediately from \cref{equ:winners-equal-util}.
\end{proof2}

\bigbreak

\begin{proof2lower}[Proof of \cref{prop:alloc_fair}]\phantomsection\label{proof:prop:alloc_fair}
Consider first the requirement of dominant-strategy incentive compatibility. Since the allocation is a uniform lottery, bidder $i$'s report does not affect the probability of winning. Weak dominant-strategy incentive compatibility therefore requires that, conditional on winning, bidder $i$'s payment does not depend on~$s_i$. Hence any admissible payment rule must be of the form $p_i(s,X_i(s)) = \tau_i(s_{-i})\,\ind{\text{$i$ wins an item}}$, for some function $\tau_i$.

If bidder $i$ wins, their realized utility is $u_i(s) = (1-c)\,s_i + \frac{c}{n}\sum_{j=1}^n s_j - \tau_i(s_{-i})$. Equalization of winners' surpluses requires this expression to be identical for every pair of winners $i,j$ and every signal profile~$s$. Taking the difference for two winners $i$ and $j$, $(1-c)(s_i - s_j) = \tau_i(s_{-i}) - \tau_j(s_{-j})$. This implies that for all $i$, $j$, and $s$, the functions $\phi_i(s_{-i})=\tau_i(s_{-i})+(1-c) \sum_{\ell \neq i} s_\ell$ satisfy $\phi_i(s_{-i})=\phi_j(s_{-j})$, so neither function can depend on $s_i$ or $s_j$. For a fixed $i$, applying this for every $j\neq i$ implies that $\phi_i(s_{-i})$ does not depend on any $s_j$, and must be a constant $C_i$. Substituting back, the equality $C_i = C_j$ must hold for all $i,j$, so by symmetry of the uniform lottery, there is a single constant~$C$ with $\tau_i(s_{-i}) = C - (1-c)\sum_{j \neq i} s_j$. Next, we require ex ante budget balance. Taking expectations and using independence of signals,
\[
\E\!\left[\sum_{i=1}^n p_i\right] = \sum_{i=1}^n \frac{k}{n}\!\left(C - (1-c)(n-1)\E[s_1]\right) = k\!\left(C - (1-c)(n-1)\E[s_1]\right).
\]
Setting this to zero yields $C = (1-c)(n-1)\E[s_1]$. Substituting back gives
\[
p_i(s,X_i(s)) = (1-c)\!((n-1)\E[s_1] - \sum_{j \neq i} s_j)\ind{\text{$i$ wins an item}}.
\]
Finally, consider individual rationality. A winner's ex post utility is $(1-c+\frac{c}{n})\sum_{j=1}^n s_j - C$, so ex post IR requires $C \leq 0$. Setting $C=0$ guarantees this, but creates an expected budget deficit of $k(1-c)(n-1)\E[s_1]$. Regarding interim IR, the interim utility, given that the bidder wins an item with probability $k/n$, and using $C=(1-c)(n-1)\E[s_1]$ for budget balance, yields
\begin{equation*}
\frac{k}{n} \left( (1-c+\frac{c}{n}){(s_i+(n-1)\E[s_1])}- C\right)
=
\frac{k}{n} \left( (1-c+\frac{c}{n})s_i+\frac{c}{n}(n-1)\E[s_1]\right) \geq 0.
\qedhere
\end{equation*}
\end{proof2lower}

\bigbreak

\begin{proof2}[Proof of \cref{theorem:surplus-equitable-monotone-alloc}]\phantomsection\label{proof:theorem:surplus-equitable-monotone-alloc}
By definition, the mechanism equalizes surplus ex post. We now prove BNIC.
Let $X_i(z,s_{-i}) \in \{0,1\}$ denote the realized item allocation under the randomization defined above, so that by independence $\E[X_i(z,s_{-i})\mid s_i]=\E_{s_{-i}}[x_i(z,s_{-i})]=q_i(z)$. Take bidder \(i\) with true signal \(s_i\) who reports \(z\). Under the payment rule:
{\small
\begin{align*}
u_i\big(s_i,s_{-i},z\big)
&= X_i(z,s_{-i})v_i(s_i,s_{-i})-p_i\big((z,s_{-i}),X_i(z,s_{-i})\big) \\
& = X_i(z,s_{-i})v_i(s_i,s_{-i})-\Big(X_i(z,s_{-i})v_i(z,s_{-i})-{\bar u(z,s_{-i})}\Big)\\
&= X_i(z,s_{-i})\big(v_i(s_i,s_{-i})-v_i(z,s_{-i})\big)+{\bar u(z,s_{-i})} \\
&= X_i(z,s_{-i}) (1-c+\tfrac{c}{n})(s_i-z)+{\bar u(z,s_{-i})}.
\end{align*}
}
Suppose that other bidders are truthful, so that $\E[X_i(z,s_{-i})\mid s_i]=q_i(z)$. Then,
\[
\E[u_i(s_i,s_{-i},z) \mid s_{i}]= (1-c+\tfrac{c}{n})(s_i-z)q_i(z)+\E[{\bar u(z,s_{-i})} \mid s_{i}].
\]

{Recall that $A_i(z)=(1-c+\tfrac{c}{n})\int_0^z q_i(t)\,dt-(1-c+\tfrac{c}{n})\E_{s_i}\!\left[\int_0^{s_i} q_i(t)\,dt\right]$. The expected mean-allocation advantage is zero under truthful reporting, so $\E[A_i(s_i)]=0$, and for $j \neq i$, truthful reporting by bidder $j$ and independence of signals imply $\E[A_j(s_j)\mid s_i]=0$.} Hence $\E[{\bar u(z,s_{-i})} \mid s_{i}]=A_i(z)+(1-c+\tfrac{c}{n}) \max_{j \in [n]} \E_{s_j}\!\left[\int_0^{s_j} q_j(t)\,dt\right]$. From this, the interim advantage of truthful reporting is
{\small
\begingroup
\allowdisplaybreaks
\begin{align*}
\E[u_i(s_i,s_{-i},s_i) \mid s_{i}] -\E[u_i(s_i,s_{-i},z) \mid s_{i}] & = A_i(s_i)-A_i(z) -(1-c+\tfrac{c}{n})(s_i-z) q_i(z) \\
&=(1-c+\tfrac{c}{n})\int_z^{s_i} q_i(y)\,dy -(1-c+\tfrac{c}{n})(s_i-z)q_i(z)\\
&=(1-c+\tfrac{c}{n})\int_z^{s_i} (q_i(y)-q_i(z))\,dy.
\end{align*}
\endgroup
}
Because $x_i$ is monotone, so is $q_i$. If $s_i\geq z$, then $q_i(y)-q_i(z)\geq0$ on $[z,s_i]$, so the integral is non-negative. If $s_i<z$, then $q_i(y)-q_i(z)\leq0$ on $[s_i,z]$, but the limits of integration are reversed, so the integral is again non-negative.  The interim utility from truthful reporting is weakly higher than from any other report; the mechanism is BNIC.

Considering interim individual rationality, note that the interim utility is given by
{\small
\[
A_i(s_i)+ \Bigl(1-c+\tfrac{c}{n}\Bigr) \max_{j \in [n]} \E_{s_j}\!\left[\int_0^{s_j} q_j(t)\,dt\right].
\]
}
By non-negativity of $q_i$ (and hence $A_i$ non-decreasing), the worst case is at $s_i=0$ with interim utility
{\small
\begin{align*}
&A_i(0) + (1-c+\tfrac{c}{n}) \max_{j \in [n]} \E_{s_j}\!\left[\int_0^{s_j} q_j(t)\,dt\right]  = (1-c+\tfrac{c}{n}) \left( \max_{j \in [n]} \E_{s_j}\!\left[\int_0^{s_j} q_j(t)\,dt\right]-\E_{s_i}\!\left[\int_0^{s_i} q_i(t)\,dt\right] \right),
\end{align*}
}
which is weakly greater than zero. Because the interim utility is always non-negative, this mechanism is interim individually rational.

We now prove that this achieves maximum revenue among all mechanisms that are interim IR, and equalize surplus ex ante. By the revelation principle, any mechanism with allocation rule $x$ can be transformed to a BNIC mechanism with the same revenue; it is sufficient to consider interim IR, BNIC mechanisms that equalize surplus ex ante. Consider any payment rule $\tilde{p}_i$ which achieves BNIC and interim IR. Letting $\tilde{U}_i(t)$ be the interim utility of truthful reporting with signal $t$, the interim utility of reporting $z$ with a private signal $t$ can be expressed as
\begin{align*}
\E[x_i(z,s_{-i})v_i(s_i,s_{-i})-\tilde{p}_i(z,s_{-i}) \mid s_i]&=\tilde{U}_i(z) + \E[x_i(z,s_{-i}) (v_i(s_i,s_{-i})-v_i(z,s_{-i})) \mid s_i]\\
&=\tilde{U}_i(z)+(1-c+\tfrac{c}{n})  (t-z) q_i(z).
\end{align*}
Therefore, the BNIC condition is equivalent to, for all $t$ and $z$, $\tilde{U}_i(t) \geq \tilde{U}_i(z)+ (1-c+\tfrac{c}{n})  (t-z) q_i(z)$. Hence, $\tilde{U}_i(t)$ is convex, absolutely continuous, with derivative $(1-c+\tfrac{c}{n})q_i(t)$ almost everywhere. Hence, for all BNIC mechanisms, it must be that $\tilde{U}_i(s_i)=\tilde{U}_i(0)+(1-c+\tfrac{c}{n}) \int_0^{s_i} q_i(t)\,dt$ (envelope condition). As above, interim individual rationality is equivalent to $\tilde{U}_i(0) \geq 0$, so taking expectation over $s_i$, this is equivalent for all $i \in [n]$ to $\tilde{U}_i(0)=\E[\tilde{u}_i(s_i,s_{-i})]-(1-c+\tfrac{c}{n}) \E_{s_i}\!\left[ \int_0^{s_i} q_i(t)\,dt\right] \geq 0$, i.e.\ $\E[\tilde{u}_i(s_i,s_{-i})] \geq (1-c+\tfrac{c}{n}) \E_{s_i}\!\left[ \int_0^{s_i} q_i(t)\,dt\right]$. Now, given that the ex ante surplus is equalized among all bidders, $\E[\tilde{u}_i(s_i,s_{-i})]$ does not depend on $i$, and the interim IR condition is equivalent to $\E[\tilde{u}_i(s)] \geq (1-c+\tfrac{c}{n})\max_{j \in [n]} \E_{s_j}\!\left[\int_0^{s_j} q_j(t)\,dt\right]$.
The expected revenue of this mechanism is
\begin{align*}
\E\!\left[\sum_{i \in [n]} \tilde{p}_i(s_i,s_{-i})\right]&=\sum_{i \in [n]}\left( \E[ x_i(s_i,s_{-i}) v_i(s_i,s_{-i})]   -\E[ \tilde{u}_i(s_i,s_{-i})] \right ) \\
&\leq \sum_{i \in [n]}\E[ x_i(s_i,s_{-i}) v_i(s_i,s_{-i})] - n (1-c+\tfrac{c}{n}) \max_{j \in [n]} \E_{s_j}\!\left[ \int_0^{s_j} q_j(t)\,dt\right],
\end{align*}
which is exactly the revenue of the mechanism provided, as $\E[A_i(s_i)]=0$.
\end{proof2}

\bigbreak

\begin{proof2}[Proof of \cref{cor:alloc_fair_all}]\phantomsection\label{proof:cor:alloc_fair_all}
It is sufficient to apply the previous mechanism to the allocation rule of the uniform lottery. Setting $q_i(z)=k/n$, we obtain
{\small
\begin{align*}
{\bar u(\ssb)}&=\sum_{i=1}^n A_i(s_i)+(1-c+\frac{c}{n}) \, \E[\int_0^{s_i} q_i(t)dt ]=\sum_{i =1}^n(1-c+\frac{c}{n}) \frac{k}{n} (s_i-\E[s_1])+(1-c+\frac{c}{n})\frac{k}{n} \E[s_1] \\
&= \left(1-c+\frac{c}{n} \right)\frac{k}{n} \left( \sum_{i=1}^n s_i -(n-1)\E[s_1] \right)
\end{align*}
}
The payment is $p_i(\ssb,X_i(\ssb))= \ind{\text{$i$ wins an item}} v_i(\ssb) - {\bar u(\ssb)}$, and thus the expected total payment
{\small
\begin{align*}
\E\left[\sum_{i=1}^n p_i(\ssb,X_i(\ssb))\right]&=\sum_{i =1}^n \left(   \frac{k}{n}((1-c)\E[s_1]+\frac{c}{n} n\E [s_1]) -  \frac{k}{n} (1-c+\frac{c}{n}) ( n \E[s_1]-(n-1)\E[s_1] )\right)\\
&= n \cdot \frac{k}{n} (1-(1-c+\frac{c}{n}))\E[s_1]= ck \frac{n-1}{n} \E[s_1] \geq 0.
\end{align*}
}
This shows that the mechanism is ex ante weakly-budget-balanced, with strictly positive expected revenue whenever $c>0$.
\end{proof2}

\clearpage
\begingroup
\section*{Supplementary Appendix}

\appendix
\numberwithin{equation}{section}\renewcommand{\theequation}{\thesection\arabic{equation}}
\numberwithin{figure}{section}\renewcommand{\thefigure}{\thesection\arabic{figure}}
\numberwithin{table}{section}\renewcommand{\thetable}{\thesection\arabic{table}}
\numberwithin{theorem}{section}\renewcommand{\thetheorem}{\thesection\arabic{theorem}}
\numberwithin{corollary}{section}\renewcommand{\thecorollary}{\thesection\arabic{corollary}}
\numberwithin{definition}{section}\renewcommand{\thedefinition}{\thesection\arabic{definition}}
\numberwithin{example}{section}\renewcommand{\theexample}{\thesection\arabic{example}}
\numberwithin{assumption}{section}\renewcommand{\theassumption}{\thesection\arabic{assumption}}
\numberwithin{lemma}{section}\renewcommand{\thelemma}{\thesection\arabic{lemma}}
\numberwithin{proposition}{section}\renewcommand{\theproposition}{\thesection\arabic{proposition}}

\renewcommand*{\theHsection}{%
  supp.\Alph{section}}
\renewcommand*{\theHsubsection}{%
  \theHsection.\arabic{subsection}}
\renewcommand*{\theHsubsubsection}{%
  \theHsubsection.\arabic{subsubsection}}

\makeatletter
\@for\supplcounter:=equation,figure,table,%
theorem,corollary,definition,example,%
assumption,lemma,%
proposition\do{%
  \expandafter\edef
    \csname theH\supplcounter\endcsname{%
      \noexpand\theHsection.%
      \noexpand\arabic{\supplcounter}}%
}
\makeatother
\renewcommand{\thesubsection}{\thesection.\arabic{subsection}}
\crefalias{section}{supplappendix}
\crefalias{subsection}{supplappendix}

This supplementary appendix contains material referenced from the main paper: extensions and robustness checks, additional proofs used in the equilibrium analysis, and further examples and figures for the numerical and illustrative results. 

\section{Extensions}\label{sec:extensions}

In this section, we discuss several extensions of our results. First, we explain how surplus equity and, in particular, the \wevtext relates to the within-bidder variation of surplus, and we discuss the empirical variance of surplus between all bidders. Next, we elaborate on why the regularity assumption of log-concavity is necessary for our argument about mixed auctions. Furthermore, we discuss what happens if unit demand is generalized to multi-unit demand. Finally, we also consider how our results extend to auctions with reserve prices.

\subsection{Variance and Risk Preferences}\label{sec:risk-aversion}

Surplus equity and distributional concerns are distinct from questions of within-agent variation and associated risk preferences. An appropriate measure to assess the latter is, e.g.,~the ex ante variance of bidder surplus. While the two notions are distinct, the measures are linked. For the variance identities below, we assume that bidder surpluses have finite second moments.
\begin{lemma}\label{lem:alternative-variance}\hyperref[proof:lem:alternative-variance]{\normalfont[Proof in \cref{proof:lem:alternative-variance}]}
    An equivalent expression for the winners' empirical variance is given by 
    \[\wev~= \var[u_1| \text{$1$ wins}] - \text{\normalfont Cov}[u_1,u_2| \text{$1$ and $2$ win}].\]
\end{lemma}
The within-agent variance addresses a bidder's individual risk-attitude, e.g.,~across a series of identical, repeated auctions, in which a given bidder redraws their signal in every auction. Its study goes back to \citet{Vickrey-1961}.\footnote{In the appendix of his famous article, Vickrey showed that, in a \emph{single-unit auction}, the ex ante variance of surplus is lower under the first-price than the second-price rule, given uniform distributions of private values. See also \citet{Krishna-2009}.}
In contrast, an equity measure must take into account the correlation of surpluses between bidders. 
With pure private values and thus ex post individual rationality, rankings of auction formats in terms of ex ante variance or winners' ex ante variance are identical. Rankings with respect to the empirical variance, however, may differ depending on whether one considers only winners or all bidders. The distinction can be seen in \Cref{lem:connecting-variances} below.

In addition, for the pure private value setting, the pay-as-bid auction also minimizes the within-bidder variation of surplus, formally stated in the following proposition.
\begin{proposition}\label{prop:ex-ante-variance}\hyperref[proof:ex-ante-variance]{\normalfont[Proof in \cref{proof:ex-ante-variance}]}
    With pure private values ($c=0$), the pay-as-bid auction minimizes the ex ante variance of surplus among all standard, winners-pay auctions with increasing symmetric equilibrium bid functions. 
\end{proposition}
Because of revenue equivalence, note that \cref{prop:ex-ante-variance} also implies that $\E[u_i^2]$ is minimal in the pay-as-bid auction among standard, winners-pay auctions.
The second moment of surplus, $\E_{\ssb}[u_1^2]$, links the winners' empirical variance (\wev) and the empirical variance among all bidders ($\ev$), while the squared first moment, $\E_{\ssb}[u_1]^2$, links the ex ante variance of surplus conditional on winning ($\wvar$) and the ex ante variance ($\var$). This is shown in the next lemma.
\begin{lemma}\label{lem:connecting-variances}\hyperref[proof:connecting-variances]{\normalfont[Proof in \cref{proof:connecting-variances}]}
    Assuming that the auction is a \emph{winners-pay auction}, the empirical variance and the ex ante variance can be decomposed, respectively, as
    $\ev = \frac{k(k-1)}{n(n-1)}\cdot\wev~+ \left(1-\frac{k-1}{n-1}\right)\E_{\ssb}[u_1^2]$ and
    $\var = \frac{k}{n}\cdot\wvar + \left(\frac{n}{k}-1\right)\cdot \E_{\ssb}[u_1]^2$.
\end{lemma}

As a consequence, surplus equity rankings with respect to the winners' empirical variance and the empirical variance among all bidders may not be equivalent. However, using \cref{prop:ex-ante-variance,lem:connecting-variances,prop:delta-dominating-uniform}, we derive the following corollary.
\begin{corollary}\label{cor:empvar-among-all-bidders}\hyperref[proof:empvar-among-all-bidders]{\normalfont[Proof in \cref{proof:empvar-among-all-bidders}]}
    Assuming pure private values ($c=0$), suppose that the pay-as-bid equilibrium bid $\bid[1]$ satisfies $\pd{\bid[1]}{s} \leq 2$ for almost every $s\in\osupport$. Then, the empirical variance (among all bidders) is weakly lower under pay-as-bid pricing than under uniform pricing.
\end{corollary}
Although this result shows that we can extend equity rankings under pure private values to the empirical variance \emph{among all bidders}, this may not hold in the general case.

\subsection{Beyond Log-Concave Distributions}\label{sec:beyond-log-concave}

A crucial ingredient for \cref{theorem:bound-on-optimal-delta} is that the derivative of the equilibrium bid function is bounded by $1$, which holds for log-concave distributions by \cref{prop:derivative-beta-bounded}. Our proof uses log-concavity of the density of the first rejected signal. In the following, we provide some insights as to why it is difficult to generalize this result beyond log-concave distributions.

For the $(A,G)$ and $(A,H)$ equivalences, consider pure private values ($c=0$) and pay-as-bid pricing ($\delta=1$). Derivative inequalities are understood almost everywhere on the open support. Considering log-concave signal distributions, we note that log-concavity is equivalent to $(A,G)$ concavity (a generalization of convexity, see \citet{Anderson-Vamanamurthy-Vuorinen-2007}), and $\frac{\partial \bid}{\partial s} \leq 1$ is thus equivalent to $(A,G)$ concavity of $s \mapsto \int_0^s G_k^{n-1}$.
One idea to extend our results could then be to consider other generalizations of convexity. Considering \cref{prop:delta-dominating-uniform}, one might attempt to bound the slope of the bid functions by $2$. It holds that  $\pd{\bid}{s} \leq 2$ is equivalent to $(A,H)$ concavity of the same function where $H$ is the harmonic mean. But contrary to $(A,G)$ concave functions, there are no simple group closure properties that allow for the $(A,H)$ concavity of $f$ to always imply that of $\int_0^s G_k^{n-1}$. Thus, this route of inquiry is not fruitful.

For $c<1$, conditions similar to MEU (\cref{def:monotone-ex-post-utility}) such that uniform pricing yields lower \PD~(or \wev) than pay-as-bid pricing are much more difficult to attain. Why? If we follow the same main ideas as in the proof of \cref{prop:delta-dominating-uniform}, a similar condition using the mean value theorem would be that, for all $s \in \osupport$, $\varphi$ is an expansive mapping, translating into $\vert (1-c)-\pd{\bid}{s} \vert \geq 1-c$. As $\pd{\bid}{s}$ is strictly positive, this forces $\pd{\bid}{s} \geq 2(1-c)$. For signal distributions with bounded density, $\gnk$ is close to zero near $\vbar$ (this follows from the definition of order statistics), and therefore $\pd{\bid}{s}$ is close to zero for a nonzero interval of signals. Thus, $\pd{\bid}{s} \geq 2(1-c)$ cannot hold for all signals on the support, and we cannot rely on similar proof techniques to produce the desired conditions.

\subsection{Multi-Unit Demand}\label{sec:multi-unit-demand}

A natural question is how robust the equity dominance results for pay-as-bid, uniform, and mixed auctions are with respect to the assumption of unit demand. In this section, we discuss a simple generalization of our model to flat $d$-unit demands for some positive integer $d$, where all bidders have the same constant marginal value for the first $d$ items obtained and $0$ after.

There are two main considerations in extending our results to this setting, both leading to potential equity-efficiency trade-offs. Firstly, each bidder may win a different number of items, and thus equity could be studied with respect to a bidder's total surplus or average per-unit surplus. The former equity metric interacts with efficiency. Secondly, we have a multiplicity of equilibria in uniform-price auctions: \citet{Noussair-1995} shows that even for two bidders and two-unit demand, there is no unique symmetric equilibrium; instead, there exists an inefficient, zero-revenue symmetric equilibrium, as well as an efficient, zero-utility symmetric equilibrium.

Uniform-price auctions that allocate $d$ items to the same winners yield the same pairwise surplus differences. \citet{Ausubel-et-al-2014} show in their Theorem $1$ that a necessary condition for the existence of an ex post efficient equilibrium in uniform-price auctions is given by the condition ``demand divides supply'', which requires the total supply of items to be an integer multiple of a homogeneous demand $d$, i.e.,~$k/d$ to be an integer. Without this condition, for $d>1$, Vickrey pricing does not coincide with the FRB uniform-price auction.

Under Vickrey pricing, a bidder who wins $m$ items must pay the sum of the $m$ highest losing bids, not including their own losing bids. This payment rule is exactly that of the VCG mechanism. With pure private values ($c=0$), truthful bidding is a weakly dominant strategy and the allocation is efficient. When demand divides supply, all winners of the VCG mechanism will receive exactly $d$ items, and each winner pays the same price, since they face the same losing bids. This renders the Vickrey payment into a uniform-pricing rule. Under pay-as-bid pricing with constant marginal values for $d$ items and zero after, the equilibrium bids for these first $d$ items are the same as in the unit-demand model with $n$ bidders and $k/d$ items, and thus each bidder also receives exactly $d$ items (see also \citealt{Ausubel-et-al-2014}).
Suppose $d$ divides $k$ and $2\leq k/d<n$. For pay-as-bid pricing, Vickrey pricing, and their mixtures, consider the symmetric equilibria obtained by repeating the corresponding unit-demand equilibrium bid for $k/d$ items on each bidder's first $d$ items. Each of the $k/d$ winning bidders receives $d$ items, and total payments and surpluses are multiplied by $d$. Thus, the comparisons of surplus equity carry over.

Even when demand divides supply, the number of winners may differ across formats and equilibria, so our winner-based equity order does not directly compare them. The example in \citet{Ausubel-et-al-2014}, with $n=2$, $k=2$, flat two-unit demand, and heterogeneous pure private values, illustrates this issue: an inefficient uniform-price equilibrium allocates one unit to each bidder, whereas the efficient pay-as-bid equilibrium allocates both to the higher-value bidder.

\subsection{Revenue Maximization and Reserve Prices}\label{sec:reserve-prices}

Standard auctions lead to an efficient allocation of items in our model. Moreover, for a large class of probability distributions (and common values), these auctions are also revenue maximizing (see \cite{Bulow-Klemperer-1996}). They show that if, in addition to independent signals, the bidders' marginal revenues are increasing in signals and weakly positive, standard auctions are revenue maximizing.\footnote{With independent signals, the marginal revenue of bidder $i$ given a realization of signals $\ssb$ is given by $\mathrm{MR}_i(\ssb) = \frac{-1}{f(s_i)}\frac{\partial}{\partial s_i}\left(v(s_i,\smi)(1-F(s_i))\right)$.} Marginal revenues are positive in many applications where bidders' values are high, i.e.,~the support of signals is sufficiently positive.\footnote{In our model, without loss of generality, the signal support includes the lower bound zero.} In those cases, it is optimal for the seller to sell their entire supply. If marginal revenues are ordered by signals, our class of auctions maximizes expected revenue among BNIC, interim-IR mechanisms required to allocate all $k$ items. In practice, either (or both) of the two conditions are observed in many high-stakes auctions, e.g.,~for spectrum licenses.

A common design instrument to raise revenue is a reserve price $r>0$: only bids at least $r$ can win, and a winner pays the auction price or $r$, whichever is higher. \citet{Myerson-1981,Riley-Samuelson-1981} characterize the revenue-maximizing reserve for pure private values; with interdependent values the optimal reserve is more intricate, but reserve prices are common in practice. In what follows, $r$ is any exogenous reserve price, and we ask whether our equity rankings are robust to this instrument.

Some of our equity comparisons between pay-as-bid and uniform-price auctions extend to environments with a common, exogenous reserve price. With a reserve price, the number of winners, although identical in different standard auctions with identical reserve prices, is not necessarily equal to the number of items $k$ and becomes a random variable. In the uniform-price auction, the standard derivation leads to an equilibrium bid of $V(s)$ for $s>s_{r}$, where $s_r$ is the participation threshold, defined as the smallest signal at which the expected payoff from bidding $V(s_r)$ weakly exceeds zero (see~\cref{eq:sr_charac} in the proof), and $0$ otherwise. Using revenue equivalence, we obtain the equilibrium bid in the pay-as-bid auction $\beta_{r}^{\delta=1}(s)=(r-V(s_r))G(s_r)/G(s) + V(s) - \int_{s_r}^s V'(y)G(y)\diff y/G(s)$ for $s>s_r$ and $0$ otherwise. For the equity comparison of pay-as-bid and uniform pricing, we establish the following.

\begin{proposition}\label{prop:PAB-dominates-uniform-reserve}\hyperref[proof:reserve_price]{\normalfont[Proof in \cref{proof:reserve_price}]}
    Suppose MEU (\cref{def:monotone-ex-post-utility}) holds in the equilibrium of the pay-as-bid auction without a reserve price. Then the pay-as-bid auction with a given reserve price is equity-preferred to any uniform-price auction with the same reserve price.
\end{proposition}
The proof amounts to showing that MEU carries over from the pay-as-bid auction without a reserve price to the auction with a strictly positive reserve price; the dominance of pay-as-bid over uniform pricing then follows by a direct pairwise comparison among winners (as in \cref{prop:delta-dominating-uniform}), which needs only MEU for the pay-as-bid reserve bid and not the $\delta$-monotonicity behind \cref{prop:PD-decreasing}. 

\section{\texorpdfstring{Proofs for \cref{sec:extensions}}{Proofs for Extensions}}

\begin{proof2}[Proof of \cref{lem:alternative-variance}]\phantomsection\label{proof:lem:alternative-variance}
The empirical variance of surplus can be transformed as follows. 
\begin{align*}
    \E_s\left[ \frac{1}{n-1 } \sum_i^{n} \left( u_i - \frac{1}{n}\sum_j^n u_j \right )^2 \right]
    &= \E_s\left[\frac{1}{2n(n-1)} \sum_{i=1}^{n}\sum_{j=1}^n (u_i - u_j)^2 \right]\\
    &=\frac{\E_s\left[(u_1-u_2)^2\right]}{2}\\
    &= \E_s[u_1^2] - \E_s[u_1 u_2]\\
    &= \var(u_1) - \text{Cov}(u_1, u_2)
\end{align*}
Similarly, the empirical variance conditioned on winning can be written as
\begin{align*}
\E_{s}\left[ \frac{1}{k-1 } \sum_{i=1}^{k} \left( u_i - \frac{1}{k}\sum_{j=1}^k u_j \right )^2 \middle | \text{$1,\dots, k$ win}\right]
&= \frac{\E_s\left[(u_1-u_2)^2\mid\text{$1$ and $2$ win}\right]}{2}\\
&= \E_s\left[u_1^2\mid\text{$1$ wins}\right]
- \E_s\left[u_1u_2\mid\text{$1$ and $2$ win}\right]\\
&= \var(u_1\mid\text{$1$ wins}) - \text{Cov}(u_1, u_2\mid\text{$1$ and $2$ win}).
\end{align*}
All equalities are due to exchangeability. The first reduces the within-winner empirical variance to a single pairwise term (exchangeability of the $k$ winners). For the remaining steps, symmetry implies that the joint law of any fixed set of winners' surpluses is unaffected by which other bidders win; conditioning on $\{1,\dots,k\text{ win}\}$ may thus be replaced, term by term, by the sub-event involving only the bidders that appear. In particular $\E_s[u_1^2\mid 1,\dots,k\text{ win}]=\E_s[u_1^2\mid 1\text{ wins}]$ and $\E_s[u_1\mid 1\text{ wins}]=\E_s[u_1\mid 1\text{ and }2\text{ win}]$.
\end{proof2}

\bigbreak

\begin{proof2}[Proof of \cref{prop:ex-ante-variance}]\phantomsection\label{proof:ex-ante-variance}
    We define the probability that $i$ wins $q_i(s_i) := \Prob_{\smi}[i\text{ wins}]$. Recall that $\bid[1](s_i)$ denotes the equilibrium bid in the pay-as-bid auction. Consider any standard auction, characterized by a payment rule $(p_1(s), \dots, p_n(s))$. Revenue equivalence implies that
    \begin{align}
        q_i(s_i)\cdot \bid[1](s_i) = \E_{\smi}[\bid[1](s_{i})\cdot \ind{i\text{ wins}}] =
        \E_{\smi}[p_i(s)].
    \end{align}
    In particular, note that if $p_i$ is chosen to be the uniform-price rule, this formula can be used to compute $\bid[1](s_i)$. Now define the ex post surplus
    $u_i(s_i, \smi) := s_i\cdot \ind{i\text{ wins}}- p_i(s)$.
    We write
    \begin{align}
        u_i(s_i,\smi)
        &= \underbrace{\ind{i\text{ wins}}\cdot (s_i-\bid[1](s_i))}_{u_i^1(s)}
        + \underbrace{\ind{i\text{ wins}}\cdot \bid[1](s_i)-p_i(s)}_{\Delta(s_i,\smi)}.
    \end{align}
    Now, observe that by revenue equivalence we have $\E_{\smi}[\Delta(s_i, \smi)] = 0$ for all $s_i$. We write
    \begin{align}
        \E_{\smi}[u_i(s_i,\smi)]^2
        &= \E_{\smi}[u_i^1(s_i, \smi)]^2\\
        \E_{\smi}[u_i(s_i,\smi)^2]
        &= \E_{\smi}[u_i^1(s_i,\smi)^2] + 2\underbrace{\E_{\smi}[u_i^1(s_i, \smi) \cdot\Delta(s_i, \smi)]}_{\geq 0} + \underbrace{\E_{\smi}[\Delta(s_i,\smi)^2]}_{\geq 0}
    \end{align}
    To show that the extra terms are non-negative, notice that $\Delta(s_i, \smi)^2 \geq 0$, and that 
    \begin{align}
        \E_{\smi}[u_i^1(s_i, \smi) \cdot\Delta(s_i, \smi)]
        &= (\underbrace{s_i-\bid[1](s_i)}_{\geq 0}) \cdot (\underbrace{q_i(s_i)\cdot \bid[1](s_i) - \E_{\smi}[\ind{\text{$i$ wins}}\cdot p_i(s)]}_{= \E_{\smi}[\Delta(s_i,\smi)] = 0})
    \end{align}
    {The second equality uses the winners-pay restriction: losers pay zero, so $\E_{\smi}[\ind{i \text{ wins}}\cdot p_i(s)]=\E_{\smi}[p_i(s)]$.}
    Therefore, putting everything together, we obtain
    \begin{align}
        \mathrm{Var}_{\smi}[u_i(s_i,\smi)]
        &= 
        \E_{\smi}[u_i(s_i,\smi)^2]
        -
        \E_{\smi}[u_i(s_i,\smi)]^2\\
        &\geq 
        \E_{\smi}[u_i^1(s_i,\smi)^2]
        -
        \E_{\smi}[u_i^1(s_i,\smi)]^2\\
        &=
        \mathrm{Var}_{\smi}[u_i^1(s_i,\smi)]
    \end{align}
    Finally, observe that an auction which minimizes the interim variance also minimizes the ex ante variance.
    Denoting by $u_i$ the utility of a bidder in the pay-as-bid auction, the law of total variance states
    \begin{equation}
        \mathrm{Var}_s[u_i]=\E_{s_i}[\mathrm{Var}_{\smi}[u_i(s_i,\smi)]]+\mathrm{Var}_{s_{i}}[\E_{\smi}[u_i]].
    \end{equation}
    By the revenue equivalence theorem, we know that $\E_{\smi}[u_i]$ is the same for all standard auctions, hence $\mathrm{Var}_{s_{i}}[\E_{\smi}[u_i]]$ is also the same for all standard auctions (it only depends on the distribution of the signals). 
    
    The interim variance $\mathrm{Var}_{\smi}[u_i(s_i,\smi)]$ is minimal point-wise (in $s_i$) for all standard auctions, hence is also minimal in expectation. Therefore, the ex ante variance is minimal in the pay-as-bid auction among standard, winners-pay auctions.
\end{proof2}

\bigbreak

With pure private values, ex post individual rationality holds. \Cref{lem:connecting-variances} shows that, in this case, any ranking of auction formats in terms of ex ante variance ($\var$) or winners' ex ante variance ($\wvar$) is identical. In contrast, a ranking with respect to empirical variance ($\ev$) can differ depending on whether only winners are considered or all bidders.

\begin{proof2}[Proof of \cref{lem:connecting-variances}]\phantomsection\label{proof:connecting-variances}
Using \cref{lem:alternative-variance-2}, we first note that
\begin{align*}
    \wvar &= \E_\ssb[u_1^2\,|\,\text{1 wins}] - \E_{\ssb}[u_1\,|\,\text{1 wins}]^2 = \frac{n}{k}\E_\ssb[u_1^2] - \left(\frac{n}{k}\right)^2\E_{\ssb}[u_1]^2
\end{align*}
For the ex ante variance, we write:
\begin{align*}
\var &= \E_\ssb[u_1^2] - \E_{\ssb}[u_1]^2
= \Prob_{\ssb}[\text{1 wins}]\cdot \E_{\ssb}[u_1^2\,|\,\text{1 wins}] - \E_{\ssb}[u_1]^2\\
&= \Prob_{\ssb}[\text{1 wins}]\cdot \E_{\ssb}[u_1^2\,|\,\text{1 wins}] - \Prob_{\ssb}[\text{1 wins}]\cdot \E_{\ssb}[u_1\,|\,\text{1 wins}]^2\\
&\quad{} + \Prob_{\ssb}[\text{1 wins}]\cdot \E_{\ssb}[u_1\,|\,\text{1 wins}]^2 - \E_{\ssb}[u_1]^2\\
&= \Prob_{\ssb}[\text{1 wins}]\cdot \wvar + \Prob_{\ssb}[\text{1 wins}]\cdot \E_{\ssb}[u_1\,|\,\text{1 wins}]^2 - \E_{\ssb}[u_1]^2\\
&= \Prob_{\ssb}[\text{1 wins}]\cdot \wvar + \frac{\Prob_{\ssb}[\text{1 wins}]^2}{\Prob_{\ssb}[\text{1 wins}]}\cdot \E_{\ssb}[u_1\,|\,\text{1 wins}]^2 - \E_{\ssb}[u_1]^2\\
&= \frac{k}{n}\cdot \wvar+\left(\frac{n}{k}-1\right)\cdot \E_{\ssb}[u_1]^2
\end{align*}
For the empirical variance, we write:
\renewcommand{\minalignsep}{0pt}
\begin{align*}
\ev~&= \E_\ssb[u_1^2] - \E_{\ssb}[u_1u_2]
= \Prob_{\ssb}[\text{1 wins}]\cdot \E_{\ssb}[u_1^2\,|\,\text{1 wins}] - \Prob_{\ssb}[\text{1 and 2 win}]\cdot \E_{\ssb}[u_1 u_2 \,|\, \text{1 and 2 win}]\\
&= \frac{k(k-1)}{n(n-1)}\cdot \wev+\left(1-\frac{k-1}{n-1}\right)\cdot \E_{\ssb}[u_1^2]
\qedhere
\end{align*}
\end{proof2}

\bigbreak

\begin{proof2}[Proof of \cref{cor:empvar-among-all-bidders}]\phantomsection\label{proof:empvar-among-all-bidders}
By \cref{lem:connecting-variances}, $\ev=\frac{k(k-1)}{n(n-1)}\,\wev+\left(1-\frac{k-1}{n-1}\right)\E_{\ssb}[u_1^2]$, and both coefficients are non-negative. Set $c=0$ and $\delta=1$. The hypothesis $\pd{\bid[1]}{s}\leq 2$ is exactly the condition of \cref{prop:delta-dominating-uniform} at $c=0$, $\delta=1$, so absolute pairwise surplus differences among winners are weakly smaller under pay-as-bid pricing. Taking conditional expectations of their squares gives $\wev^{1}\leq\wev^{U}$. Moreover, by \cref{prop:ex-ante-variance} and revenue equivalence, pay-as-bid minimizes $\E_{\ssb}[u_1^2]$ among standard winners-pay auctions, so $\E_{\ssb}[(u_1^{1})^2]\leq \E_{\ssb}[(u_1^{U})^2]$. Since both coefficients are non-negative, $\ev^{1}\leq\ev^{U}$.
\end{proof2}

\bigbreak

\begin{proof2}[Proof of \cref{prop:PAB-dominates-uniform-reserve}]\phantomsection\label{proof:reserve_price}
    We show that if MEU holds without a reserve price in the pay-as-bid auction, it also holds with a strictly positive reserve price $r>0$. Thus, once MEU holds for the pay-as-bid reserve bid, the dominance of pay-as-bid over uniform pricing follows from a direct pairwise comparison among winners (as in \cref{prop:delta-dominating-uniform}), which we give at the end. To establish that MEU carries over, we derive the equilibrium for the FRB uniform-price auction with a common value and given reserve price $r>0$. Fix signal $s$, let $\beta$ be an increasing symmetric equilibrium, and let $s_r$ be the participation threshold, i.e., the smallest signal at which entering with the positive bid $V(s_r)$ weakly dominates non-participation. If no type participates, the claim is immediate.  For $z \geq s_r$, we consider $U(s_i,z)$, the expected payoff of bidding $\beta(z)$ with signal $s_i$:
    \begin{align*}
        U(s_i,z) & =\int_0^z \widetilde V(s_i,y) g(y) \diff y - \int_0^{s_r} r g(y) \diff y - \int_{s_r}^z \beta(y) g(y)\diff y \\
        & =\int_0^z \widetilde V(s_i,y) g(y) \diff y- rG(s_r) - \int_{s_r}^z \beta(y) g(y) \diff y.
    \end{align*}
    If $z<s_r$, the bidder does not participate, receives no item, and $U(s_i,z)=0$. If the payoff is maximized for $z \geq s_r$, then, by solving the first order condition, a bid of $\widetilde V(s_i,s_i)=V(s_i)$ is optimal. Hence, bidding $V(s_i)$ is preferred to bidding zero if the expected payoff is greater than zero. Because $V(s_i)$ is increasing and continuous, these two payoffs are equal for $s_i=s_r$ by definition: $s_r$ corresponds to the threshold signal beyond which a positive bid of $V(s_i)$ is preferred to a zero profit. The equation
    \begin{equation} \label{eq:sr_charac}
       U(s_r,s_r)= \int_0^{s_r} \widetilde V(s_r,y) g(y) \diff y - rG(s_r)=0,
    \end{equation}
    implicitly characterizes $s_r$. Since $\widetilde V(s_r,y)\leq \widetilde V(s_r,s_r)=V(s_r)$ for all $y\leq s_r$, \cref{eq:sr_charac} implies $r\leq V(s_r)$. Since $V$ is increasing, all types $s\geq s_r$ bid at least the reserve. The equilibrium bid is $\beta_r^{\delta=0}=V(s_i)$ for $s_i \geq s_r$ and $\beta_r^{\delta=0}=0$ otherwise.
    For the pay-as-bid auction, by revenue equivalence, we have that, for $s_i \geq s_r$,
    \begin{equation}
        \beta_r^{\delta=1}(s_i)=\int_0^{s_r} \frac{r g(y)}{G(s_i)}\diff y + \int_{s_r}^{s_i} \frac{V(y) g(y)}{G(s_i)}\diff y=V(s_i)+(r-V(s_r))\frac{G(s_r)}{G(s_i)}- \int_{s_r}^{s_i} \frac{ V'(y) G(y)}{G(s_i)}\diff y.
    \end{equation}
    Taking the derivative yields
    \begin{align*}
        \frac{\partial \beta_{r}^{\delta=1}(s_i)}{\partial s_i}&= \frac{g(s_i)}{G^2(s_i)}\left((V(s_r)-r)G(s_r)+\int_{s_r}^{s_i} V'(y)G(y)\diff y \right) \\
        &=\frac{g(s_i)}{G^2(s_i)} \left( \int_0^{s_r} V(y)g(y)\diff y - r G(s_r) + \int_0^{s_i} V'(y)G(y)\diff y \right)\\
        &=\frac{g(s_i)}{G^2(s_i)} \left( \int_0^{s_r} (\widetilde V(y,y)- \widetilde V(s_r,y))g(y)\diff y + \int_0^{s_i} V'(y)G(y)\diff y \right) \\
        &\leq \frac{g(s_i)}{G^2(s_i)} \int_0^{s_i} V'(y)G(y)\diff y\\
        & =\frac{\partial \beta_{r=0}^{\delta=1}(s_i)}{\partial s_i}.
    \end{align*}
    We use \cref{eq:sr_charac} for the second-to-last equality, and the fact that $\widetilde V(y,y) \leq \widetilde V(s_r,y)$ for $y \leq s_r$, by monotonicity of $\widetilde V$, for the inequality. Hence the slope of the pay-as-bid bid with a reserve price is bounded by its slope without a reserve price, so by \cref{lem:monotone-ex-post-equivalence} MEU also holds in the reserve-price setting, as claimed.
    Finally, with a common reserve price the two formats share the same set of winners, and among winners the pairwise surplus difference is $|u_i^{\delta=1}-u_j^{\delta=1}|=|(1-c)(s_i-s_j)-(\beta_r^{\delta=1}(s_i)-\beta_r^{\delta=1}(s_j))|$, since the common-value and reserve terms cancel. As $0\leq(\beta_r^{\delta=1})'\leq 1-c$ by MEU, the map $s\mapsto(1-c)s-\beta_r^{\delta=1}(s)$ is non-decreasing and $(1-c)$-Lipschitz, so $|u_i^{\delta=1}-u_j^{\delta=1}|\leq(1-c)|s_i-s_j|=|u_i^{\delta=0}-u_j^{\delta=0}|$. Hence pay-as-bid is equity-preferred to uniform pricing, exactly as in \cref{prop:delta-dominating-uniform}; this uses only the $\delta=1$ endpoint, not the $\delta$-monotonicity (\cref{lem:equilibrium-bid-monotone}) behind \cref{prop:PD-decreasing}.
\end{proof2}

\section{Additional Proofs for the Main Paper}\label{suppl-apx:additional-proofs}

The three proofs below support the equilibrium characterizations stated in \cref{prop:equilibrium-bid-FRB,prop:equilibrium-bid-delta} and establish the derivative formulas in \cref{lem:derivatives}.

In the equilibrium proofs, admissible symmetric strategies are strictly increasing on $\osupport$ and use their continuous extension at $s=0$. Because signals are atomless, equilibrium bids tie with probability zero, so the strict inequality in \cref{def:mixed-auctions} and any rule for rationing ties induce the same interim allocation.

\begin{proof2}[Proof of \cref{prop:equilibrium-bid-FRB}]\phantomsection\label{proof:prop:equilibrium-bid-FRB}
    Consider bidder $i$ and let all bidders $j\neq i$ follow the bidding strategy $\bid[U](s_j) = \widetilde V(s_j,s_j)$. First, observe that $\bid[U]$ is continuous and increasing. Then bidder $i$'s expected payoff when their signal is $s_i$ and bidding $\bid[U](z)$ is given by
    \begin{align*}
        U(s_i,z) :=
        \int_0^{z} \left(\widetilde V(s_i,y) - \widetilde V(y,y)\right) \gnk(y)\diff y
    \end{align*}
    Because $\widetilde V(s_i,y)$ is increasing in $s_i$, it holds for all $y < s_i$ that $\widetilde V(s_i,y)-\widetilde V(y,y) > 0 $, and for all $y > s_i$ that $\widetilde V(s_i,y) - \widetilde V(y,y) < 0$.
    Therefore, choosing $z = s_i$ maximizes bidder $i$'s expected payoff $U(s_i,z)$.

    Because $\bid[U]=V$ is continuous and strictly increasing on $\support$, its range is
    \begin{align*}
        \bid[U](\support)=\left[V(0),\lim_{s\uparrow\vbar}V(s)\right),
    \end{align*}
    where the upper endpoint may be infinite. The preceding maximization covers every bid in this range. A bid below $V(0)$ loses and yields zero. If the upper endpoint is finite, any bid at or above it wins for sure and gives the same payoff as the limit $z\uparrow\vbar$, because the bidder's own bid does not enter the uniform price. Hence no bid outside the range improves on $z=s_i$. To establish uniqueness, let $\widehat\beta$ be any symmetric and strictly increasing equilibrium strategy. It induces the same allocation rule as $\bid[U]$, and revenue equivalence therefore implies, for every $s\in\osupport$,
    \begin{align*}
        \int_0^s \widehat\beta(y)\gnk(y)\diff y
        = \int_0^s V(y)\gnk(y)\diff y.
    \end{align*}
    Differentiating almost everywhere gives $\widehat\beta=V=\bid[U]$ almost everywhere on $\osupport$. Since $\gnk$ has full support on $\osupport$, monotonicity of $\widehat\beta$ and continuity of $V$ extend the equality to every $s\in\osupport$, and the common continuous extension pins down the strategy at $s=0$.
\end{proof2}

\bigbreak

\begin{proof2}[Proof of \cref{prop:equilibrium-bid-delta}]\phantomsection\label{proof:prop:equilibrium-bid-delta}
    First, observe that $\bid$ is continuous. We verify that it is also monotone: writing $\Gnk =: G$, $\gnk =:g$, and $\widetilde V(s,s) =: V(s)$, an alternative expression for $\bid$ is given by
    \begin{equation}
        \bid(s) = V(s) - \frac{\int_{0}^s V'(y)G(y)^{\frac{1}{\delta}}\diff y}{G(s)^{\frac{1}{\delta}}}.
    \label{eq:biddelta}
    \end{equation}
    In particular, it is differentiable almost everywhere and we can compute its derivative.
    \begin{equation}
        \bid{}'(s) = \frac{g(s)\int_{0}^s V'(y)G(y)^{\frac{1}{\delta}}\diff y}{\delta G(s)^{1+\frac{1}{\delta}}}
    \label{eq:biddeltaprime}
    \end{equation}
    which is positive almost everywhere. Next, assume that all bidders $j\neq i$ follow the bidding strategy $\bid$, and let $\bid(z)$ be bidder $i$'s bid, whose expected utility is given by
    \begin{align*}
        U(s_i,z) :=
        \int_0^{z} \left( \widetilde V(s_i,y) - \delta \bid(z)  - (1-\delta)  \bid(y) \right) g(y) \diff y
    \end{align*}
    The derivative of $U(s_i,z)$ is
    \begin{align*}
        \frac{\diff U}{\diff z}(s_i,z) &=
        \widetilde V(s_i,z)g(z)-\delta\bid{}'(z)G(z)-\delta\bid(z)g(z)-(1-\delta)\bid(z)g(z)\\
        &= (\widetilde V(s_i,z)-\bid(z))g(z)-\delta\bid{}'(z)G(z).
    \end{align*}
    In equilibrium, the first order condition requires $\frac{\diff U}{\diff z}(s_i,s_i) = 0$. We solve this differential equation using $G^{\frac{1}{\delta}}$ as the integrating factor. We obtain
    \begin{align*}
        \frac{\diff}{\diff z} \left[G(z)^{\frac{1}{\delta}}\bid(z)\right]
        = \left(\frac{1}{\delta}G(z)^{\frac{1}{\delta}-1}\right)\cdot
        (\bid(z)g(z)+\delta\bid{}'G(z))
        = \left(\frac{1}{\delta}G(z)^{\frac{1}{\delta}-1}\right) \cdot V(z) g(z).
    \end{align*}
    Solving for $\bid$, we obtain
    \begin{align}\label{equ:equilibrium-bid-delta}
        \bid(s) = \frac{\int_{0}^{s} V(y) \gnk(y) \Gnk(y)^{\frac{1}{\delta} -1} \diff y}{\delta \Gnk(s)^{\frac{1}{\delta}}}.
    \end{align}
    Using equations (\ref{eq:biddelta}) and (\ref{eq:biddeltaprime}), and the fact that $\widetilde V(s_i,z)$ is increasing in $s_i$, we have that $\frac{\diff U}{\diff z}$ is positive when $z< s_i$, zero when $z=s_i$, and negative when $z > s_i$. Therefore, choosing $z=s_i$ maximizes $i$'s expected payoff $U(s_i,z)$.

    Because $\bid$ is continuous and strictly increasing on $\support$, its range is
    \begin{align*}
        \bid(\support)=\left[\bid(0),\lim_{s\uparrow\vbar}\bid(s)\right),
    \end{align*}
    where the upper endpoint may be infinite. The preceding maximization covers every bid in this range. A bid below $\bid(0)$ loses and yields zero. If the upper endpoint is finite, any bid at or above it wins for sure and induces the same allocation and uniform-price component as the limit $z\uparrow\vbar$, but weakly raises the discriminatory component $\delta b_i$. Hence no bid outside the range is profitable.
    
    Finally, we derive the expression for $\bid$ stated in the proposition from \cref{equ:equilibrium-bid-delta}. Writing $\Gnk =: G$ and $\gnk =: g$, observe that the derivative of $\delta G^{\frac{1}{\delta}}$ is $gG^{\frac{1}{\delta}-1}$. Using integration by parts and a change of variable, we obtain
    \begin{align*}
        \int_0^s V(y) g(y) G(y)^{\frac{1}{\delta}-1}\diff y
        &=\left[\delta V(y) G(y)^{\frac{1}{\delta}}\right]_0^s-\delta \int_0^s V'(y) G(y)^{\frac{1}{\delta}}\diff y\\
        &=\delta V(s) G(s)^{\frac{1}{\delta}}-\delta \int_{V(0)}^{V(s)} G(V^{-1}(y))^{\frac{1}{\delta}}\diff y.
    \end{align*}
    Dividing by $\delta G^{\frac{1}{\delta}}$ gives the result.

To establish uniqueness, let $\widehat\beta$ be any symmetric and strictly increasing equilibrium strategy. Since both $\widehat\beta$ and $\bid$ are symmetric and strictly increasing, they induce the same allocation rule. Revenue equivalence and the zero utility of the lowest type therefore imply, for every $s\in\osupport$,
\begin{equation*}
    \delta G(s)\widehat\beta(s)
    +(1-\delta)\int_0^s \widehat\beta(y)g(y)\diff y
    =
    \delta G(s)\bid(s)
    +(1-\delta)\int_0^s \bid(y)g(y)\diff y.
\end{equation*}
The preceding identity implies that $\widehat\beta$ is locally absolutely continuous on $\osupport$. Differentiating almost everywhere gives $\delta G(s)\frac{\diff}{\diff s}\widehat\beta(s)
    +g(s)\widehat\beta(s)
    =
    \delta G(s)\frac{\diff}{\diff s}\bid(s)
    +g(s)\bid(s)$.
Multiplying by $G(s)^{1/\delta-1}/\delta$, we obtain
\begin{equation*}
    \frac{\diff}{\diff s}
    \left[G(s)^{1/\delta}\widehat\beta(s)\right]
    =
    \frac{\diff}{\diff s}
    \left[G(s)^{1/\delta}\bid(s)\right]
\end{equation*}
almost everywhere on $\osupport$. Hence
$G(s)^{1/\delta}\widehat\beta(s)-G(s)^{1/\delta}\bid(s)$ is constant. This constant is zero, since $G(s)\to0$ as $s\downarrow0$ and both strategies have finite continuous extensions at zero. Thus $\widehat\beta=\bid$ on $\osupport$, and their common continuous extension gives equality at $s=0$.
\end{proof2}

\bigbreak

\begin{proof2lower}[Proof of \cref{lem:derivatives}]\phantomsection\label{proof:lem:derivatives}
The value decomposition and equilibrium formula imply that $V$ and $\bid$ are locally absolutely continuous on $\osupport$. Since $V$ is increasing and continuous at zero, $\int_0^s V'(y)\diff y=V(s)-V(0)<\infty$. We use dominated convergence for the limiting values below.

(1) Let $s \in \osupport$. We first look at $\bid(s)=V(s) - \int_0^s V'(y) \left(\frac{G(y)}{G(s)}\right)^{\frac{1}{\delta}}\diff y$. Let $h(\delta,y)$ be the function under the integral. Because $G$ is increasing, for $y<s$ we have that $G(y)/G(s)<1$. Hence $h$ is dominated by $V'$, and $\lim_{\delta \rightarrow 0}h(\delta,y)=0$, hence by dominated convergence $\bid(s)=V(s)$ when $\delta=0$, and the function is separately continuous over $[0,1] \times \support$.

(2) Consider the derivative of $\bid$ with respect to $s$. Fix $\delta>0$. Differentiating $\bid(s)$ yields
{\small
\begin{align*}
    \frac{\partial \bid(s)}{\partial s}
    \!=\!V'(s)
    -\frac{
        V'(s)G(s)^{2/\delta}
        -\frac{1}{\delta}g(s)G(s)^{1/\delta-1}
        \!\!\int_0^s V'(y)G(y)^{1/\delta}\diff y
    }{G(s)^{2/\delta}} =\frac{g(s)}{G(s)}
    \int_0^s
    \!\!\frac{V'(y)}{\delta}
    \left(\frac{G(y)}{G(s)}\right)^{1/\delta}\!\!\!\!\!\diff y .
\end{align*}
}
For $\delta=0$, we have $\bid=\bid[U]=V$, and therefore $(\bid[U])'(s)=V'(s)$.

(3)  Let us now compute the derivative of $\bid$ with respect to $\delta$. Again, we use a dominated convergence property to show that the integral and derivative can be inverted. It is easier to show that this can be done for the function $\delta \bid(s)$, and we have $\frac{\partial \delta \bid(s)}{\partial \delta} =\bid+  \delta \frac{\partial \bid(s)}{\partial \delta}$. Computing the derivative of $\delta \bid(s)$, we can recover that of $\bid(s)$.
Let $h(\delta,y,s)=\delta V'(y)(G(y)/G(s))^{1/\delta}$ be the function under the integral part of $\delta \bid$. We have
{\small
\begin{align*}
\frac{\partial h(\delta,y,s)}{\partial \delta}&=V'(y)\left(\frac{G(y)}{G(s)}\right)^{1/\delta}-V'(y)\frac{\delta}{\delta^2}\log\left(\frac{G(y)}{G(s)}\right)\left(\frac{G(y)}{G(s)}\right)^{1/\delta}\\
&=V'(y)\left(\frac{G(y)}{G(s)}\right)^{1/\delta}-V'(y)\log\left(\left(\frac{G(y)}{G(s)}\right)^{1/\delta}\right)\left(\frac{G(y)}{G(s)}\right)^{1/\delta}.
\end{align*}
}

The first part is again dominated by $V'$, which is integrable. Focusing on the second part, we define for $0< u < w <1$ the function $\psi(u,w)=(u/w)\log(w/u)$. Note that $0<y<s<\vbar$ implies that for $u=G^{1/\delta}(y)$ and $w=G^{1/\delta}(s)$, we have $0<u<w<1$ as $G$ is increasing and takes values in $(0,1)$ over $\osupport$ by definition. Fix $w$, and take the derivative with respect to $u$: we obtain that $\psi'(u,w)=(\log(w/u)-1)/w$ which is positive as long as $u\leq w/e$ and negative otherwise. The maximum of $\psi$ for $u<w$ is at $u=w/e$ and $\psi(w/e,w)=1/e$. The second term of $\partial_\delta h$ is bounded by $V'(y)/e$. Hence $|\partial_\delta h|\leq(1+1/e)V'(y)$, which is integrable. Dominated convergence therefore justifies differentiation under the integral: $\frac{\partial}{\partial \delta} \int h= \int \frac{\partial}{\partial \delta} h$. Thus
{\small
\begin{equation*}
\frac{\partial \delta \bid}{\partial \delta}= V(s) - \int_0^s V'(y)\left(\frac{G(y)}{G(s)}\right)^{1/\delta} \diff y + \int_0^s V'(y)\log\left(\left(\frac{G(y)}{G(s)}\right)^{1/\delta}\right)\left(\frac{G(y)}{G(s)}\right)^{1/\delta} \diff y,
\end{equation*}
}
and we recover
{\small
\begin{equation*}
 \delta \frac{\partial \bid(s)}{\partial \delta}= \int_0^s V'(y)\log\left(\left(\frac{G(y)}{G(s)}\right)^{1/\delta}\right)\left(\frac{G(y)}{G(s)}\right)^{1/\delta} \diff y.
\end{equation*}
}

Using the same upper bound on $\psi$, we can show that the absolute value of the integrand of $\delta \frac{\partial \bid(s)}{\partial \delta}$ is bounded by $V'(y)/e$ which allows for domination both in small $\delta$ and small $s$. By dominated convergence, we obtain that the limit of $\delta \frac{\partial \bid(s)}{\partial \delta}$ as either $\delta$ or $s$ go to $0$ is $0$.

(4) Finally, let us compute the cross derivative. The integrand of
$\frac{\partial(\delta\bid(s))}{\partial s}$ is $h(\delta,y,s)=V'(y) (G(y)/G(s))^{1/\delta}$, and its derivative with respect to delta is $-\frac{1}{\delta}V'(y)\log\left(\left(\frac{G(y)}{G(s)}\right)^{1/\delta}\right)\left(\frac{G(y)}{G(s)}\right)^{1/\delta}$. For $0<\varepsilon\leq\delta\leq1$, the absolute value of this derivative is bounded by $V'(y)/(e\varepsilon)$, which is integrable on $[0,s]$. Dominated convergence therefore justifies differentiation under the integral. Therefore
{\small
\begin{align*}
\frac{\partial^2  \delta \bid }{\partial \delta \partial s}&=\frac{g(s)}{G(s)\,G(s)^{2/\delta}}\Bigl[ - G^{1/\delta} (s) \tfrac{1}{\delta}\int_{0}^s V'(y) \log(G^{1/\delta}(y))G^{1/\delta}(y) \diff y \\
&\qquad + \frac{1}{\delta} \log(G^{1/\delta}(s)) G^{1/\delta}(s)\int_{0}^s V'(y) G^{1/\delta}(y) \diff y \Bigr]\\
&=\frac{-g(s)}{\delta G(s)}\int_0^s V'(y) \log\left( \left( \frac{G(y)}{G(s)}\right)^{1/\delta} \right) \left( \frac{G(y)}{G(s)}\right)^{1/\delta} \diff y. \qedhere
\end{align*}
}
\end{proof2lower}

\section{Further Results, Examples, and Additional Figures}\label{app:further-results-examples-figures}

This section collects supporting material for the numerical and illustrative parts of the main text: \cref{example:bid-functions}, the refined bounds used in \cref{sec:numerical-experiments}, the calculations behind \cref{fig:common,fig:private}, the figures for \cref{example:counter} of the main paper, and additional super-mixed pricing illustrations. \Cref{example:payments-from-losers} illustrates surplus equalization conditional on winning when losing bidders also make payments.

\begin{example} \label{example:bid-functions}
    
Consider $n=3$ bidders competing for $k=2$ items. The bidders' signals are distributed uniformly on the support $[0,1]$. We illustrate the \wevtext for different values of price-mixing and the common-value proportion $c$ (\cref{fig:WEV-uniform}).
\begin{figure}[htp]
    \centering
    \captionsetup{skip=6pt}
    \includegraphics[scale=0.3,trim={0 1cm 0 0},clip]{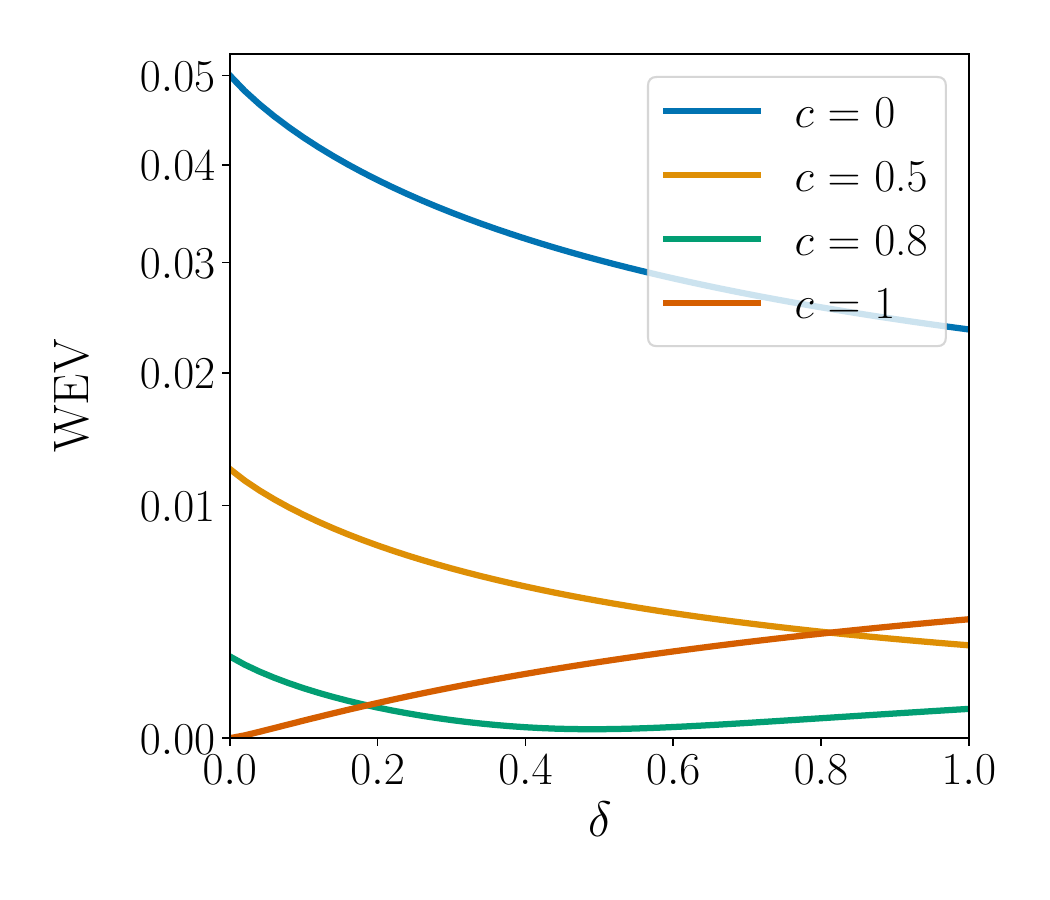}
    \caption{\wev as a function of price discrimination $\delta$ for different private-common values $c$.}
    \label{fig:WEV-uniform}
\end{figure}
For pure private and intermediate common values, the pay-as-bid auction ($\delta = 1$) minimizes \wev. For $c=0.8$, we observe an interior optimum, and for a pure common value, uniform pricing ($\delta = 0$) minimizes \wev. Note that \wev~is not necessarily convex in $\delta$ for a fixed $c$.

Whether MEU is satisfied can be verified numerically, either by computing differences in realized utilities for every pair of signals or simply by checking the derivative of the bid function. We illustrate this for the example of uniform signal distributions and $n=3$ and $k=2$ in \cref{fig:plot-MEU-uniform} below. For example, with $c=0.8$ and $\delta = 0.4$, above the MEU boundary in \cref{fig:plot-MEU-uniform}, the derivative of the bid function cannot be larger than $0.5 = \frac{1-0.8}{0.4}$ for MEU to be satisfied. The slope of the bid function is close to $0.62$ for low signals, exceeding this bound. Thus, for this combination of $c$ and $\delta$, MEU does not hold.
\begin{figure}[htp]
    \centering
    \includegraphics[scale=0.38,trim={0 1cm 0 0},clip]{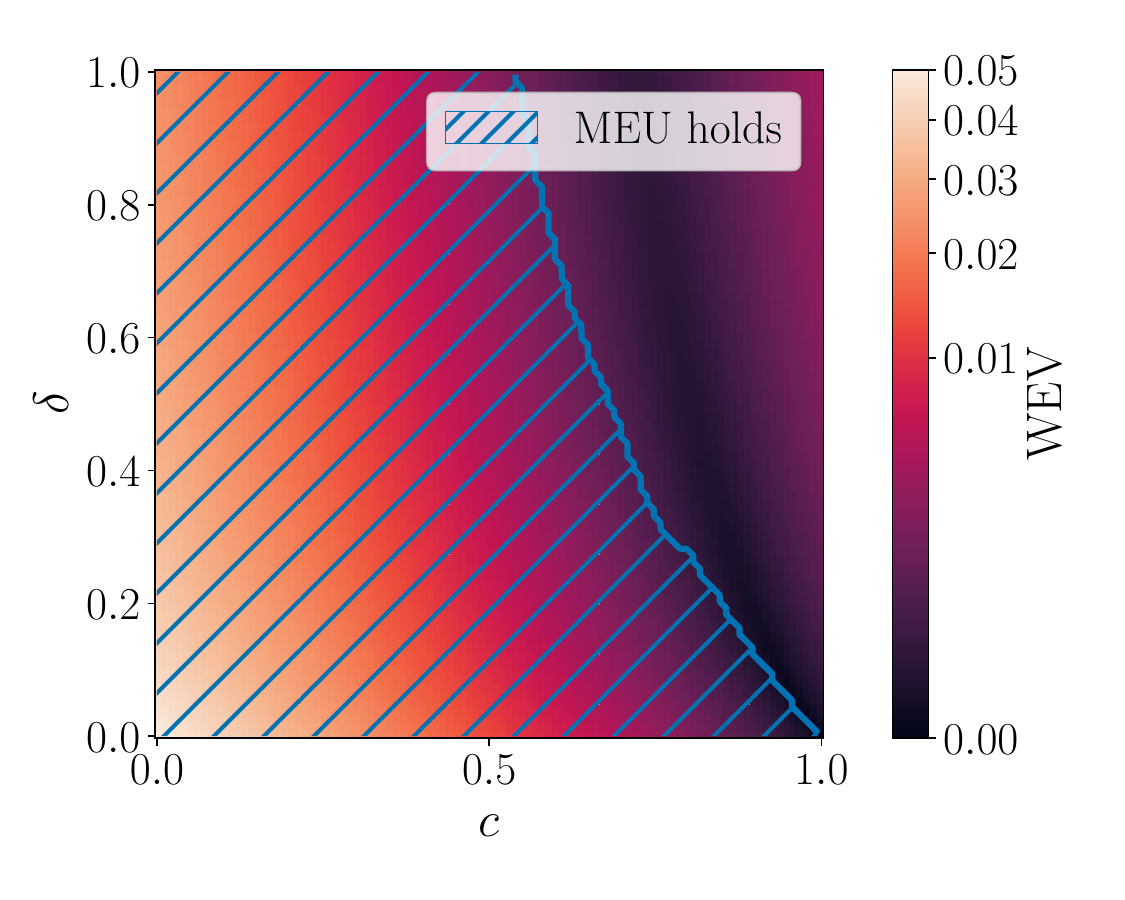}
    \caption{Monotone ex post utility for common value $c$ and price discrimination $\delta$}
    \label{fig:plot-MEU-uniform}
\end{figure}
\end{example}
\FloatBarrier

We can also compare two mixed auctions in terms of their \wevtext under a bound on their bid function slopes, stated in the following theorem.

\begin{theorem}\label{theorem:generalized-delta-vs-delta}%
For a given common value component $c$, consider two $\delta$-mixed auctions for $\delta_1 \leq \delta_2$ and suppose the equilibrium bidding functions $\bid$ satisfy $\delta_1 \frac{\partial \bid[\delta_1](s)}{\partial s}+\delta_2 \frac{\partial \bid[\delta_2](s)}{\partial s} \leq 2(1-c)$ for almost every $s\in\osupport$. Then, \wev~is weakly lower for the $\delta_2$-mixed auction than for the $\delta_1$ one.
\end{theorem}

\begin{proof}
    Let $\varphi^{\delta}(s)=(1-c)s-\delta \bid(s)$. We have $u_i^{\delta}(\ssb)-u_j^{\delta}(\ssb)=\varphi^{\delta}(s_i)-\varphi^{\delta}(s_j)$.
    By \cref{lem:equilibrium-bid-monotone}, $(\varphi^{\delta_2})'\leq(\varphi^{\delta_1})'$ almost everywhere. 
Moreover, the assumption on the bid slopes gives
\begin{equation*}
    (\varphi^{\delta_1})'+(\varphi^{\delta_2})'
    =
    2(1-c)
    -\delta_1\frac{\partial \bid[\delta_1]}{\partial s}
    -\delta_2\frac{\partial \bid[\delta_2]}{\partial s}
    \geq0.
\end{equation*}
Hence $|(\varphi^{\delta_2})'|\leq(\varphi^{\delta_1})'$ almost everywhere on $\osupport$.
    By \cref{lem:derivatives}, both functions are locally absolutely continuous. For winning signals $0<s_i\leq s_j<\vbar$,
    \begin{align*}
        |\varphi^{\delta_2}(s_j)-\varphi^{\delta_2}(s_i)|
        \leq\int_{s_i}^{s_j}|(\varphi^{\delta_2})'(t)|\diff t
        \leq\int_{s_i}^{s_j}(\varphi^{\delta_1})'(t)\diff t=\varphi^{\delta_1}(s_j)-\varphi^{\delta_1}(s_i).
    \end{align*}
    Taking conditional expectations of the squared pairwise differences, as in \cref{def:WEV}, proves the claim.
\end{proof}

\bigbreak

\begin{corollary}\label{cor:better-bounds}
    For uniformly distributed signals, any pricing dominant in \PD~contains a discriminatory proportion of at least $\frac{2n(1-c)}{2n-c(n-2)}$, and for exponentially distributed signals at least $\frac{2n(1-c)}{2n-c(n-(k+1))}$.
\end{corollary}
\begin{proof}
    The proof of \cref{theorem:bound-on-optimal-delta} applies with $\sup_{s\in\osupport}V'(s)$ in place of $1$ as the bound on $V'$ and on the slope of the equilibrium bid function. For the uniform distribution, we have $\sup_{\support} V'(s)=1-c\frac{n-2}{2n}$. Thus, using the same argument as in the proof of \cref{theorem:bound-on-optimal-delta}, it follows that the mixed auction with $\delta = \frac{2n(1-c)}{2n-c(n-2)} \rightarrow_{n\rightarrow \infty} \frac{(1-c)}{1-c/2}$ dominates in \PD~every mixed auction with less price discrimination. For the exponential distribution, we have $\sup_{\support} V'(s)=1-c(\frac{1}{2}-\frac{k+1}{2n})$, and thus this holds for $\delta= \frac{2n(1-c)}{2n-c(n-(k+1))} \rightarrow_{n\rightarrow \infty} \frac{(1-c)}{1-c/2}$.
\end{proof}

\begin{example}\label{example:illustration-pab-vs-uniform-common}
    We illustrate the redistributive effect of the pricing rule in a pure common value environment, where bidders' signals are $\operatorname{Beta}(1.9,1)$-distributed on $[0,1]$. \cref{fig:common} illustrates the realized surpluses, in equilibrium, in a setting with $k=5$ items and $n=8$ bidders. The chosen signal realizations are $(0.98, 0.94, 0.89, 0.55, 0.15, 0.10, 0.09, 0.08)$. We round all subsequent calculations to two decimals. The realized common value is given by the average signal in the market, $v=\bar s=0.47$.
    
    Writing $V(s)=\frac{1}{8}\left(\left(2+2\frac{1.9}{2.9}\right)s+4\frac{1.9}{2.9}\frac{1-s^{2.9}}{1-s^{1.9}}\right)$ and $G(s)=1-(1-s^{1.9})^7-7s^{1.9}(1-s^{1.9})^6-21s^{3.8}(1-s^{1.9})^5$, the equilibrium bid functions are $\beta^{\text{PAB}}(s)=V(s)-\frac{\int_0^s V'(y)G(y)\diff y}{G(s)}$ and $\beta^{U}(s)=V(s)$. The bids and realized surpluses are given in \cref{tab:common_value_example}.
    \begin{table}[ht]
        \centering
        \begin{tabular}{c c c c c c}
        \toprule
        Bidder & Signal & $\text{bid}^{\text{PAB}}$ & $\text{bid}^{\text{U}}$ & $\text{surplus}^{\text{PAB}}$ & $\text{surplus}^{\text{U}}$ \\
        \midrule
        1 & 0.98 & 0.65 & 0.90 & $-0.17$ & 0.10 \\
        2 & 0.94 & 0.65 & 0.87 & $-0.17$ & 0.10 \\
        3 & 0.89 & 0.65 & 0.84 & $-0.17$ & 0.10 \\
        4 & 0.55 & 0.56 & 0.62 & $-0.09$ & 0.10 \\
        5 & 0.15 & 0.39 & 0.40 & 0.09 & 0.10 \\
        6 & 0.10 & 0.37 & 0.37 & 0.00 & 0.00 \\
        7 & 0.09 & 0.36 & 0.37 & 0.00 & 0.00 \\
        8 & 0.08 & 0.36 & 0.36 & 0.00 & 0.00 \\
        \bottomrule
        \end{tabular}
        \caption{Bids and realized surpluses in a pure common-value example ($n=8$, $k=5$).
        The realized common value is the average signal, $v=\bar s=0.47$. The uniform price is $0.37$.}
        \label{tab:common_value_example}
    \end{table}
\end{example}

\begin{example}\label{example:illustration-pab-vs-uniform-private}
    This example illustrates the redistributive effects of pricing with pure private values. The bidders' signals are uniformly distributed on $[0,1]$. \cref{fig:private} illustrates the realized surpluses, in equilibrium, in a setting with $k=5$ items and $n=8$ bidders. The chosen signal realizations are $(0.34, 0.26, 0.20, 0.12, 0.08, 0.06, 0.04, 0.02)$. We choose different signal realizations than in the previous example with common values because with private values pay-as-bid equilibrium bid functions are very flat for high realizations, and thus the differences in realized surpluses are smaller. However, our results hold for any realization of signals. We round all subsequent calculations to two decimals.
    
    Writing $G(s)=1-(1-s)^7-7s(1-s)^6-21s^2(1-s)^5$, the equilibrium bid functions in the pay-as-bid auction and the uniform-price auction are $\beta^{\text{PAB}}(s)=\frac{105\int_0^s y^3(1-y)^4\diff y}{G(s)}$ and $\beta^{U}(s)=s$. The bids and realized surpluses are given in \cref{tab:private_value_example}.
    \begin{table}[htp]
    \centering
    \begin{tabular}{r c c c c c}
    \toprule
    Bidder & Signal & $\text{bid}^{\text{PAB}}$ & $\text{bid}^{\text{U}}$ & $\text{surplus}^{\text{PAB}}$ & $\text{surplus}^{\text{U}}$ \\
    \midrule
    1 & 0.34 & 0.23 & 0.34 & 0.11 & 0.28 \\
    2 & 0.26 & 0.18 & 0.26 & 0.08 & 0.20 \\
    3 & 0.20 & 0.14 & 0.20 & 0.06 & 0.14 \\
    4 & 0.12 & 0.09 & 0.12 & 0.03 & 0.06 \\
    5 & 0.08 & 0.06 & 0.08 & 0.02 & 0.02 \\
    6 & 0.06 & 0.04 & 0.06 & 0.00 & 0.00 \\
    7 & 0.04 & 0.03 & 0.04 & 0.00 & 0.00 \\
    8 & 0.02 & 0.01 & 0.02 & 0.00 & 0.00 \\
    \bottomrule
    \end{tabular}
    \caption{Bids and realized surpluses in a pure private-value example ($n=8$, $k=5$). The uniform price is $0.06$.}
    \label{tab:private_value_example}
    \end{table}
\end{example}

The next two figures provide the distributional
and bidding-function details for the counterexample in
\cref{sec:counterexample}. \Cref{fig:counterexample-signals-quantiles}
plots the signal cdf and quantile function, while
\cref{fig:example-bidfunction} plots the corresponding equilibrium bid
as a function of quantiles.

\begin{figure}[htp]
    \centering
    \captionsetup{skip=6pt}
    \begin{adjustbox}{max width=0.95\textwidth}
        \hspace{-0.5cm}
        \begin{minipage}[t]{0.4\textwidth}
        \captionsetup[subfigure]{font=footnotesize,margin={1.2cm,0.5cm}}
        \subcaptionbox*{Signal cdf}{%
        \includegraphics[width=\linewidth, trim={0 0 0 1.5cm}, clip]{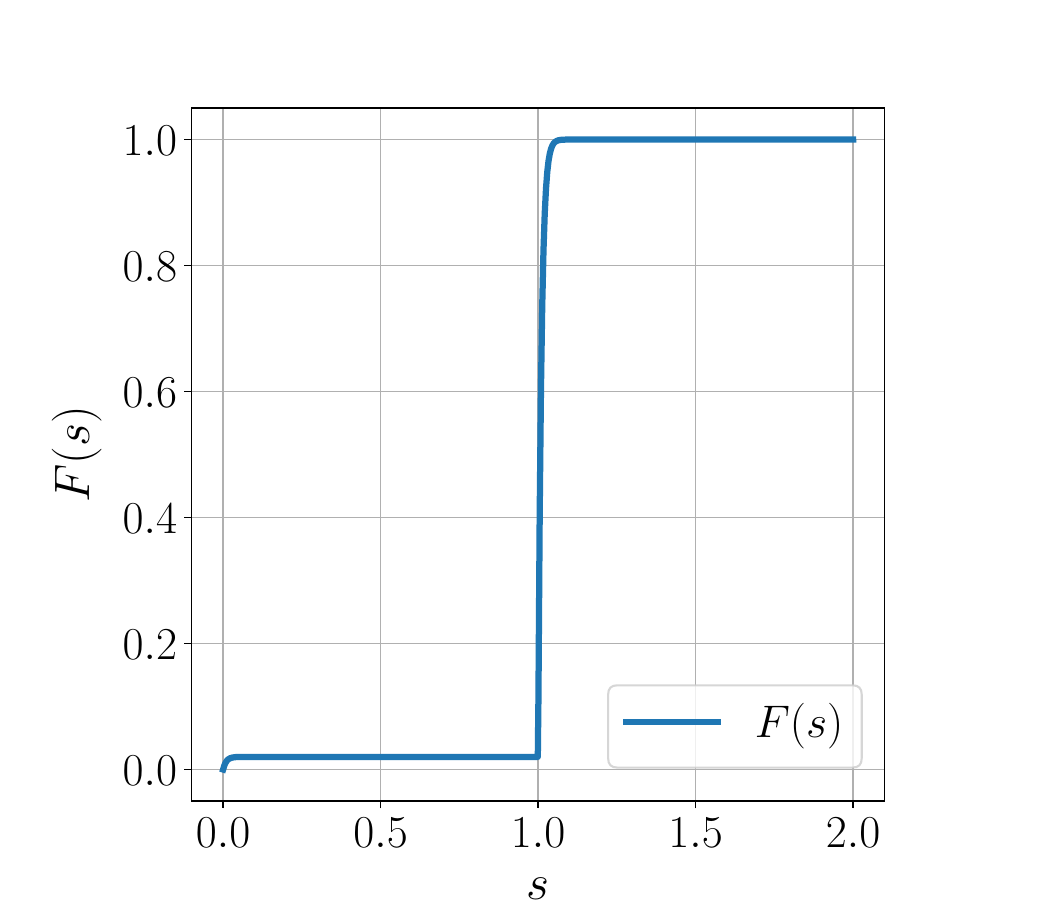}
        }
        \end{minipage}
        \hspace{0.9cm}
        \begin{minipage}[t]{0.4\textwidth}
        \captionsetup[subfigure]{font=footnotesize,margin={0.5cm,0.5cm}}
        \subcaptionbox*{Quantile function}{%
        \includegraphics[width=\linewidth, trim={0 0 0 1.5cm}, clip]{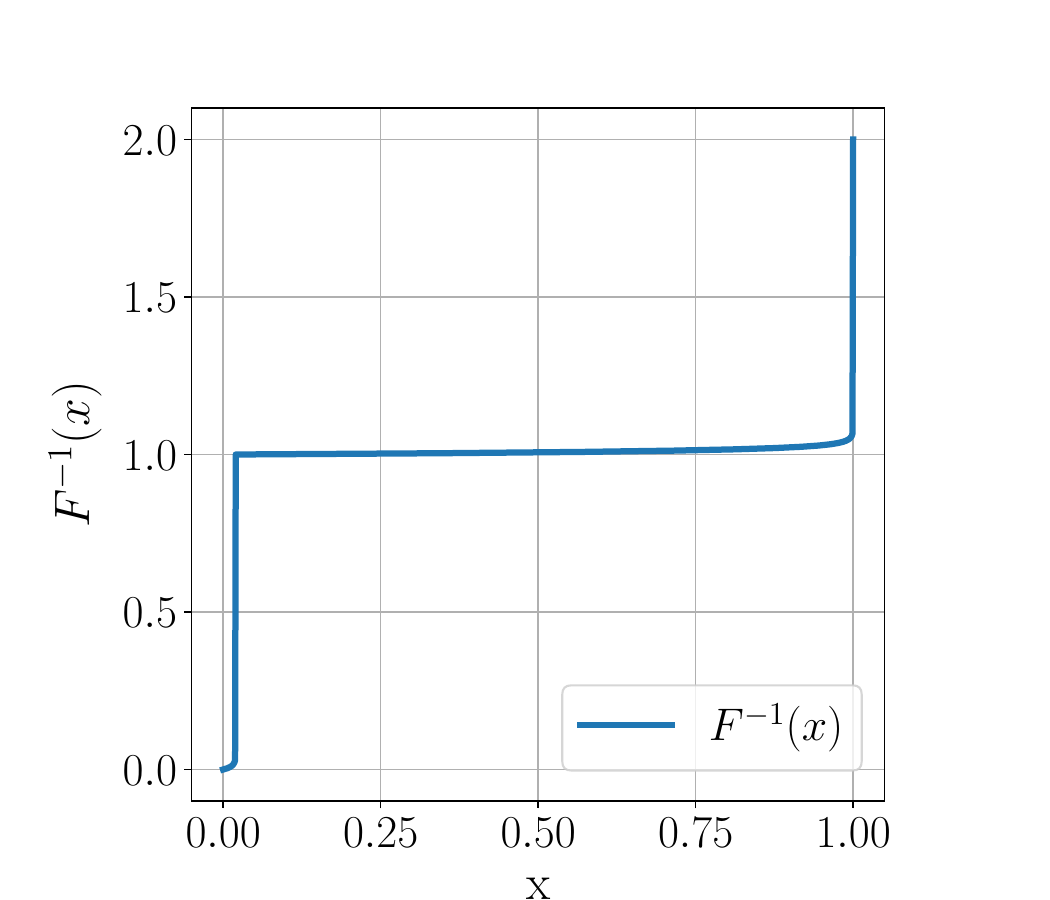}
        }
        \end{minipage}
    \end{adjustbox}
    \caption{Signal distribution and quantile function for $n=5$ and $\eta=0.01$}
    \label{fig:counterexample-signals-quantiles}
\end{figure}

\begin{figure}[htp]
    \centering
    \captionsetup{skip=6pt}
    \includegraphics[scale=0.38,trim={0 12 0 0},clip]
    {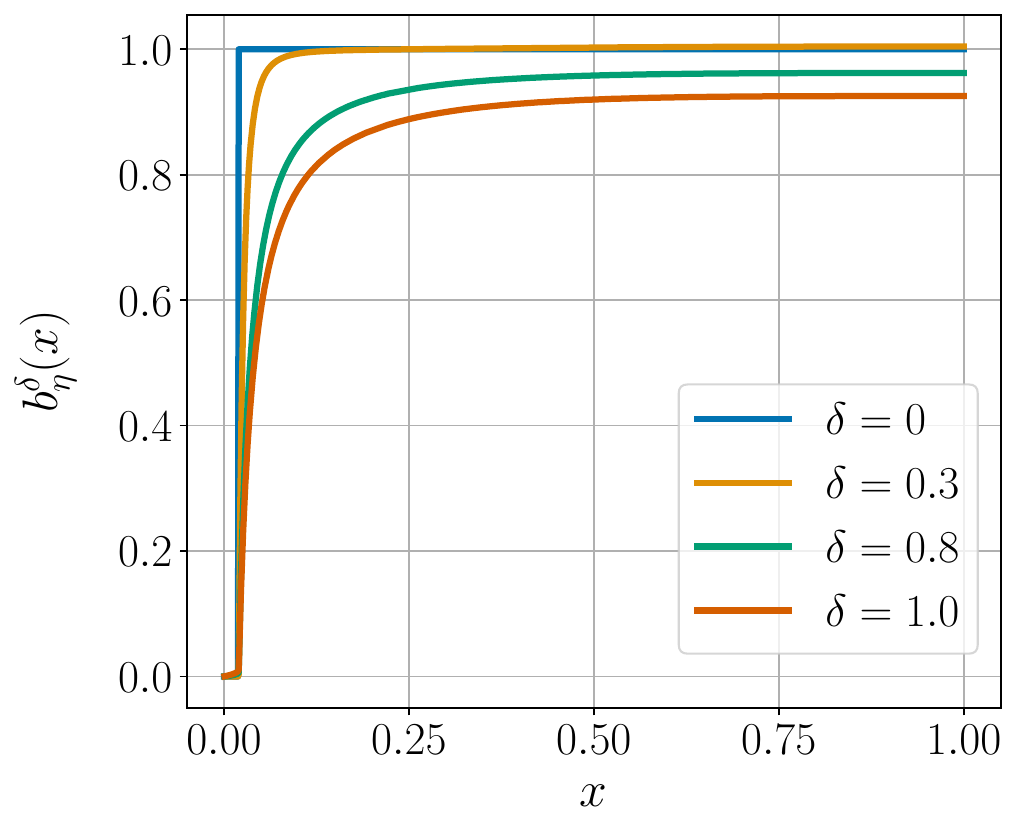}
    \caption{Equilibrium bid as a function of quantiles for $n = 5$ and $\eta = 0.01$}
    \label{fig:example-bidfunction}
\end{figure}

The numerical experiments in \cref{sec:numerical-experiments} use the theoretical simplifications in \cref{lem:value-function-bounded} of the main paper and \cref{lem:alternative-variance-2}. We rewrite multidimensional expectations as nested one-dimensional integrals with variable bounds, which we compute by integrating polynomial interpolations. The code uses memoization so that each quantity is computed at most once. We choose between analytical expressions for the same quantity according to their accuracy and computational cost at the given signal value. For example, \cref{equ:alternative-equilibrium-delta} can be integrated more efficiently than \cref{equ:equilibrium-bid-delta}, but is less accurate for small signals.

\begin{figure}[htp]
    \centering
    \hspace{-1.2cm}
    \begin{subfigure}[t]{0.22\textwidth}
        \includegraphics[scale=0.24]{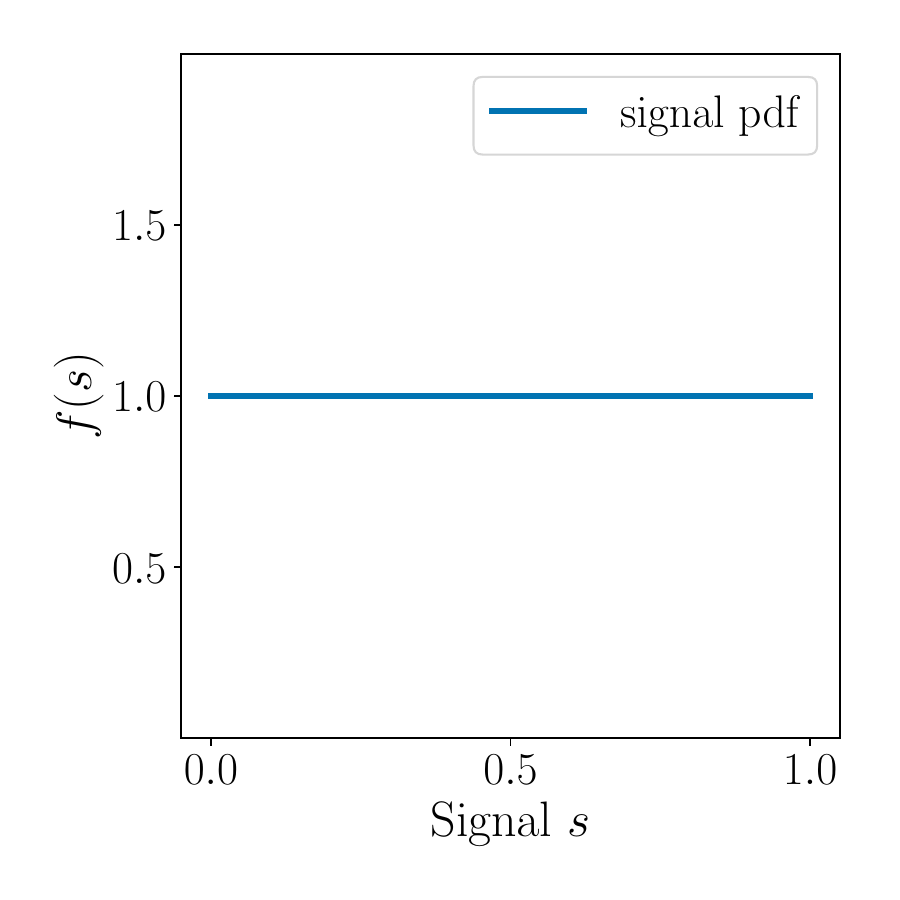}
    \end{subfigure}
    \hspace{0.23cm}
    \begin{subfigure}[t]{0.22\textwidth}
        \includegraphics[scale=0.24]{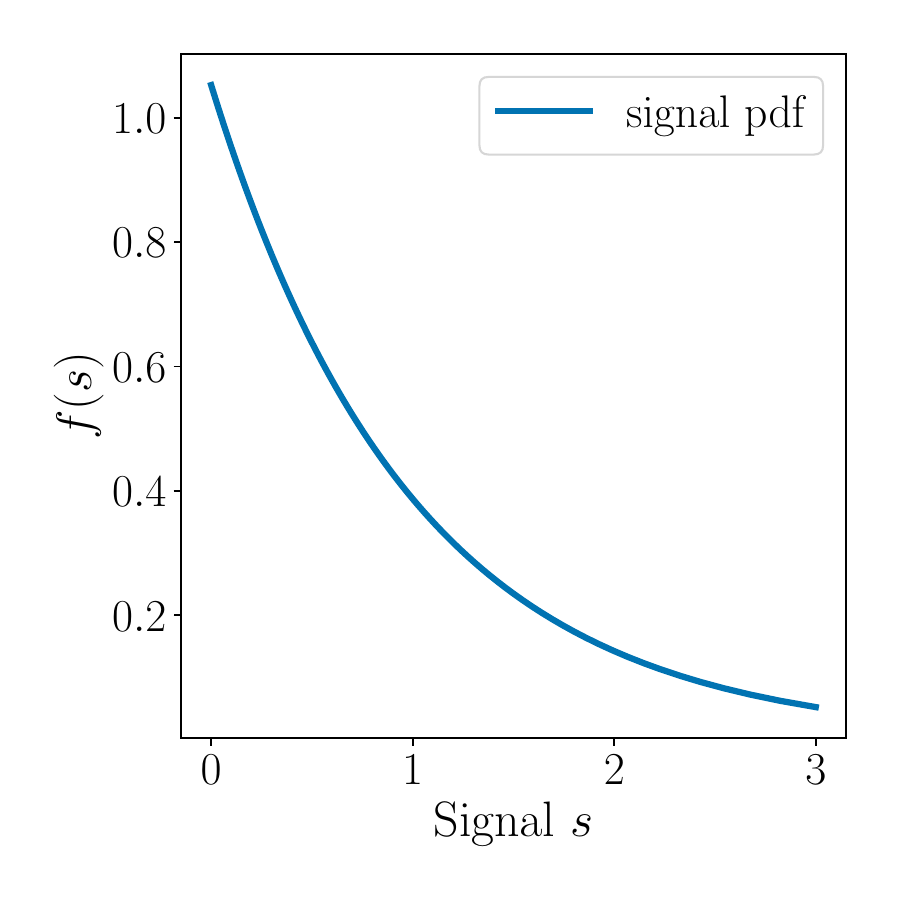}
    \end{subfigure}
    \hspace{0.23cm}
    \begin{subfigure}[t]{0.22\textwidth}
        \includegraphics[scale=0.24]{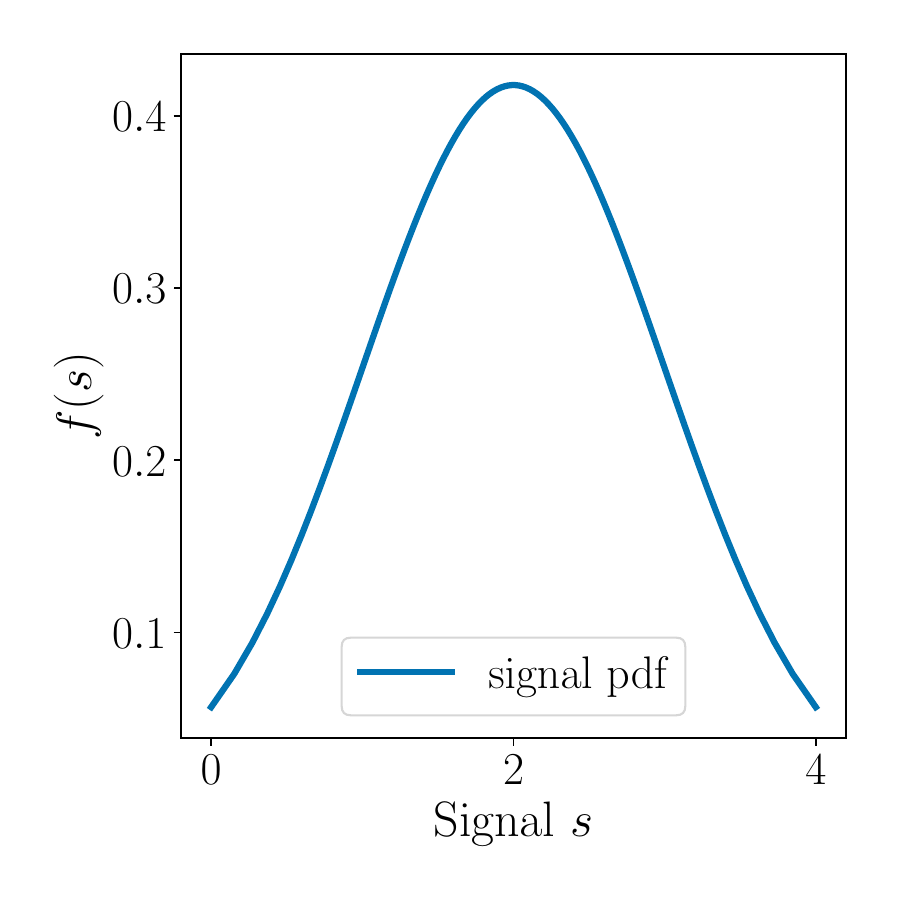}
    \end{subfigure}
    \hspace{0.23cm}
    \begin{subfigure}[t]{0.22\textwidth}
        \includegraphics[scale=0.24]{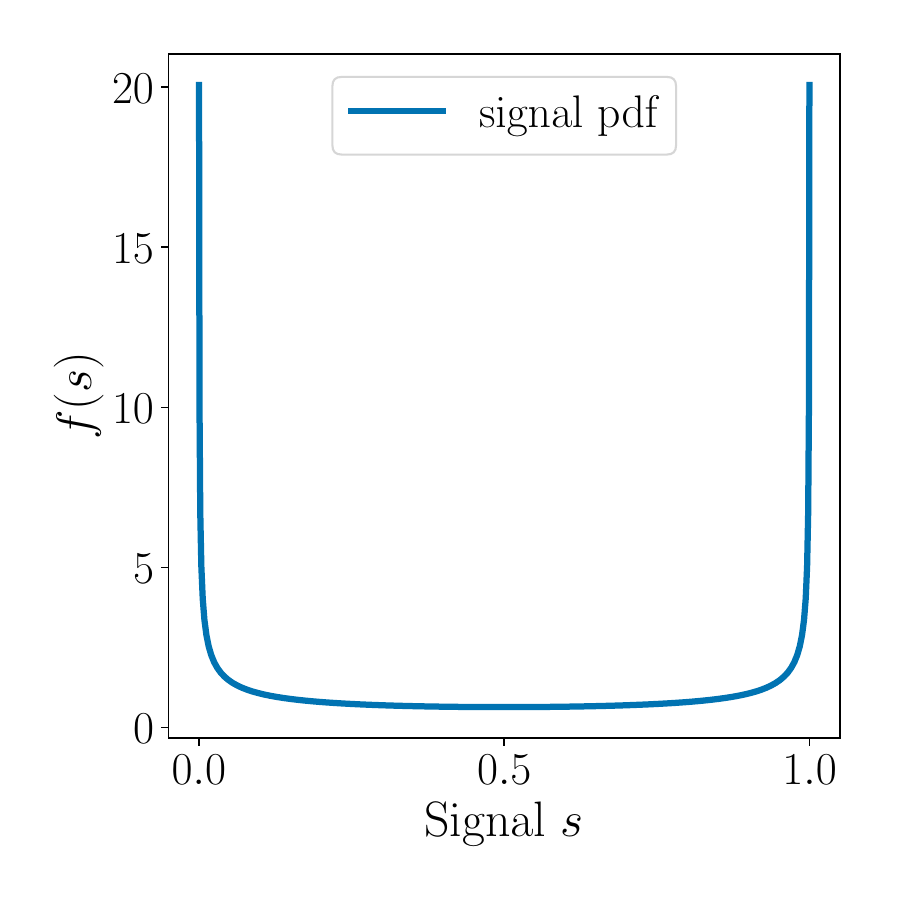}
    \end{subfigure}\\[-0.6em]
    \hspace{-0.8cm}
    \begin{subfigure}[t]{0.22\textwidth}
        \subcaptionbox*{Uniform[0,1]}{%
        \includegraphics[scale=0.242,trim={0 0.3cm 1.75cm 0},clip]{figures/plot_OPT_MEU_LB_uniform_n10_k4.pdf}}
    \end{subfigure}
    \hspace{0.5cm}
    \begin{subfigure}[t]{0.22\textwidth}
        \subcaptionbox*{Truncated exponential}{%
        \includegraphics[scale=0.242,trim={1.2cm 0.3cm 1.8cm 0},clip]{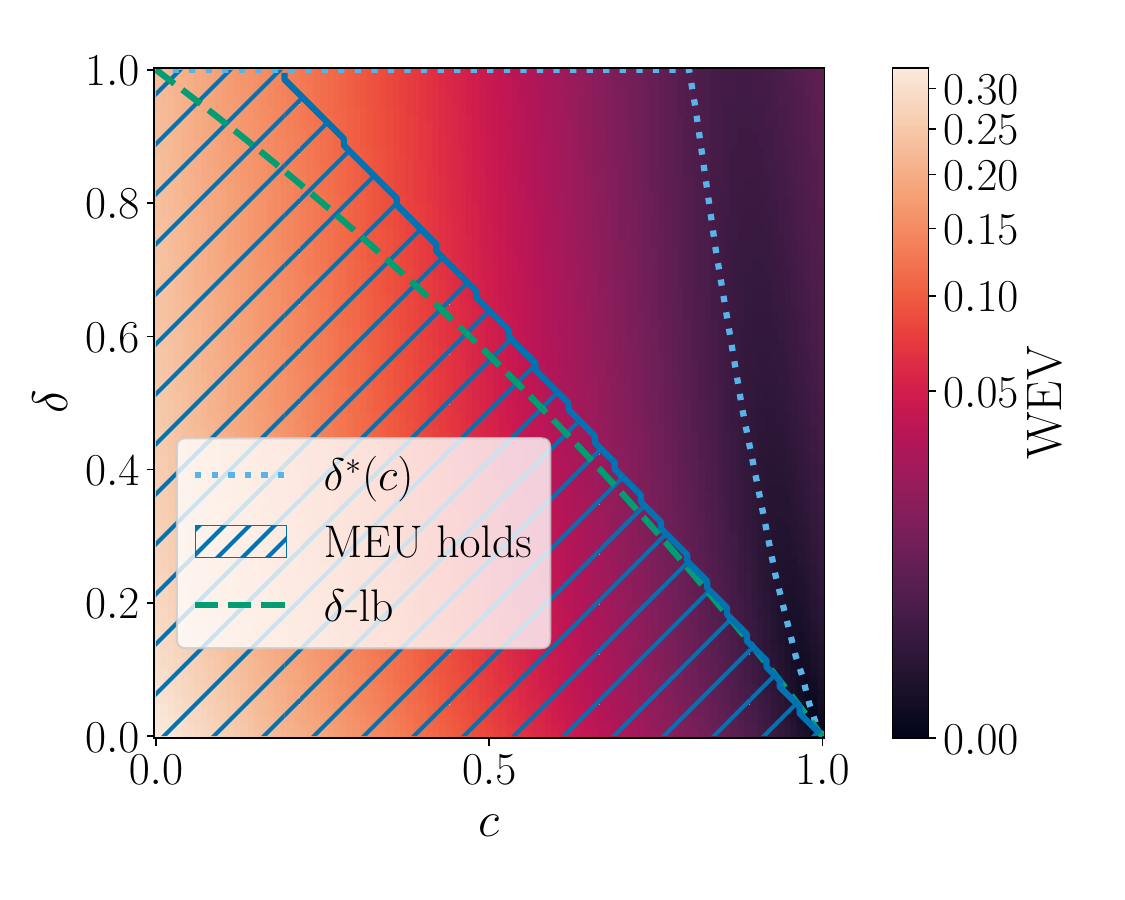}}
    \end{subfigure}
    \hspace{0.22cm} 
    \begin{subfigure}[t]{0.22\textwidth}
        \subcaptionbox*{Truncated normal~~~}{%
        \includegraphics[scale=0.24,trim={1.2cm 0.3cm 1.75cm 0},clip]{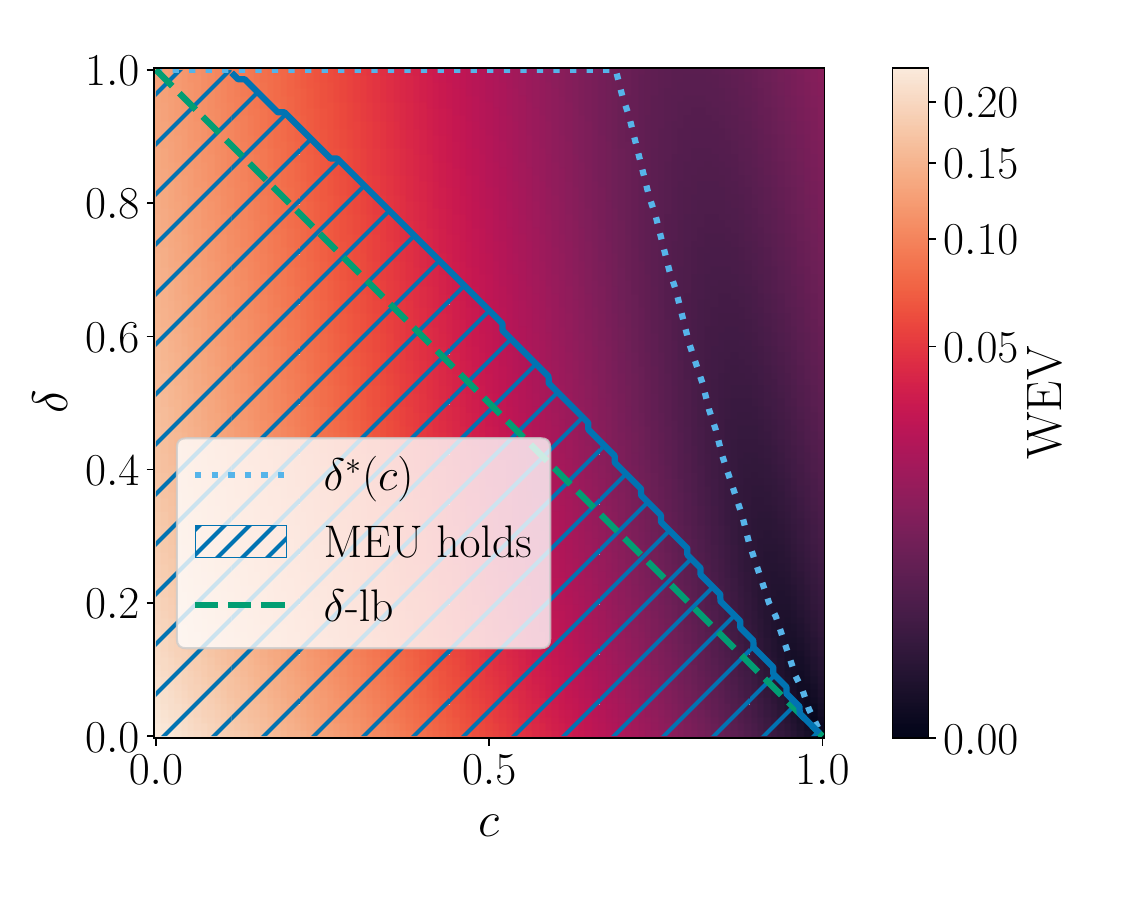}}
    \end{subfigure}
    \hspace{0.22cm}
    \begin{subfigure}[t]{0.22\textwidth}
        \subcaptionbox*{Beta$(0.5,0.5)$~~~~~~~}{%
        \includegraphics[scale=0.24,trim={1.2cm 0.3cm 0 0},clip]{figures/plot_OPT_MEUbeta0.5-0.5_n10_k4.pdf}}
    \end{subfigure}
    \caption{\wev-minimizing design $\delta^*(c)$, lower bounds $\delta$-lb on the minimizer, and monotone ex post utility for four signal distributions}
    \label{fig:WEV-optimal-various-distrib-full}
\end{figure}

The following figures extend the numerical search to $\delta>1$. For uniform signals and $n=10,k=4$, the grid minimum at $c=0$ occurs near $\delta=3$. For $n=3,k=2$, the minimizing grid point reaches the displayed upper bound $\delta=10$ for small $c$.

\begin{figure}[htp]
    \centering
    \begin{minipage}{0.45\textwidth}
        \includegraphics[width=\linewidth,trim={0 1.1cm 0 0},clip]{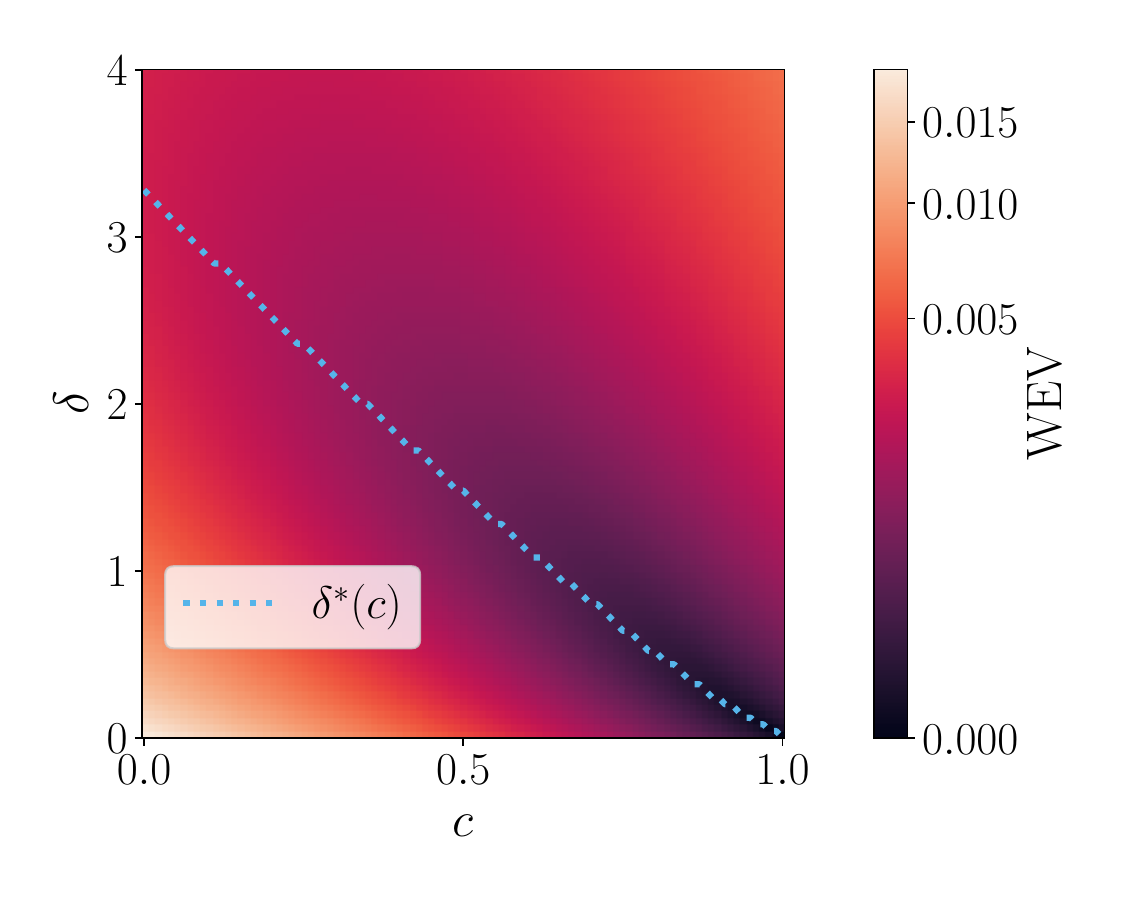}
    \caption{Super-mixed pricing instance with uniform signals and $n=10$ bidders and $k=4$ items}
    \label{fig:plot-OPT-supermixed-n10-k4}
    \end{minipage}
    \hspace{1cm}
    \begin{minipage}{0.45\textwidth}
        \includegraphics[width=\linewidth,trim={0 1.1cm 0 0},clip]{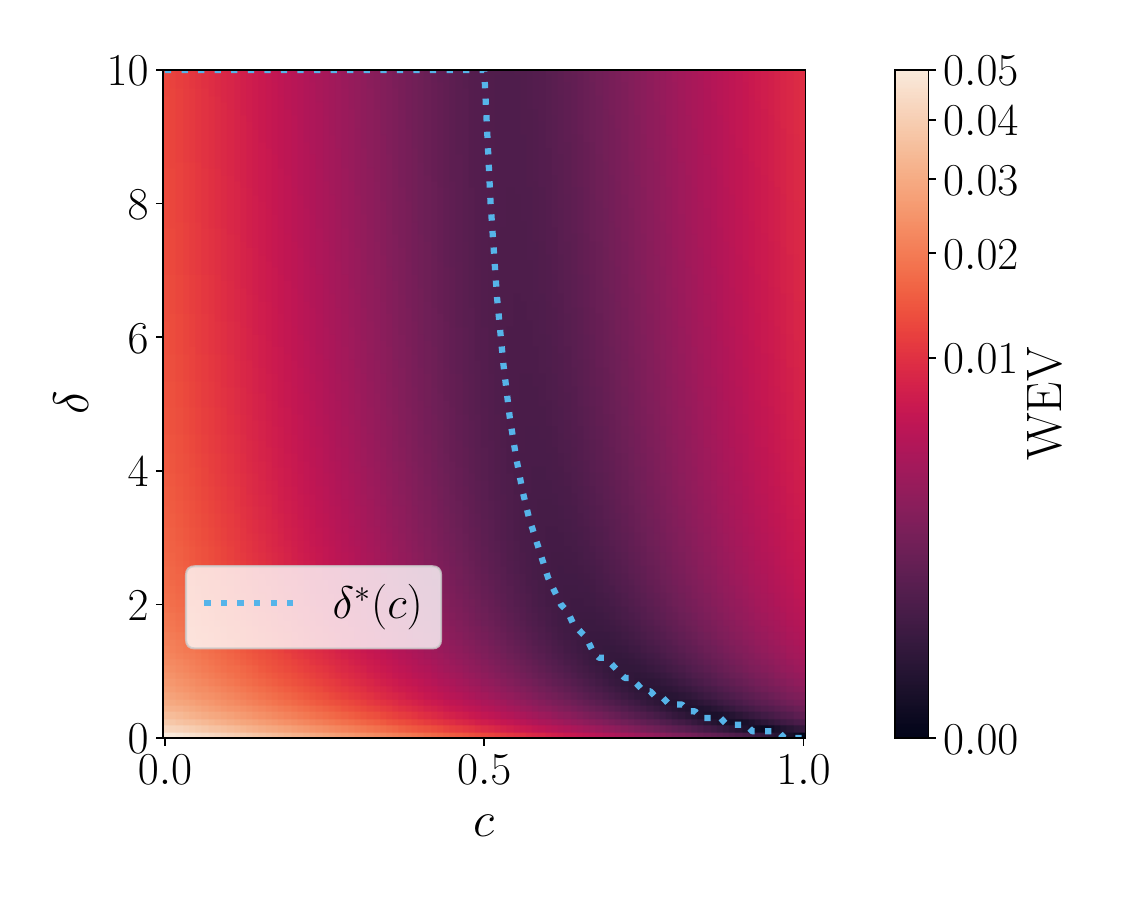}
    \caption{Super-mixed pricing instance with uniform signals and $n=3$ bidders and $k=2$ items}
    \label{fig:plot-OPT-supermixed-n3-k2}
    \end{minipage}
\end{figure}

In our numerical simulations, we make use of the following lemma.

\begin{lemma}\label{lem:alternative-variance-2}
    Suppose an auction is symmetric, winners-pay, and allocates exactly $k$ items in every realized signal profile. Then we can write
    $\E_\ssb[u_1\ | \text{$1$ wins}] = \frac{n}{k} \E_\ssb[u_1] $,
    $\E_\ssb[u_1^2 | \text{$1$ wins}] = \frac{n}{k} \E_\ssb[u_1^2]$, and
    $\E_\ssb[u_1u_2 | \text{$1$ and $2$ win}] = \frac{n(n-1)}{k(k-1)} \E_\ssb[u_1 u_2]$.
\end{lemma}

\begin{proof}
    Since the auction is winners-pay, bidder $1$ has zero utility whenever bidder $1$ loses, and $u_1u_2=0$ unless bidders $1$ and $2$ both win. Hence,
    \begin{align*}
    \E_\ssb[u_1] &= \Prob_\ssb[\text{$1$ wins}] \, \E_\ssb[u_1 \mid \text{$1$ wins}] + \Prob_\ssb[\text{$1$ loses}] \, \E_\ssb[u_1 \mid \text{$1$ loses}]\\
    &= \Prob_\ssb[\text{$1$ wins}] \, \E_\ssb[u_1 \mid \text{$1$ wins}],\\
    \E_\ssb[u_1^2] &= \Prob_\ssb[\text{$1$ wins}] \, \E_\ssb[u_1^2 \mid \text{$1$ wins}],\\
    \E_\ssb[u_1u_2] &= \Prob_\ssb[\text{$1$ and $2$ win}] \, \E_\ssb[u_1u_2 \mid \text{$1$ and $2$ win}].
    \end{align*}
    By symmetry and because exactly $k$ out of $n$ bidders win in every realization,
    \begin{align*}
    \Prob_\ssb[\text{$1$ wins}] &= \frac{k}{n},\\
    \Prob_\ssb[\text{$1$ and $2$ win}] &= \frac{\binom{n-2}{k-2}}{\binom{n}{k}} = \frac{k(k-1)}{n(n-1)}.
    \end{align*}
    Rearranging yields the result.
\end{proof}

\begin{example}[Equal surplus conditional on winning]\label{example:payments-from-losers}
For $n=3$, $k=2$, uniformly distributed signals on $[0,1]$, and pure private values ($c=0$), assign the items efficiently. Charge a winning bidder $s_i-\frac{2}{3}$ and a losing bidder $L(s_i)=\frac{s_i(4-5s_i+s_i^2)}{3(1-s_i)^2}$, $0\leq s_i<1$. Every winner receives surplus $\frac{2}{3}$. The resulting interim utility is $U_i(s_i)=s_i^2-s_i^3/3$, which satisfies the envelope condition and $U_i(0)=0$. Thus, the mechanism is BNIC and interim IR. Losers have realized surplus $-L(s_i)\leq0$. A bidder with signal $s_i=1$ wins almost surely.
\end{example}

\clearpage
\endgroup

\begingroup
\emergencystretch=3em
\printbibliography
\endgroup

\end{document}